%% file: main.tex
\documentclass[final,3p,times]{elsarticle}

\usepackage{amssymb}

\usepackage{comment}
\input{addon/pkgloading}

\usepackage[T1]{fontenc}
\usepackage{lmodern}
\input{addon/macros}
\input{addon/utils}
\usepackage{orcidlink}
\crefformat{enumi}{#2\textup{#1}#3}

\journal{Journal of Information Security and Applications}

\input{addon/commands-for-revisions.tex}
\usepackage{tikz}
\newcommand{\tikzxmark}{%
\tikz[scale=0.23] {
    \draw[line width=0.7,line cap=round] (0,0) to [bend left=6] (1,1);
    \draw[line width=0.7,line cap=round] (0.2,0.95) to [bend right=3] (0.8,0.05);
}}
\newcommand{\tikzcmark}{%
\tikz[scale=0.23] {
    \draw[line width=0.7,line cap=round] (0.25,0) to [bend left=10] (1,1);
    \draw[line width=0.8,line cap=round] (0,0.35) to [bend right=1] (0.23,0);
}}
\begin{document}
\begin{frontmatter}

\title{A TEE-based Approach for Preserving Data Secrecy in Process Mining with Decentralized Sources}

\author[a1]{Davide~Basile\orcidlink{0000-0002-5804-4036}}
\author[a1]{Valerio Goretti\orcidlink{0000-0001-9714-4278}}
\author[a1]{Luca Barbaro\orcidlink{0000-0002-2975-5330}}
\author[a2]{Hajo A. Reijers\orcidlink{0000-0001-9634-5852}}
\author[a2]{Claudio Di Ciccio\orcidlink{0000-0001-5570-0475}}

\affiliation[a1]{organization={Sapienza University of Rome},
addressline={Via Regina Elena 295},
postcode={00161},
city={Rome},
country={Italy}}

\affiliation[a2]{organization={Utrecht University},
addressline={Princetonplein 5},
postcode={3584 CC},
city={Utrecht},
country={The Netherlands}}

\begin{abstract}
\input{content/abstract}

\end{abstract}




\end{frontmatter}



\section{Introduction}\label{sec:introduction}
\input{content/introduction}
\section{Background and Related Work}\label{sec:backgroundRelated}
\input{content/relatedwork}
\section{Motivating Scenario}\label{sec:motivating}
\input{content/motivating}

\section{Definitions and Formal Properties}\label{sec:formal}
\input{content/notation}
\section{Design}\label{sec:design}
\input{content/design.tex}
\section{Realization}\label{sec:realization}
\input{content/realization}
\section{System Guarantees}\label{sec:formalvalidation}
\input{content/formalvalidation}
\section{Empirical Evaluation}\label{sec:evaluation}
\input{content/evaluation}
\section{Discussion}\label{sec:discussion}
\input{content/discussion}
\section{Conclusion}\label{sec:conclusion}
\input{content/conclusion}

\input{content/ack}

\bibliographystyle{elsarticle-num-names} 
\bibliography{bibliography}

\end{document}

%% file: addon/pkgloading.tex
\usepackage[english]{babel}
\usepackage{subcaption}
\usepackage{xcolor}
\usepackage{url}
\usepackage{colortbl}
\usepackage{amssymb}
\usepackage{stmaryrd}
\usepackage{amsmath}
\usepackage{mathrsfs}
\usepackage{amssymb}
\usepackage{amsthm}
\newdefinition{definition}{Definition}[section]
\newtheorem{theorem}{Theorem}[section]
\newtheorem{assumption}{Postulate}[section]
\newtheorem{lemma}{Lemma}[section]
\newtheorem{proposition}{Proposition}[section]

\usepackage[left]{lineno}
\usepackage{xfrac}
\usepackage{nicefrac}
\usepackage[textsize=scriptsize,backgroundcolor=yellow!40]{todonotes}
\usepackage[nonumberlist,acronym,sanitize=none]{glossaries}
\glsdisablehyper
\usepackage{comment}
\usepackage[pdftex, colorlinks=true, hyperfootnotes=true, hyperindex=true,
            plainpages=false, pagebackref=false, pdfpagelabels=true, pdfstartview=FitH,
            linkcolor=blue, citecolor=blue, urlcolor=blue,
            bookmarks, bookmarksopen, bookmarksdepth=3]{hyperref}
\usepackage[capitalise,nameinlink]{cleveref}
\crefname{algocf}{Alg.}{Algs.}
\Crefname{algocf}{Alg.}{Algs.}
\crefname{definition}{Def.}{Defs.}
\Crefname{definition}{Def.}{Defs.}
\crefname{assumption}{Post.}{Posts.}
\Crefname{assumption}{Post.}{Posts.}
\crefname{theorem}{Thm.}{Thms.}
\Crefname{theorem}{Thm.}{Thms.}
\crefname{lemma}{Lem.}{Lemt.}
\Crefname{lemma}{Lem.}{Lemt.}
\crefname{corollary}{Cor.}{Cors.}
\Crefname{corollary}{Cor.}{Cors.}
\Crefname{proposition}{Prop.}{Props.}
\crefname{proposition}{Prop.}{Props.}
\crefname{aAlgocflLine}{ln.}{lns.}
\Crefname{algocfline}{ln.}{lns.}
\crefname{section}{Sect.}{Sects.} 
\Crefname{section}{Sect.}{Sects.} 
\crefname{table}{Tab.}{Tabs.} 
\Crefname{table}{Tab.}{Tabs.} 
\crefname{figure}{Fig.}{Figs.} 
\Crefname{figure}{Fig.}{Figs.} 
\usepackage{tkz-base}
\usepackage{tikz}
\usetikzlibrary{shapes.geometric}
\usetikzlibrary{decorations.pathmorphing,trees,snakes,arrows,shapes,automata,petri}
\usetikzlibrary{arrows.meta, positioning}
\usepackage{paralist}
\usepackage{multicol}
\usepackage{booktabs}
\usepackage{enumitem}
\usepackage{multirow}
\usepackage{rotating}
\usepackage{etaremune}
\usepackage[ruled,linesnumbered,algo2e]{algorithm2e}

\SetCommentSty{mycommfont} 
\SetArgSty{textrm} 
\SetFuncArgSty{textrm} 
\renewcommand{\SetProgSty}[1]{\renewcommand{\ProgSty}[1]{\textnormal{\csname#1\endcsname{##1}}\unskip}}
\SetProgSty{textrm} 
\usepackage{marginnote}
\usepackage{mathtools}
\usepackage{mathabx}
\usepackage{adjustbox}
\usepackage{ifthen}
\usepackage[normalem]{ulem}
\usepackage{lineno}
\usepackage{soul}
\usepackage{floatrow}
\usepackage{listings}

\crefname{lstlisting}{Listing}{Listings}
\usepackage{lipsum}
\usepackage{makecell}
\usepackage{diagbox}
\usepackage[scientific-notation=false,group-separator={,}]{siunitx}
\usepackage{microtype}
\usepackage{footmisc}
\usepackage{wrapfig}

\usepackage[ddmmyyyy]{datetime}
\usepackage{scrextend}
\usepackage{textcase}
\usepackage{setspace}

\usepackage{amsmath}
\usepackage{pifont}

%% file: addon/macros.tex
\definecolor{specializedcliniccolor}{RGB}{0, 153, 0}
\definecolor{pharmacolor}{RGB}{0, 110, 175}
\definecolor{scriptcolor}{RGB}{237, 237, 237}
\newcommand{\Actor}[1]{\textrm{\MakeTextLowercase{#1}}}
\newcommand{\Compo}[1]{\textnormal{\texttt{#1}}}
\newcommand{\Activ}[1]{\textrm{#1}}
\def\Org{\ensuremath{\mathscr{O}}}

\def\RCoefficent{\ensuremath{R^2}}
\newcommand{\RCoefficentF}[0]{\ensuremath{%
		 1-\frac{\sum_{i=1}^{k}(y_i- \hat{y_i})^2}{\sum_{i=1}^{k}(y_i- \overline{y_i})^2}}}
\def\Rlin{\ensuremath{R^2_\textrm{lin}}}
\def\Rlog{\ensuremath{R^2_\textrm{log}}}

\def\Slope{\ensuremath{\widehat{\mathcal{\beta}}}}
\newcommand{\SlopeF}[0]{\ensuremath{%
\frac{\sum_{i=1}^k\left(x_i-\bar{x}\right)\left(y_i-\bar{y}\right)}{\sum_{i=1}^k\left(x_i-\bar{x}\right)^2}}}

\def\OrgU {\ensuremath{\widehat{\mathcal{O}}}} 

\def\LPrv {\ensuremath{\mathcal{P}}}
\def\SecM {\ensuremath{\mathcal{M}}}   
\def\LPrvS {\ensuremath{\widehat{\mathcal{P}}}}  
\def\SecMS {\ensuremath{\widehat{\mathcal{M}}}}  
\def\CStor {\ensuremath{\textsl{Cases}}}  
\def\Case {\ensuremath{\textsl{C}}}  
\def\Evt {\ensuremath{e}}  
\def\EvtU {\ensuremath{\widehat{\mathrm{E}}}}  
\def\CasP {\ensuremath{\Case^{\LPrv}}}  
\def\CId {\ensuremath{\textrm{\textsl{iid}}}}  
\def\CIdF {\ensuremath{\mathfrak{iid}}}  
\def\CIdU {\ensuremath{\widehat{\mathrm{IID}}}}  
\def\CasU {\ensuremath{\widehat{\mathrm{C}}}}  
\def\CasF {\ensuremath{\mathcal{C}}} 
\def\CIdS {\ensuremath{\textsl{IIDs}}}  
\def\CIDMap {\ensuremath{\textsl{IIDMap}}} 
\def\LPrvMap {\ensuremath{\textsl{PMap}}} 
\def\Merge {\ensuremath{\bigoplus}}
\def\SegSize {\ensuremath{\textsl{seg\_size}}} 
\def\Rep {\ensuremath{\textsl{rep}}} 
\def\SegSet {\ensuremath{\textsl{SegSet}}} 
\def\PrevEvt {\ensuremath{\textsl{Prev}}} 
\def\PostEvt {\ensuremath{\textsl{Post}}} 
\def\Segm {\ensuremath{S}} 
\def\ActS {\ensuremath{\widehat{\mathrm{A}}}}  
\def\ActF {\ensuremath{\mathfrak{act}}}  
\def\TimeF {\ensuremath{\mathfrak{tms}}}  
\def\AuthPtP{\ensuremath{\textsl{al}}}
\def\EvtL{\ensuremath{\mathrm{L}}}
\newcommand{\Confine}{\textsc{Cnf}}

\def\Inst{\textbf{instance}}

\def\DoYieldCases {\ensuremath{\textsl{do\_yield\_cases}}}

\newcommand{\EqRel}[1]{\ensuremath{\overset{#1}{\sim}}}
\newcommand{\EqFun}[1]{\ensuremath{\overset{#1}{\to}}}
\newcommand{\EqCls}[1]{\ensuremath{\EvtU^{\sim}_{#1}}}

\DeclareMathOperator*{\argmax}{argmax\phantom{.}}
\DeclareMathOperator*{\argmin}{argmin\phantom{.}}

\def\PMParams{\ensuremath{\hbar_1,\ldots,\hbar_k}}
\def\PMFunc{\ensuremath{\underset{\PMParams}{\textit{prom}}}}

\newcommand{\Kpub}{\ensuremath{\mathcal{K}_\textrm{pub}}}
\newcommand{\Kpriv}{\ensuremath{\mathcal{K}_\textrm{priv}}}
\newcommand{\Ksym}{\ensuremath{\mathcal{K}_\textrm{sym}}}

\newcounter{asm}
\renewcommand{\theasm}{A\arabic{asm}}
\newcommand\Asm[0]{%
	\refstepcounter{asm}%
	\textbf{Assumption \theasm}}
\crefname{asm}{Asm.}{Asms.}
\Crefname{asm}{Asm.}{Asms.}

%% file: addon/utils.tex
\SetKw{Trigger}{trigger}
\SetKw{Yield}{yield}
\SetKw{Return}{return}
\SetKwProg{UponEv}{upon event}{ do}{}
\SetKwProg{Upon}{upon}{ do}{}
\SetKwInput{Implements}{Implements}
\SetKwInput{Uses}{Uses}
\SetKwInput{Hyperparameters}{Hyperparameters}
\SetKwProg{Fn}{Function}{ is}{}

\newcommand{\Mesg}[1]{\ensuremath{\textrm{\textsc{#1}}}}
\newcommand{\SuchThat}[1]{~\textrm{\textbf{such that}}~{#1}}
\newcommand{\To}[1]{~\textrm{\textbf{to}}~{#1}}
\newcommand{\Call}[1]{\textnormal{\texttt{#1}}}

\DeclareDocumentCommand{\crefalgln}{ o m }{%
	\IfNoValueF{#1}{\cref{#1}, }
	{\hyperref[#2]{ln.~\ref{#2}}}
}%


%% file: content/abstract.tex
Process mining techniques enable organizations to gain insights into their business processes through the analysis of execution records (event logs) stored by information systems. 
While most process mining efforts focus on intra-organizational scenarios, many real-world business processes span multiple independent organizations. 

Inter-organizational process mining, though, faces significant challenges, particularly regarding confidentiality guarantees: The analysis of data can reveal information that the participating organizations may not consent to disclose to one another, or to a third party hosting process mining services. 
To overcome this issue, this paper presents CONFINE, an approach for secrecy-preserving inter-organizational process mining. 
CONFINE leverages Trusted Execution Environments (TEEs) to deploy trusted applications that are capable of securely mining multi-party event logs while preserving data secrecy. We propose an architecture supporting a four-stage protocol to secure data exchange and processing, allowing for protected transfer and aggregation of unaltered process data across organizational boundaries. To avoid out-of-memory errors due to the limited capacity of TEEs, our protocol employs a segmentation-based strategy, whereby event logs are transmitted to TEEs in smaller batches.
We conduct a formal verification of our approach's correctness alongside a security analysis on the guarantees provided by the TEE core. 
We test our implementation using real-world and synthetic data to assess memory usage.
%
Our experiments show that an incremental approach to segment processing in discovery and conformance checking is preferable over non-incremental strategies as the former maintains memory usage trends within a narrow range at runtime, whereas the latter exhibit high peaks towards the end of the execution. Furthermore, our results confirm that our prototype can handle real-world workloads without out-of-memory failures. The scalability tests reveal that memory usage grows logarithmically as the event log size increases. Memory consumption grows linearly with the number of provisioning organizations, indicating potential scalability limitations and opportunities for further optimizations.

%% file: content/introduction.tex
In the management landscape, the notion of \emph{business process} (or process for short) conceptualizes and gives structure to the interplay of events, activities, and decision points to achieve organizations' operational objectives~\cite{Dumas.etal/2018:FundamentalsofBPM}.
With the rise of process automation, information systems' data have become a precious information source to enhance processes. 
The discovery, monitoring, and improvement of processes via the extraction of knowledge from chronological execution records is the focus of the research field of process mining~\cite{DBLP:journals/tmis/Aalst12}. Input data is typically stored in structures known as \textit{event logs}~\cite{DBLP:books/sp/22/WeerdtW22}, which collect multiple process instance execution traces (named \emph{cases}) as sequences of instantaneous records (named \emph{events}).

To date, the vast majority of process mining contributions consider \textit{intra-organizational} settings, in which business processes are executed within individual organizations. However, organizations recognize the value of collaboration and synergy in achieving operational excellence. \textit{Inter-organizational} business processes involve several independent organizations cooperating to achieve a shared objective~\cite{van2011intra}. 
Through inter-organizational process mining, organizations can gain a deeper understanding of their shared business processes, which would be missed by only considering fragmented views of distinct executions~\cite{DBLP:journals/bise/OberdorfSWSMF23}. Such a comprehensive perspective supports the identification of bottlenecks entailed by the cross-organizational interactions, providing qualitative and quantitative insights into the effectiveness of the collaboration~\cite{DBLP:journals/peerj-cs/Morales-Sandoval21}. In the context of supply-chain management, this broader view enhances the analysis of how goods, information, and decisions circulate across multiple organizations, thereby supporting the optimization of logistic workflows, reducing operational costs, and improving delivery performance~\cite{ho2020supply}.
Additionally, inter-organizational process mining can reveal behavioral dependencies and structural mismatches arising from erroneous process enactment practices, thus being instrumental to root-cause analysis of failures. In the manufacturing industry, for example, it can identify the causes of recurrent production anomalies by uncovering incorrectly modeled dependencies between activities carried out across different organizational units~\cite{DBLP:journals/isf/TanWC16}.
Inter-organizational process mining also provides a shared ground for jointly verifying conformance with standards that span multiple organizational contexts and regulate cooperation~\cite{DBLP:journals/is/FdhilaKRR22}. This becomes particularly valuable in domains such as healthcare, where stringent regulations oversee the administration of patient pathways that involve multiple autonomous institutions~\cite{DBLP:journals/tase/LiuLZCZ23,DBLP:journals/csse/JeonKSK23}.

Despite the advantages of transparency, performance optimization, and benchmarking that companies can gain, inter-organizational process mining raises challenges that hinder its application. The major issue concerns secrecy. Companies are reluctant to share private information, as required to execute process mining algorithms with their partners~\cite{Liu.etal/ISF2009:ChallengesOpportunitiesCollaborativeBPM}. Indeed, letting sensitive operational data traverse organizational boundaries introduces concerns about data privacy, security, and compliance with internal regulations~\cite{muller2021trust}. To address this issue, the majority of research endeavors have focused thus far on the alteration of input data or of intermediate analysis by-products, with the aim to impede the reconstruction of the original information sources from the counterparts ~\cite{DBLP:journals/dke/FahrenkrogPetersenAW23,DBLP:conf/icpm/Fahrenkrog-Petersen19,elkoumy2020shareprom,elkoumy2020secure}. These preemptive solutions have the merit of neutralizing information leakage by malicious parties \textit{a priori}. Nevertheless, they entail an ex-ante information loss, thus compromising downstream process mining accuracy~\cite{DBLP:journals/dke/FahrenkrogPetersenAW23,DBLP:conf/icpm/Fahrenkrog-Petersen19}, or require the use of intermediary log representations, which restricts the range of applicable algorithms~\cite{elkoumy2020shareprom,elkoumy2020secure}.

To improve over the state of the art, we propose CONFINE, a novel approach and toolset aimed at enhancing collaborative information system architectures with secrecy-preserving process mining capabilities. 
%
Inspired by the vision of decentralized data ownership in Solid~\cite{DBLP:conf/www/MansourSHZCGAB16,DBLP:conf/icdcs/BasileCGK23}, our solution decentralizes data provenance across multiple autonomous provisioning organizations, each retaining ownership of their (partitions of the) event logs. Individual logs embody the organizations' own partial view on the shared inter-organizational process.

To secure information secrecy during the exchange and elaboration of data, our solution resorts to \emph{Trusted Execution Environments} (TEEs)~\cite{DBLP:conf/trustcom/SabtAB15}. These are hardware-secured contexts, which guarantee code integrity and data confidentiality before, during, and after their use. Owing to these characteristics, CONFINE lets information be securely transferred beyond an organization's borders. Therefore, computing nodes other than the information provisioners can aggregate and elaborate the original, unaltered process data in a secure, externally inaccessible vault.


%

The architecture of CONFINE supports a four-staged protocol:
\begin{inparaenum}[\itshape(i)\upshape]
	\item The initial exchange of preliminary metadata,
	\item the attestation of the mining entities,
	\item the protected transmission and secrecy-preserving merge of encrypted information segments among multiple organizations, \label{item:phase:datatransmission}
	\item the computation of process discovery and conformance checking algorithms on joined data in a hardware-secured enclave.
\end{inparaenum}
We evaluate CONFINE proving the soundness and completeness of our protocol in its most crucial phase of information transmission and merge \itshape(iii)\upshape. Also, we discuss the security properties brought by the integration of TEEs at the core of our approach against a set of security threats.


Since TEEs operate with limited sections of the main memory shielded from access by external entities (operating system included), exceeding these limits may expose the system to potential data leakage. Therefore, we endow our experiments with an analysis of memory usage.
Given its direct impact on data secrecy, we analyze how memory usage scales by increasing the number of 
\begin{inparaenum}[\itshape(i)\upshape]
	\item processed events,
	\item processed cases, and
	\item participating organizations.
\end{inparaenum}
The results indicate that memory consumption exhibit a logarithmic increase as the volume and length of process instances grow, thereby reducing the likelihood of out-of-memory conditions even with large amounts of processed data. However, when increasing the number of participating organizations, the experiments show a linear growth of memory consumption, thereby delineating an avenue for future optimization work.

In this article, we extend and enhance our earlier work \cite{DBLP:conf/caise/GorettiBBC24,DBLP:conf/icpm/Goretti0BC24} with the following main contributions:
\begin{inparaenum}[\itshape(i)\upshape]
     \item we formally define the core concepts underpinning our approach (\cref{sec:formal});
    \item we endow the description of our CONFINE protocol with a formal specification with pseudocode, and provide an in-depth exploration of the main components' behavior (\cref{alg:secm,alg:merge,alg:lprv,alg:segment} in \cref{sec:realization});
    \item we discuss the security guarantees provided by the TEE to frame and motivate its role within our solution (\cref{sec:discussion:security});
    \item %
    we demonstrate the soundness and completeness of the data-transmission and merge phases within the 
    CONFINE protocol (\cref{sec:correctness}); 
    \item we extend the performance assessment on memory usage by introducing novel scalability experiments (\cref{sec:evaluation:subsec:MemoryUsage,sec:eval:segmentsize}); 
    \item we discuss the practical implications inherent to a real-world deployment of CONFINE (\cref{sec:discussion:practicals}).
\end{inparaenum}


The remainder of this paper is as follows. \Cref{sec:backgroundRelated} provides an overview of the background and the related work. In \cref{sec:motivating}, we introduce a motivating use-case scenario in healthcare. In \cref{sec:formal}, we propose the formal notations characterizing our work. \Cref{sec:design} presents the functional viewpoint of CONFINE and \cref{sec:realization} describes the realization of CONFINE alongside its four-stage protocol. 
In the latter sections, we highlight the assumptions under which our solution operates as the discourse unfolds.
In \cref{sec:formalvalidation}, we report on formal validation of our approach's correctness and security. \Cref{sec:evaluation} presents an empirical validation of our implementation with tests on memory usage and scalability. In \cref{sec:discussion}, we discuss the practical implications and limitations, and outline future research directions
beginning with a revisit of the assumptions mentioned in \cref{sec:design,sec:realization}.
We conclude our work in \cref{sec:conclusion}.

%% file: content/relatedwork.tex
In what follows, we provide an overview of the background literature and relevant existing research, laying the foundation for our study.
\subsection{Background}
Here we introduce the theoretical concepts pertinent to the core areas of our work, i.e., process mining and trusted execution environments. 

\subsubsection{Process Mining}
\label{sec:background:processmining}
Process mining refers to a suite of techniques designed to uncover, monitor, and refine processes by analyzing event logs to offer analytical insights into business operations \cite{van2012process}. Events within an event log typically possess attributes providing descriptive details such as the executed task, its timestamp, and the involved human resources. Event logs organize events related to the same process instances into \emph{cases}. The process mining literature delineates three primary families of techniques: \emph{discovery}, \emph{conformance checking}, and \emph{enhancement} algorithms~\cite{DBLP:journals/tmis/Aalst12}. Discovery techniques generate a process model (i.e., a depiction of the process) from an event log without relying on any initial information~\cite{weijters2006process}. 
The \emph{HeuristicsMiner} is a seminal process discovery algorithm~\cite{weijters2006process}. It constructs a dependency graph by computing the frequency of directly subsequent occurrences between activities within the same cases and deriving dependency measures capturing their causal relations. It then filters relations below a significance threshold to remove noise, and employs the remaining ones to build a weighted graph. This graph is finally transformed into a \emph{causal net}, 
a process model supporting a basic set of core workflow patterns including concurrency and (short) loops~\cite{Aalst.etal/2003:WorkflowPatterns}.

Conformance checking algorithms verify whether a given event log aligns with an input process model and vice versa~\cite{vanderAalst2016}.
In this paper, we consider the \emph{DeclareConformance} algorithm as a noticeable example~\cite{DBLP:conf/bpm/DonadelloRMS22}. It verifies the compliance of an event log with a declarative process specification, i.e., a set of behavioral constraints that capture dependencies among activities. The algorithm first projects the event log into activity sequences on a per-trace basis and then systematically checks each constraint of the model against these sequences. Detected violations are used to compute a fitness score for each trace, which the algorithm subsequently aggregates into a global score representing the overall adherence of the event log to the declarative specification.%
Enhancement techniques extend or refine an existing process model using information gathered from the actual process recorded in an event log~\cite{DBLP:journals/tmis/Aalst12, DBLP:conf/simpda/YasminBS18}.

Conventionally, companies apply process mining techniques to analyze their internal business processes within \emph{intra-organizational} settings. 
However, the occurrence of scenarios where business processes span across multiple independent organizations, with each organization maintaining its own event log, is becoming increasingly common~\cite{Liu.etal/ISF2009:ChallengesOpportunitiesCollaborativeBPM}.
We refer to this application of process mining as \emph{inter-organizational process mining}.  
Inter-organizational process mining classifies interactions into two basic settings: \emph{collaboration} and \emph{exploiting commonality}. In collaboration, business processes are distributed among collaborators, who cooperate to achieve a shared objective. 
In the commonality exploitation setting, organizations share common processes and may compete, learn from each other, or exchange experiences, knowledge, and infrastructure. Hence, the involved organizations need not collaborate as each of them carries out the same tasks independently. 
In both inter-organizational scenarios, multiple data sources are involved, posing significant challenges in data aggregation and subsequent algorithm application~\cite{van2011intra}.

In this work, we concentrate on the privacy-preserving aspects of collaborative scenarios, recognizing organizations' need to share event logs across organizational boundaries while safeguarding data confidentiality. To tackle this challenge, our solution draws on the theoretical foundations of confidential computing, whose core concepts are elucidated as follows.

\subsubsection{Confidential Computing and Trusted Execution Environments}
\label{sec:background:tee}
Confidential computing encompasses a suite of techniques aimed at safeguarding data while it is in use. These methods leverage hardware-backed protection of highly sensitive computations within attested \emph{Trusted Execution Environments} (TEEs). According to the Confidential Computing Consortium (CCC)~\footnote{\href{https://confidentialcomputing.io}{\nolinkurl{confidentialcomputing.io}}. Accessed: December 2, 2025.}, TEEs are secure execution contexts that provide assurance regarding three key properties:
\begin{inparaenum}[]
\item (\textit{data confidentiality}) unauthorized entities cannot view data while it is in use within the TEE, \item  (\textit{data integrity}) unauthorized entities cannot add, remove, or alter data while it is in use within the TEE, and \item (\textit{code integrity}) unauthorized entities cannot add, remove, or alter code executing in the TEE~\cite{DBLP:conf/trustcom/SabtAB15}.
\end{inparaenum}
%
TEE technologies are diverse, and their features may differ depending on the specific implementation of the underlying Central Processing Unit (CPU). 

TEEs differ primarily in the extension of their trusted surface, that is, the hardware-enforced boundary defining the extent of isolation. Two main classes can be distinguished based on this criterion~\cite{DBLP:journals/ieeesp/JauernigSS20}. The first includes application-level (or process-level) TEEs, which protect a restricted portion of an application’s address space. In such architectures, implemented among others by Intel SGX%
\footnote{\href{https://www.intel.com/content/www/us/en/developer/tools/software-guard-extensions/overview.html}{\nolinkurl{software-guard-extensions/overview}}. Accessed: December 2, 2025.}
and ARM TrustZone%
\footnote{\href{https://www.arm.com/technologies/trustzone-for-cortex-a}{\nolinkurl{arm.com/technologies/trustzone}}. Accessed: December 2, 2025.}, only specific regions of memory explicitly designated as trusted are protected with confidentiality and integrity guarantees, whereas the remaining system components are treated as untrusted~\cite{antonino2021guardian}. The second class comprises virtual-machine level TEEs, such as AMD SEV%
\footnote{\href{https://www.amd.com/en/developer/sev.html}{\nolinkurl{amd.com/sev}}. Accessed: December 2, 2025.}
and Intel TDX,%
\footnote{\href{https://www.intel.com/content/www/us/en/developer/tools/trust-domain-extensions/overview.html}{\nolinkurl{intel.com/trust-domain-extensions/overview}}. Accessed: December 2, 2025.} which extend the isolation boundary to an entire guest virtual machine. These TEEs protect memory and CPU state belonging to the virtual machine from inspection or tampering by (even higher-privileged) external software~\cite{mei2024cabin}. 

Another distinguishing characteristic among TEE technologies lies in whether the memory regions allocated to the isolated environment are protected through hardware-based encryption. In TEE technologies such as Intel SGX, AMD SEV, and Intel TDX, the CPU encrypts and protects the integrity of data within the trusted domain. This feature relies on dedicated sections of a CPU capable of handling encrypted data within a reserved section of the main memory~\cite{costan2016intel}. This protected memory section is encrypted by the CPU using a random key derived at every power cycle. In contrast, technologies like ARM TrustZone mainly rely on access control layers between trusted and untrusted execution modes, with no hardware-based encryption involved.

A further essential feature characterizing certain TEE technologies, including Intel SGX, Intel TDX, and AMD SEV, is \emph{attestability}. Through this feature, the CPU can produce a cryptographic proof, typically referred to as \emph{attestation evidence}, to demonstrate the trusted status of the TEE to a local or remote verifier.
Following the Remote ATtestation procedureS (RATS) standard~\cite{birkholz2023rfc}, remote attestation protocols generally consist of the following steps:
\begin{inparaenum}[(1)]
\item The \emph{verifier} establishes a communication channel with the \emph{attester} (which may correspond to a trusted application or a virtual machine) and requests the generation of the attestation evidence.
\item The attester then instructs the CPU to produce the attestation evidence, which serves as cryptographic proof of its trusted state. In most protocols, the attestation evidence takes the form of a digital signature derived from hardware-bound secrets and includes a \emph{measurement}, namely a cryptographic hash representing the state of the virtual machine or the trusted application’s code. The attestation evidence is signed using a hardware-protected attestation key embedded in the processor. Depending on the TEE technology, the verifier may rely on the CPU manufacturer to validate this signature using the corresponding public endorsement credentials and to confirm the integrity of the reported measurement..
\item Upon receiving the attestation evidence, the verifier validates its authenticity by relying on an \emph{endorser} (typically the TEE manufacturer) responsible for providing the necessary public credentials and reference values. The verifier subsequently compares the reported measurement against one or more reference values to reach a trust decision. The successful completion of this process attests to the integrity and authenticity of the entity providing the attestation evidence.
\end{inparaenum}

Attestable TEE implementations enable trusted virtual machines and applications to embed custom data into the attestation evidence at the time of its generation (see step 2 above). This feature allows verifiers to apply tailored appraisal policies by evaluating the included information. Furthermore, this custom data field can serve to authenticate the trusted origin of cryptographic material, which may subsequently be leveraged to establish secure communication channels once the attestation evidence has been successfully verified~\cite{helble2021flexible}. While CPU manufacturers such as Intel and AMD provide web services supplying the validation material required to verify attestation evidence generated by their TEE-enabled processors, they delegate to system designers the responsibility of defining appropriate trust-decision policies and establishing secure communication channels~\cite{DBLP:conf/dais/MenetreyGKPFSR22}.

In this work, we consider application-level TEEs that provide hardware-backed memory encryption and support remote attestation. Accordingly, our contribution implicitly assumes TEEs endowed with these three essential characteristics.

\subsection{Related Work}
\label{sec:relatedWork}
Despite the recency of this research branch across process mining and collaborative information systems, scientific literature already includes noticeable contributions to inter-organizational process mining.
Before discussing specific contributions, we categorize existing secrecy preservation methodologies into three main classes, depending on whether they are based on 
\begin{inparaenum}[(I)]
    \item event log sanitization, 
    \item Secure Multiparty Computation (SMPC), and 
    \item TEE.
\end{inparaenum}
\begin{table}[b]
	\centering
	\caption{Comparison of secrecy-preserving approaches for process mining}
	\label{tab:comparisonTable}
	\resizebox{\textwidth}{!}{
		\begin{tabular}{@{}lcccc@{}}
			\toprule
			\textbf{Approach} & \textbf{Zero noise} & \textbf{Algorithm compatibility} &\textbf{Hardware independence} &\textbf{Computational control}\\ 
			\midrule
			Sanitization-based & \tikzxmark & \tikzcmark &\tikzcmark &\tikzxmark\\
			SMPC-based         & \tikzxmark  & \tikzxmark &\tikzcmark &\tikzcmark\\
			TEE-based          & \tikzcmark & \tikzcmark & \tikzxmark &\tikzcmark\\
			\bottomrule
		\end{tabular}
	}
\end{table}
In \cref{tab:comparisonTable}, we summarize and compare the pros and cons of each approach.
\begin{inparaenum}[(I)]
	\item {\itshape Sanitization-based approaches} apply obfuscation mechanisms to transform event logs prior to their analysis. The main advantage of these approaches lies in preserving privacy while maintaining compatibility with existing process mining algorithms, as the sanitized output remains a full-fledged event log. However, they introduce noise that may distort analytical results, hindering case-level mining capabilities, and exert no control over the subsequent process mining computation.
	\item {\itshape SMPC-based approaches} employ cryptographic protocols that allow multiple organizations to jointly execute verifiable computations over encrypted data without revealing their original event logs to each other. These approaches ensure a high degree of control over the computation itself. Nevertheless, they typically limit the range of process mining algorithms supported and may introduce noise in the final results.
	\item {\itshape TEE-based approaches} leverage hardware-backed cryptographical methodologies to enable the secure execution of process mining algorithms within isolated environments. This paradigm achieves zero-noise processing (as data remains encrypted even during processing, avoiding distortions), control over the computation, and seamless compatibility with existing process mining algorithms, as TEEs allow the execution of Turing-complete code either as-is or with only minimal adaptations. However, unlike sanitization and SMPC-based approaches, these solutions depend on specific hardware platforms, thereby imposing constraints on the underlying technology.
\end{inparaenum}

Next, we discuss existing research contributions, grouped according to the three aforementioned categories.

\begin{asparaenum}[(I)]
\item \textbf{Sanitization-based approaches.} Fahrenkrog-Petersen et al.~\cite{DBLP:journals/dke/FahrenkrogPetersenAW23,DBLP:conf/icpm/Fahrenkrog-Petersen19} introduce the PRETSA algorithms family: a set of event log sanitization techniques that perform step-wise transformations of prefix-tree event log representation into a sanitized output ensuring \emph{k-anonimization} and \emph{t-closeness}. Batista et al.~\cite{DBLP:journals/istr/BatistaMS22} present the k-PPPM algorithm, designed to ensure the protection of identities within an input event log. This objective is achieved through the clustering of the sub-logs corresponding to each individual and the subsequent approximation by selecting representative centroids. While these two approaches effectively minimize information loss, they introduce approximations within the original event log, which may compromise the exactness of process mining results or inhibit mining tasks. In contrast with the latter two techniques, our research proposes a secure architecture wherein secure computational vaults collect event logs devoid of upstream alterations and protect them at runtime, thus generating results derived directly from the original information source. 

\item  \textbf{SMPC-based approaches.} Elkoumy et al.~\cite{elkoumy2020shareprom,elkoumy2020secure} present Shareprom, an SMPC-based approach. Akin to our work, their solution offers a means for independent entities to execute process mining algorithms in inter-organizational settings while safeguarding their proprietary input data from exposure to external organizations operating within the same context. Shareprom's functionality, though, is confined to the execution of operations involving event log abstractions~\cite{FederatedPM2021} represented as directed acyclic graphs, which the organizations employ as intermediate pre-elaboration to be fed into SMPC routines~\cite{SMPC2015}. The final outputs are then modified through the addition of Laplacian noise. This approach empowers data owners to retain a degree of control over the process mining techniques applied to their event logs. However, as the authors note, the reliance on the graph representation introduces constraints that may be limiting in various process mining scenarios. In particular, the resulting graph abstraction is compatible with only a subset of discovery and conformance checking algorithms. In contrast, our approach supports the encrypted transmission and analysis of event logs without relying on any intermediate abstraction. This enables the faithful reproduction of existing process mining techniques within computationally isolated environments, thereby providing greater flexibility.   
%
%
%
%
\item \textbf{TEE-based approaches.} The seminal work of M{\"u}ller et al.~\cite{muller2021process} is the first to envision the adoption of TEEs within the process mining domain. In their paper, the authors propose a conceptual architecture in which each participating organization encrypts its local event log, publishes it to a public repository, and delegates the mining computation to a remote trusted application executed inside a TEE. The interaction between miners and the public repository is coordinated through a blockchain-based control plane, while a Secret Management System (SMS) running in a separate TEE is responsible for storing the decryption keys and releasing them to the miner after successful attestation.
While this work represents a foundational contribution and an essential source of inspiration for our research, it leaves several challenges unaddressed that our work explicitly aims to address.

\end{asparaenum}

A key distinction from the work of M{\"u}ller et al.\ concerns the event log provisioning protocol, which constitutes a key contributions of our work. In~\cite{muller2021process}, a specific provisioning strategy was not designed, and event logs are transmitted and processed as single and monolithic datasets. This can lead to confidentiality breaches as large event logs may exceed the limited memory available within TEEs and result in crashes or untrusted memory allocations~\cite{DBLP:conf/ecrts/YuhalaGMSKF25,DBLP:journals/pvldb/SunWL021}. To overcome these issues, our research introduces a segmentation-based provisioning protocol in which event logs are partitioned into smaller batches that are individually transmitted and merged within the TEE. Building on this segmentation mechanism, we further propose an incremental processing model that allows the miner running within the TEE to update a partial result as new segments arrive from different organizations during transmission. Therefore, the process mining output converges to the same result as a full merged log computation, while avoiding the risks associated with memory peaks that may occur with monolithic processing. 
From an architectural standpoint, our approach reduces the attack surface inherent in that seminal work, which arises from the inclusion of utility components operating outside organizational boundaries. Although designed to store encrypted data, the event log repository proposed in the research of M{\"u}ller et al.\ remains vulnerable to malicious inspection or cryptanalysis attacks due to its public nature. Similarly, the SMS running within the TEE may entail vulnerabilities. Since TEEs are inherently non-persistent environments, any data residing in their volatile memory is lost once the TEE's machine is turned off. Consequently, the previous approach would require additional mechanisms to anchor the cryptographic material in persistent memory outside the TEE. However, this practice inevitably extends the trusted boundary and may introduce rollback attacks~\cite{DBLP:conf/dsn/BrandenburgerCL17}, in which an adversary with control over the host system can revert a sealed state to a previous version, thereby re-enabling outdated or revoked authorizations. To overcome these limitations, we propose a peer-to-peer model in which miners operating within a TEE retrieve event logs directly from native organizations, which, in turn, act as decentralized data sources~\cite{DBLP:conf/www/MansourSHZCGAB16} that independently handle the encrypted data transmission. Additionally, within every execution of our protocol, we generate ephemeral encryption keys on demand and use them solely for the current transmission session, without requiring any permanent external key storage.
 

%% file: content/motivating.tex
\def\Talice{\Activ{PH, COPA, OD, \color{pharmacolor}DOR, PDL, SD, \color{black}RD, AD, TP, \color{specializedcliniccolor}PAFH, PIA, PT, VRT, TPB, \color{black}RPB, DPH, PCD, DP}}
\def\Tbob{\Activ{PH, COPA, OD, \color{pharmacolor}DOR, PDL, SD, \color{black}RD, AD, PRTA, PCD, DPH, DP}}
\def\TaliceUncolored{\Activ{PH, COPA, OD, DOR, PDL, SD, RD, AD, TP, PAFH, PIA, PT, VRT, TPB, RPB, DPH, PCD, DP}}
\def\TbobUncolored{\Activ{PH, COPA, OD, DOR, PDL, SD, RD, AD, PRTA, PCD, DPH, DP}}
\begin{figure}[t]
\centering
\includegraphics[width=\linewidth]{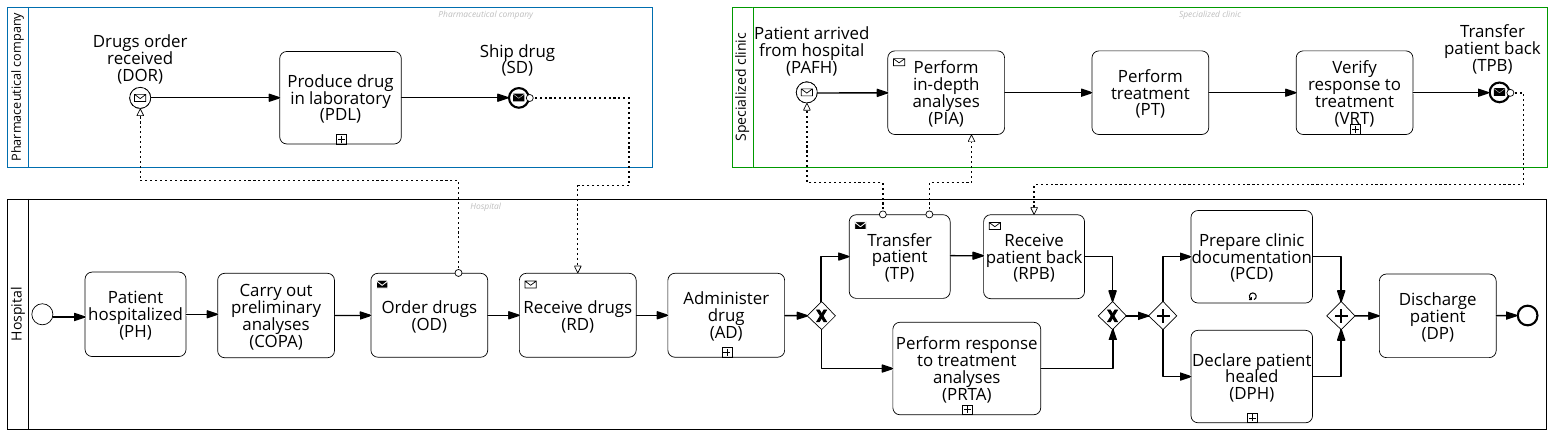}
\caption{A BPMN collaboration diagram of a simplified healthcare scenario}
\label{fig:BPMN_Healthcare}
\end{figure}
\begin{table}[b]
  \caption[Event log]{Events from cases 312 (Alice) and 711 (Bob) recorded by the \Actor{Hospital}, the \Actor{Specialized clinic}, and the \Actor{Pharmaceutical company}}
 \centering
    \begin{minipage}[c]{0.55\linewidth}
       \resizebox{0.99\textwidth}{!}{
        \begin{tabular}{ll|l|l|l|l|l|l|}
        
        & \multicolumn{7}{c}{\large{\textbf{Hospital}}\vspace{0.2em}}                                               \\ \cline{2-4}\cline{6-8}
        \multicolumn{1}{c|}{} & {Case} & {Timestamp} & Activity &  & {Case} & {Timestamp} & Activity \\ \cline{2-4} \cline{6-8}
        \multicolumn{1}{c|}{$\Evt_1$} & {312}         & {2022-07-14T10:36}           & PH    & $\Evt_{11}$ & {312}         & {2022-07-15T22:06}           & TP     \\ \cline{2-4} \cline{6-8}
        \multicolumn{1}{c|}{$\Evt_2$} & {312}         & {2022-07-14T16:36}           & COPA  & $\Evt_{12}$ &  {711}         & {2022-07-16T00:55}           & PRTA          \\ \cline{2-4} \cline{6-8} \cline{6-8}
        \multicolumn{1}{c|}{$\Evt_3$} & {711}         & {2022-07-14T17:21}           & PH    & $\Evt_{13}$ &  {711}         & {2022-07-16T01:55}           & PCD            \\ \cline{2-4} \cline{6-8}
        \multicolumn{1}{c|}{$\Evt_4$} & {312}         & {2022-07-14T17:36}           & OD    & $\Evt_{14}$ &  {711}         & {2022-07-16T02:55}           & DPH           \\ \cline{2-4} \cline{6-8}
        \multicolumn{1}{c|}{$\Evt_5$} & {711}         & {2022-07-14T23:21}           & COPA  & $\Evt_{15}$ &  {711}         & {2022-07-16T04:55}           & DP          \\ \cline{2-4} \cline{6-8}
        \multicolumn{1}{c|}{$\Evt_6$} & {711}         & {2022-07-15T00:21}           & OD    & $\Evt_{16}$ &   {312}         & {2022-07-16T07:06}           & RPB        \\ \cline{2-4} \cline{6-8}
        \multicolumn{1}{c|}{$\Evt_7$} & {711}         & {2022-07-15T18:55}           & RD    & $\Evt_{17}$ &   {312}         & {2022-07-16T09:06}           & DPH          \\ \cline{2-4} \cline{6-8}
        \multicolumn{1}{c|}{$\Evt_8$} & {312}         & {2022-07-15T19:06}           & RD    & $\Evt_{18}$ &   {312}         & {2022-07-16T10:06}           & PCD         \\ \cline{2-4} \cline{6-8}
        \multicolumn{1}{c|}{$\Evt_9$} & {711}         & {2022-07-15T20:55}           & AD    & $\Evt_{19}$ &   {312}        & {2022-07-16T11:06}           & DP            \\ \cline{2-4} \cline{6-8}
        \multicolumn{1}{c|}{$\Evt_{10}$} & {312} & {2022-07-15T21:06}           & AD \\ \cline{2-4}
        \end{tabular}%
        \label{tab:partitions}
        }
    \end{minipage}%
    \hspace{0.004\textwidth}%
    \begin{minipage}[c]{0.33\linewidth}
        \resizebox{\textwidth}{!}{
            \raisebox{2\baselineskip}{%
                \begin{tabular}{ll!{\color{pharmacolor}\vline}l!{\color{pharmacolor}\vline}l!{\color{pharmacolor}\vline}}\arrayrulecolor{pharmacolor}
                & \multicolumn{3}{c}{\large{\textbf{Pharmaceutical company}}\vspace{0.2em}}                                               \\ \cline{2-4}
                \multicolumn{1}{c|}{}&{HospitalCaseID} & {Timestamp} & Activity \\ \cline{2-4}
                \multicolumn{1}{c|}{$\Evt_{20}$} &312          & {2022-07-15T09:06}           & DOR \\ \cline{2-4}
                \multicolumn{1}{c|}{$\Evt_{21}$} &711         & {2022-07-15T09:30}           & DOR \\ \cline{2-4}
                \multicolumn{1}{c|}{$\Evt_{22}$} &312         & {2022-07-15T11:06}           & PDL \\ \cline{2-4}
                \multicolumn{1}{c|}{$\Evt_{23}$} &711         & {2022-07-15T11:30}           & PDL \\ \cline{2-4}
                \multicolumn{1}{c|}{$\Evt_{24}$} &312         & {2022-07-15T13:06}           & SD  \\ \cline{2-4}
                \multicolumn{1}{c|}{$\Evt_{25}$} &711          & {2022-07-15T13:30}           & SD  \\ \cline{2-4}
                \end{tabular}%
            }
        }
        \\[0.3em]
        \resizebox{\textwidth}{!}{
            \raisebox{2.5\baselineskip}{%
                \begin{tabular}{ll!{\color{specializedcliniccolor}\vline}l!{\color{specializedcliniccolor}\vline}l!{\color{specializedcliniccolor}\vline}}\arrayrulecolor{specializedcliniccolor}
        
                & \multicolumn{3}{c}{\large{\textbf{Specialized clinic}}\vspace{0.2em}}                                               \\\cline{2-4}
                \multicolumn{1}{c|}{} & {TreatmentID}  & {Timestamp} & Activity \\ \cline{2-4}
                \multicolumn{1}{c|}{$\Evt_{26}$}&312         & {2022-07-16T00:06}           & PAFH        \\ \cline{2-4}
                \multicolumn{1}{c|}{$\Evt_{27}$}&312          & {2022-07-16T01:06}           & PIA         \\ \cline{2-4}
                \multicolumn{1}{c|}{$\Evt_{28}$}&312          & {2022-07-16T03:06}           &  PT         \\ \cline{2-4}
                \multicolumn{1}{c|}{$\Evt_{29}$}&312          & {2022-07-16T04:06}           & VRT         \\ \cline{2-4}
                \multicolumn{1}{c|}{$\Evt_{30}$}&312          & {2022-07-16T05:06}           & TPB         \\ \cline{2-4}
                \end{tabular}%
            }
        }
    \end{minipage}
    \\[0.3em]
    \begin{minipage}[t]{\linewidth}
            \resizebox{\textwidth}{!}{%
                \begin{tabular}{p{18cm}}
                $T_{312}=${\textlangle}  \Talice \,{\textrangle}\\[0.5\baselineskip]
                $T_{711}=${\textlangle}  \Tbob \,{\textrangle}\\
                \end{tabular}%
                     \label{tab:trace}
            }
\end{minipage}
\end{table}

To provide a running example and motivating scenario for our investigation, we propose a simplified hospitalization process for the treatment of rare diseases. 
Our running example draws inspiration from two notable contributions in the domain of healthcare processes: the work of Aguilar et al.~\cite {DBLP:conf/biostec/AguilarCRPCGM08}, which illustrates modeling practices for representing healthcare workflows, and the work of Sampson et al.~\cite{sampson2015process}, which emphasizes coordination principles among independent healthcare institutions.
The process model of our running example is represented as a Business Process Model and Notation (BPMN) diagram in \cref{fig:BPMN_Healthcare}. BPMN is among the most widely adopted standards for process modeling, describing processes in terms of interacting activities, events, and control-flow constructs like decision and concurrent execution gateways.%
\footnote{\url{https://www.omg.org/spec/BPMN/2.0}. Accessed: December 2, 2025.}

Our scenario involves three cooperating organizations: a \Actor{Hospital}, a \Actor{Pharmaceutical company}, and a \Actor{Specialized clinic}. For the sake of simplicity, we describe the process through two cases, recorded by the information systems as in \cref{tab:partitions}. 
Each patient in the hospital is associated with an id which would be the identifier of the case in the hospital log.
Alice's journey 
begins when she enters the hospital for the preliminary examinations (patient hospitalized, \Activ{PH}). The \Actor{Hospital} then places an order for the drugs (\Activ{OD}) to the \Actor{Pharmaceutical company} for treating Alice's specific condition. Afterward, the \Actor{Pharmaceutical company} acknowledges that the drugs order is received (\Activ{DOR}), proceeds to produce the drugs in the laboratory (\Activ{PDL}), and ships the drugs (\Activ{SD}) back to the \Actor{Hospital}. Upon receiving the medications, the \Actor{Hospital} administers the drug (\Activ{AD}), and conducts an assessment to determine if Alice can be treated internally. If specialized care is required, the \Actor{Hospital} transfers the patient (\Activ{TP}) to the \Actor{Specialized clinic}. When the patient arrives from the hospital (\Activ{PAFH}), the \Actor{Specialized clinic} performs in-depth analyses (\Activ{PIA}) and proceeds with the treatment (\Activ{PT}). Once the \Actor{Specialized clinic} had completed the evaluations and verified the response to the treatment (\Activ{VRT}), it transfers the patient back (\Activ{TPB}). The \Actor{Hospital} receives the patient back \Activ(RPB) and prepares the clinical documentation (\Activ{PCD}). If Alice has successfully recovered, the \Actor{Hospital} declares the patient as healed (\Activ{DPH}). When Alice's treatment is complete, the \Actor{Hospital} discharges the patient (\Activ{DP}). 
Bob 
enters the \Actor{Hospital} a few hours later than Alice. His hospitalization process is similar to Alice's. However, he does not need specialized care, and his case is only treated by the \Actor{Hospital}. Therefore, the \Actor{Hospital} performs the response to treatment analyses (\Activ{PRTA}) instead of transferring him to the \Actor{Specialized clinic}. 

Both the National Institute of Statistics of the country in which the three organizations reside and the University that hosts the hospital wish to uncover information on this inter-organizational process for reporting and auditing purposes~\citep{Jans.Hosseinpour/IJAIS2019:ActiveLearningProcessMiningForAuditing} via process analytics. The involved organizations share the urge for such an analysis and wish to be able to repeat the mining task also in-house. 
The \Actor{Hospital}, the \Actor{Specialized clinic}, and the \Actor{Pharmaceutical company} have a partial view of the overall unfolding of the inter-organizational process as they record the events stemming from the parts of their pertinence. 

In \cref{tab:partitions}, we show Alice and Bob's cases (identified by the \textbf{312} and \textbf{711} codes respectively) recorded by the  by the \Actor{Hospital} (i.e., $T^H_{312}$ and $T^H_{711}$), the \Actor{Specialized clinic} (i.e., $T^S_{312}$ and $T^S_{711}$), and the \Actor{Pharmaceutical company} (i.e., $T^C_{312}$ and $T^C_{711}$). The \Actor{Hospital} stores the identifier of these cases in the \textit{case} id attribute of its event log. Differently, the \Actor{Specialized clinic} and the \Actor{Pharmaceutical company} employees a different case denomination and stores the cross-organizational identifiers in other attributes (\textit{`TreatmentID}' and \textit{`HospitalCaseID}' respectively).  
 The partial traces of the three organizations are projections of the two combined ones for the whole inter-organizational process: $T_{312}=$\textlangle{}{\TaliceUncolored}\textrangle{} and $T_{711}=$\textlangle{}{\TbobUncolored}\textrangle{}. 
Results stemming from the analysis of the local cases would not provide a full picture. Data should be merged. However, to preserve the privacy of the people involved and safeguard the confidentiality of the information, the involved organizations \textit{cannot give open access to their traces to other organizations}. The diverging interests (being able to conduct process mining on data from multiple sources without giving away the local event logs in-clear) motivate our research.
 

%% file: content/notation.tex
Thus far, we have intuitively introduced the concepts of events, cases, event logs and related aspects.
In this section, we delve deeper into these notions from a formal standpoint. The definitions provided henceforward pave the ground for properties and aspects to which we will resort in the design and realization of our solution.

\begin{definition}[Event]\label{def:evt}
	Let $\EvtU \ni \Evt$ be a finite non-empty set of symbols. We refer to $\Evt$ as \emph{event} and to $\EvtU$ as the \emph{universe of events}.
	An \emph{attribute} is a function having the universe of events as its domain.
\end{definition}
%
Referring to the example of Alice and Bob's hospitalization process (see \cref{sec:motivating}), each event denoted as $\Evt_i$ with $1 \leq i \leq 30$ in \cref{tab:partitions} is contained in the universe of events $\EvtU$.

Due to the general nature of elements that can be linked to events via attribute functions (in the example, e.g., timestamps and activity labels), we leave the range unspecified, akin to the definition of finite sequence as a function in \cite{Mendelson/2015:IntroductionMathematicalLogic}. 
We assume the following attributes as mandatory as in the classical process mining literature~\cite{Aalst/2016:ProcessMiningBook:DataScienceinAction}.
\begin{assumption}[Mandatory attributes]\label{post:mandatoryattributes}
	The following attributes are total on the universe of events $\EvtU$:
	\begin{inparadesc}
		\item[instance identifier] (or {\CId} for short), a total surjective function $\CIdF : \EvtU \to \CIdU$, where $\CIdU$ is a finite non-empty set we name \emph{instance set};
		\item[activity label\textnormal{,}] a total function $\ActF : \EvtU \to \ActS$, where $\ActS$ is a finite non-empty set we name \emph{activity set};
		\item[timestamp\textnormal{,}] a total function $\TimeF : \EvtU \to \mathbb{N}$, where $\mathbb{N}$ is the set of positive integers.
	\end{inparadesc}
\end{assumption}
%
In the context of our motivating scenario (see \cref{tab:partitions}), 
$\CIdF(\Evt_1) = 312$. This instance identifier serves the purpose of collating the events associated with Alice's hospitalization across the three organizations. Furthermore, 
{$\ActF(\Evt_1)$} and {$\TimeF(\Evt_1)$} map $\Evt_1$ to its activity label (i.e., \texttt{PH}) and timestamp (i.e., \texttt{2022-07-14T10:36}), respectively. 
For the sake of readability, we henceforth format the timestamps in our examples using the ISO 8601 string standard. Nonetheless, they stem from an integer representation (e.g., UNIX time), as per the above definition.

\begin{definition}[Event log]\label{def:evt:log}
	Given the universe of events $\EvtU$ as per \cref{def:evt}, let ${\preceq \;\subseteq\; \EvtU \times \EvtU}$ be a total-order relation defined over $\EvtU$.
	The pair $\EvtL = \left( \EvtU, \preceq \right)$ is an \emph{event log}.
\end{definition}
In our example, 
the event log $\EvtL = \left(\EvtU,\preceq\right)$ is the totally ordered set consisting of all the events $\Evt_1, \ldots, \Evt_{30}$ recorded by the \Actor{Hospital}, the \Actor{Pharmaceutical company}, and the \Actor{Specialized clinic}. We assume that the total order of events can be reconstructed based on the timestamp associated with events, which we express with the following postulate.
\begin{assumption}[Timestamp as an order isomorphism]\label{post:order:timestamp}
	Given an event log $\EvtL = \left( \EvtU, \preceq \right)$, the timestamp attribute $\TimeF: \EvtU \to \mathbb{N}$ is an \textit{order isomorphism}, i.e., given $\Evt, \Evt' \in \EvtU$ with $\Evt \neq \Evt'$, $\Evt \preceq \Evt'$ if and only if $\TimeF(\Evt) \leqslant \TimeF(\Evt')$.
\end{assumption}
Examining the events $\Evt_{30}$ and $\Evt_{16}$ in our motivating scenario, we can deduce from the order relation established based on their timestamps (i.e., \texttt{2022-07-16T05:06} and \texttt{2022-07-16T07:06}, respectively) that $\Evt_{30}$ precedes $\Evt_{16}$, i.e., ${\Evt_1}\preceq{\Evt_{30}}$. This rationale holds true for any generic pair of events selected from the event log.

\begin{definition}[Process mining function]\label{def:process:mining}
	Given an event log $\EvtL$, a \textit{process mining function} $\PMFunc$ is a function having as its arguments $\EvtL$ and a tuple of hyperparameters $\PMParams$, with $k \in \mathbb{N}$.
\end{definition}
We leave the above definition purposefully underspecified to cater for the width of process mining techniques ranging over discovery, conformance checking, and process enhancement. Those categories all have in common that the event log is the input, whereas the hyperparameters and the output vary. A process discovery technique such as the \emph{heuristic miner}, e.g., takes a set of thresholds to pruning out noise as its hyperparameters and returns a heuristic net~\cite{weijters2006process}. A consequence of covering process mining techniques under the umbrella of a function is that we assume them to be deterministic.

In our setting, events may be recorded by different entities. We name these entities \emph{provisioners} and formally define them as follows.
\begin{definition}[Provisioner]\label{def:provisioner}
	Let $\LPrvS = \{ \LPrv_1, \ldots, \LPrv_n \}$ be a non-empty finite set of entities (with $n \in \mathbb{N}$) named \emph{provisioner}s. 
    Given the universe of events $\EvtU$ as per \cref{def:evt}, $\EqFun{\LPrv}$ is a total surjective function having $\EvtU$ as the domain and $\LPrvS$ as its range. Function $\EqFun{\LPrv}$ defines an equivalence relation $\EqRel{\LPrv}$ on $\EvtU$ (i.e., $\Evt \EqRel{\LPrv} \Evt'$ iff $\Evt \EqFun{\LPrv} \LPrv_i$ and $\Evt' \EqFun{\LPrv} \LPrv_i$ for some $\LPrv_i \in \LPrvS$, $1 \leq i \leq n$) and hence partitions $\EvtU$ into equivalence classes we denote with $\EqCls{\LPrv_i}$ for any $\LPrv_i \in \LPrvS$.
\end{definition}
In the depicted scenario, we introduce three provisioners: the \Actor{Hospital}, the \Actor{Pharmaceutical company}, and the \Actor{Specialized clinic}, henceforth referred to as $\LPrv_\textrm{Hsp}$, $\LPrv_\textrm{PhC}$, and $\LPrv_\textrm{SpC}$. Events $\Evt_{1}$ and $\Evt_{19}$ in \cref{tab:trace} fall within the same equivalence class $\EqCls{\LPrv_\textrm{Hsp}}$ as they are both recorded by the \Actor{Hospital}. Consequently, the relation $\Evt_1 \EqRel{\LPrv} \Evt_{19}$ holds true. Conversely, $\Evt_{1}$ and $\Evt_{30}$ are recorded by the \Actor{Hospital} and the \Actor{Specialized clinic}, respectively, placing them in distinct equivalence classes, namely $\EqCls{\LPrv_\textrm{Hsp}}$ and $\EqCls{\LPrv_\textrm{SpC}}$. Therefore, the relation $\Evt_1 \EqRel{\LPrv} \Evt_{30}$ does not hold.

\begin{definition}[Log partition]\label{def:partition}
	Given an event log $\EvtL = \left( \EvtU, \preceq \right)$ as per \cref{def:evt:log}, provisioner $\LPrv \in \LPrvS$ and the equivalence class 
	$\EqCls{\LPrv} \subseteq \EvtU$
	as per \cref{def:provisioner}, a \emph{log partition} is a pair $\EvtL_\LPrv = \left( \EqCls{\LPrv}, \preceq_\LPrv \right)$, where $\preceq_\LPrv$ is a total order restricting $\preceq$ on $\EqCls{\LPrv}$, i.e., $\;\preceq_\LPrv$ is $(\preceq) \cap \left(\EqCls{\LPrv} \times \EqCls{\LPrv}\right)$.
	
\end{definition}

In our setting, we posit that each organization independently maintains its own log partition. The events  $\Evt_{1}, \Evt_{2}, \ldots, \Evt_{19}$ in the  $\EqCls{\LPrv_\textrm{Hsp}}$ equivalence class are contained in the \Actor{Hospital}'s log partition  $\left( \EqCls{\LPrv_\textrm{Hsp}}, \preceq_{\LPrv_\textrm{Hsp}} \right)$. Similarly, the totally ordered sets defined over $\EqCls{\LPrv_\textrm{PhC}} = \left\{\Evt_{20}, \Evt_{21}, \ldots, \Evt_{25} \right\}$ and $\EqCls{\LPrv_\textrm{SpC}} = \left\{ \Evt_{26}, \Evt_{27}, \ldots, \Evt_{30} \right\}$ are the \Actor{Pharmaceutical company} and the \Actor{Specialized clinic}'s log partitions, respectively -- i.e., $\left( \EqCls{\LPrv_\textrm{PhC}}, \preceq_{\LPrv_\textrm{PhC}} \right)$ and  $\left( \EqCls{\LPrv_\textrm{SpC}}, \preceq_{\LPrv_\textrm{SpC}} \right)$.

We recall that the restriction of a totally ordered set (like the event log) is a total order. Therefore, the log partition is a totally ordered set, too, and hence an event log. In the remainder of this paper, we shall name as \textbf{sublog} a restriction of an event log (omitting the reference to the event log whenever clear from the context). For instance, since $\Evt_1 \preceq \Evt_2$, then $\Evt_1~\preceq_{\LPrv_\textrm{Hsp}} \Evt_2$.

\begin{definition}[Segment]\label{def:segment}
	Given a log partition $\left( \EqCls{\LPrv}, \preceq_\LPrv \right)$ 
	defined over the equivalence class $\EqCls{\LPrv}$ by provisioner $\LPrv \in \LPrvS$ as per \cref{def:provisioner,def:partition}, 
	and a subset $\{\CId_1, \ldots, \CId_k\} \subseteq \CIdU$ (see~\cref{post:mandatoryattributes}) with $1 \leq k \leq |\CIdU|$, $k \in \mathbb{N}$,
	a \emph{segment} over $\{\CId_1, \ldots, \CId_k\}$
	is a pair
	$\Segm = \left( \EqCls{\Segm}, \preceq_\Segm \right)$ 
	wherein events
	$\Evt_\Segm \in \EqCls{\Segm} \subseteq \EqCls{\LPrv}$
	are such that $\CIdF(\Evt_\Segm) \in \{\CId_1, \ldots, \CId_k\}$
	and $\;\preceq_\Segm$ is a total order restricting $\preceq_\LPrv$ on $\EqCls{\Segm}$, i.e.,
	$(\preceq_\LPrv) \cap \left(\EqCls{\Segm} \times \EqCls{\Segm}\right)$.
	
\end{definition}
In light of the above, the segment is a totally ordered set, too. Therefore, a segment is an event log (or sublog, depending on the context). Also, notice that a segment can be the whole log partition if the given set of {\CId}s contains all the {\CId}s to which the events in the log partition map. Finally, notice that all events that map to a given $\CId$ in a log partition are in the same segment.
In the context of our example, e.g., consider the \Actor{Pharmaceutical company}'s log partition $\left( \EqCls{\LPrv_\textrm{PhC}}, \preceq_{\LPrv_\textrm{PhC}} \right)$ and Alice's hospitalization ({\CId} 312). As $\CIdF(\Evt_{20}) = \CIdF(\Evt_{22}) = \CIdF(\Evt_{24}) = 312$ and $ \Evt_{20}, \Evt_{22}, \Evt_{24} \in \EqCls{\LPrv_\textrm{PhC}} $, all the three aforementioned events reside within the same segment. In other words, \cref{def:segment} precludes any scenario wherein 
two events pertaining to Alice's case in the pharmaceutical company's partition are allocated to different segments (e.g., it cannot happen that $\Evt_{20}$ is in \Segm' and $\Evt_{22}$, $\Evt_{24}$ are in \Segm'' if $\Segm' \neq \Segm''$).

\begin{definition}[Case]\label{def:case}
	Let $\EvtL=\left( \EvtU, \preceq \right)$ be an event log as per \cref{def:evt:log}. 
	Let $\CIdF : \EvtU \to \CIdU$ be the total surjective function of the mandatory attribute {\CId} (\cref{post:mandatoryattributes})
	that defines an equivalence relation $\EqRel{\CIdF}$ on $\EvtU$ (i.e., $\Evt \EqRel{\CIdF} \Evt'$ iff $ \CIdF(\Evt) = \CId$ and $\CIdF(\Evt') = \CId$ for some $\CId \in \CIdU$) and hence divides $\EvtU$ into equivalence classes.
    Given an instance identifier $\CId \in \CIdU$, 
    a \emph{case} $\Case_{\CId}$ 
    is a pair $\left( \EqCls{\CId}, \preceq_\CId \right)$, wherein $\EqCls{\CId} \subseteq \EvtU$ is the equivalence class for $\CId$ by the equivalence relation $\EqRel{\CIdF}$, and $\preceq_\CId$ is a total order restricting $\preceq$ on $\EqCls{\CId}$, i.e., $\;\preceq_\CId$ is $(\preceq) \cap \left(\EqCls{\CId} \times \EqCls{\CId}\right)$.
    Given a $\CId$ and $\EvtL$, {\CasF} is the function returning $\Case_{\CId}$ from $\EvtL$ as above, i.e., $\Case_{\CId} = \CasF\left(\EvtL, \CId\right)$.
    We shall omit the $\CId$ subscript from $\Case_\CId$ when clear from the context. 
    We denote with $\CasU \ni \Case$ the set of all cases in $\EvtL$, i.e., $\CasU = \bigcup\limits_{\CId \in \CIdU} \left\{\left( \EqCls{\CId}, \preceq_\CId \right)\right\}$.
\end{definition}
Considering our motivating scenario, we can distinguish Alice and Bob's cases, referred to as $\Case_{312} = \CasF(\EvtL, 312)$ and $\Case_{711} = \CasF(\EvtL, 711)$, respectively. Events $\Evt_1$ and $\Evt_{30}$ belong to the same equivalence class  $\EqCls{312}$ (i.e., the class associated with $\Case_{312}$), given that $\CIdF(\Evt_1)=\CIdF(\Evt_{30})=312$. Therefore, the relation $\Evt_1 \EqRel{\CIdF} \Evt_{30}$ holds. In contrast, $\Evt_{1}$ and $\Evt_{21}$ fall in different equivalence classes (i.e., $\EqCls{312}$ and $\EqCls{711}$, respectively), as $\CIdF(\Evt_1) = 312$ and $\CIdF(\Evt_{21}) = 711$. Thus, $\Evt_1 \EqRel{\CIdF} \Evt_{30}$ does not hold true.

\begin{table}[bt]
	\caption{Glossary of notation and summary of formal concepts}
	\label{tab:notations}
	\resizebox{\linewidth}{!}{%
		\input{tables/glossary}%
	}%
\end{table}
Intuitively, a case is derived from an event log by taking all events that share the same {\CId} in that log. Also partitions, segments, and sublogs in general are logs, thus the composition of cases depend on the structure from which they are extrapolated. Furthermore, since a case is per se a totally ordered set of events, it is an event log itself. 
Notice that a case is a restriction of an \emph{event log} to events having the same \emph{unique instance identifier} in common. A segment, instead, is a restriction of a \emph{log partition} to events that have a non-empty \emph{set of instance identifiers} in common.
In other words, segments can contain events mapping to different {\CId's}, whilst cases cannot. On the other hand, cases can contain events stemming from different provisioners, whilst segments cannot.

\bigskip

We conclude this section with an operation defined on event logs (hence sublogs, i.e., segments, partitions, and cases as well) we name \emph{merge}, and an important property that a merge should enjoy to be declared \emph{safe}.

\begin{definition}[Merge]\label{def:merge}
	Let
	${\EvtL' = \left(\EvtU_{\EvtL'}, \preceq_{\EvtL'}\right)}$ and
	${\EvtL'' = \left(\EvtU_{\EvtL''}, \preceq_{\EvtL''}\right)}$ be two event logs as per \cref{def:evt:log}. 
	The \emph{merge} of two event logs $\EvtL' \Merge \EvtL''$ is a pair $\left(\EvtU_{\EvtL'} \cup \EvtU_{\EvtL''}, \preceq_{\Merge} \right)$ where $\preceq_{\Merge}$ is a superset of $\left( \preceq_{\EvtL'} \cup \preceq_{\EvtL''} \right)$.
	A merge is \emph{safe} if the two following conditions hold:
	\begin{compactenum}
		\item $\preceq_{\Merge}$ is a total order over $\EvtU_{\EvtL'} \cup \EvtU_{\EvtL''}$;
		\item the timestamp attribute $\TimeF: \EvtU_{\EvtL'} \cup \EvtU_{\EvtL''} \to \mathbb{N}$ is an order isomorphism as per \cref{post:order:timestamp}.
	\end{compactenum}
\end{definition}
Notice that the merge of two event logs (or sublogs) still is a totally ordered set (and hence an event log, or a sublog) if and only if the merge is safe. In other words, event logs are closed under the safe-merge operation. For instance, a non-safe merge between $\EvtL'$ and $\EvtL''$ could be computed as $\left(\EvtU_{\EvtL'} \cup \EvtU_{\EvtL''}, \preceq_{\EvtL'} \cup \preceq _{\EvtL''} \right)$. Take, e.g., the segment over {\CId} $312$ in the \Actor{Hospital}'s partition as $\EvtL'$ and the segment over {\CId} $312$ in the \Actor{Pharmaceutical company}'s partition as $\EvtL''$. With a non-safe merge like the one we just mentioned, the order would not be defined for either of $\Evt_{11}$ (in $\EvtL'$) and $\Evt_{20}$ (in $\EvtL''$), among others.
\begin{proposition}\label{prop:safemerge:assoc:commut}
    Let $\EvtL'$, $\EvtL''$, $\EvtL'''$ be three event logs. If $\Merge$ is a safe merge, it enjoys the properties of 
    \begin{inparaenum}[\itshape(i)\upshape]
    \item associativity $((\EvtL' \Merge \EvtL'') \Merge \EvtL''') = (\EvtL' \Merge (\EvtL'' \Merge \EvtL'''))$ and
    \item commutativity $(\EvtL' \Merge \EvtL'') = (\EvtL'' \Merge \EvtL')$, and 
    \item has $(\{\}, \{\})$, i.e., a log with no events, as its identity element.
    \end{inparaenum}
\end{proposition}
\begin{proof}{(Sketch)}
    It follows from the definition of $\Merge$, employing the associative and commutative operator $\cup$ having $\{\}$ as its identity element, and the safety condition imposing a deterministic assignment of the order relation based on the events' timestamps over $\preceq_{\Merge}$.
\end{proof}

For the sake of readability, we recapitulate the formal concepts introduced thus far in a glossary of notation in \cref{tab:notations}, listing them in their order of appearance within this section.
%
%
Equipped with these notions, we provide an architectural overview of our solution in the next section.

%% file: tables/glossary.tex
\begin{tabular}{lp{14cm}l}
	\toprule
	Symbol & Brief description & Reference \\
	\midrule
	$\Evt \in \EvtU$ & An event $e$ in a universe of events $\EvtU$ & \cref{def:evt} \\
	$\CIdF : \EvtU \to \CIdU$ & The {\CId} attribute mapping events to instance identifiers in $\CIdU$ & \cref{def:evt,post:mandatoryattributes} \\
	$\ActF : \EvtU \to \ActS$ & The attribute mapping events to activity labels in $\ActS$ & \cref{def:evt,post:mandatoryattributes} \\
	$\TimeF : \EvtU \to \mathbb{N}$ & The attribute mapping events to timestamps & \cref{def:evt,post:mandatoryattributes} \\
	$\EvtL = \left( \EvtU, \preceq \right)$ & The event log, a set of events totally ordered by ${\preceq \;\subseteq\; \EvtU \times \EvtU}$. The timestamp is an order isomorphism to $\preceq$. & \cref{def:evt:log,post:order:timestamp} \\
        $\PMFunc$& A process mining function with $k$ hyperparameters $\PMParams$  &\cref{def:process:mining}\\
        $\LPrv \in \LPrvS$&A provisioner {\LPrv} in a universe of provisioners {\LPrvS} &\cref{def:provisioner}\\
        $\EqFun{\LPrv} : \EvtU \to \LPrvS$& The function mapping events to provisioners in {\LPrvS}&\cref{def:provisioner}\\
        $\EqCls{\LPrv}$&The equivalence class of the events in {\EvtU} related to the same provisioner {\LPrv}&\cref{def:provisioner}\\
        $\Evt\EqRel{\LPrv}\Evt'$&The equivalence relation between two events in {\EvtU}  related to the same provisioner {\LPrv}&\cref{def:provisioner}\\
        $\EvtL_\LPrv = \left( \EqCls{\LPrv}, \preceq_\LPrv \right)$&The log partition, a set of events in \EqCls{\LPrv} ordered by $\preceq_\LPrv$. $\preceq_\LPrv$ is a total order restricting $\preceq$ to $\EqCls{\LPrv}$&\cref{def:provisioner}\\
        $\Segm = \left( \EqCls{\Segm}, \preceq_\Segm \right)$&The segment, a subset of a log partition collecting all events whose $\CId$ belongs to a given subset of {\CIdU} ordered by the restriction $\preceq_\Segm$ of $\preceq_\LPrv$&\cref{def:segment}\\
        $\Evt \EqRel{\CIdF} \Evt'$&The equivalence relation between two events related to the same {\CId}&\cref{def:case}\\
        $\EqCls{\CId}$&The equivalence class of the events in {\EvtU} related to the same {\CId} &\cref{def:case}\\
        $\Case_{\CId} = \left( \EqCls{\CId}, \preceq_\CId \right)$& The case, a set of events in {\EvtU} related to the same {\CId} ordered by $\preceq_\CId$. $\preceq_\CId$ is a total order restricting $\preceq$ on $\EqRel{\CIdF}$&\cref{def:case}\\
        $ \left(\EvtU_{\EvtL'}, \preceq_{\EvtL'}\right) \Merge \left(\EvtU_{\EvtL''}, \preceq_{\EvtL''}\right)$& The merge between two event logs, an event log containing the set of all the events in $\EvtU_{\EvtL'}$ and  $\EvtU_{\EvtL''}$ ordered by a superset of $\left( \preceq_{\EvtL'} \cup \preceq_{\EvtL''} \right)$. A merge is \emph{safe} if the output is totally ordered and the timestamp attribute is an order isomorphism. & \cref{def:merge,prop:safemerge:assoc:commut}\\
        
	\bottomrule
        
\end{tabular}

%% file: content/design.tex
Our goal is to enable the secure aggregation and elaboration of original, unaltered event logs from decentralized sources in dedicated environments that potentially lie beyond the individual organizations' information perimeter.
With this objective in mind, we devise the \Compo{Secure Miner} component, which is capable of securing data merge and processing by running certified code in an isolated execution vault. Thus, we decouple provisioning from treatment, and the two tasks can be carried out by distinct computing nodes. In this section, we outline the high-level architecture of CONFINE.

\subsection{Architectural overview}
\label{sec:design:arch}
Here, we introduce CONFINE's key components, with a special focus on the \Compo{Secure Miner} (\cref{sec:design:arch:secureminer}) after an outline of the overall architecture (\cref{sec:design:arch:overview}).
\subsubsection{The CONFINE architecture at large}
\label{sec:design:arch:overview}
Our architecture involves different information systems running on multiple machines. An organization can take at least one of the following roles: 
\begin{inparadesc}
\item[provisioning] if it delivers local event logs to be mined (as per \cref{def:provisioner});
\item[miner] if it applies process mining algorithms (see \cref{def:process:mining}) using event logs retrieved from provisioners.
\end{inparadesc}
In our adversarial model, we assume that provisioning organizations behave as \emph{honest} entities (\Asm{}\label{asm:honestprovisioners}), whereas miner organizations may act in a \emph{malicious} manner (\Asm{}\label{asm:maliciousminers}). This scenario characterizes situations in which the members of a consortium cooperate towards a common business objective and are thus motivated to gain benefits from constructive behavior. Consequently, we do not expect provisioning organizations to deviate from the prescribed protocol or manipulate the transmitted event logs. By contrast, miner organizations may attempt to exploit their role by inferring sensitive information from the received event logs. 

\begin{figure}[t]
	\begin{floatrow}
		\ffigbox[0.63\textwidth]
		{\includegraphics[width=0.63\textwidth]{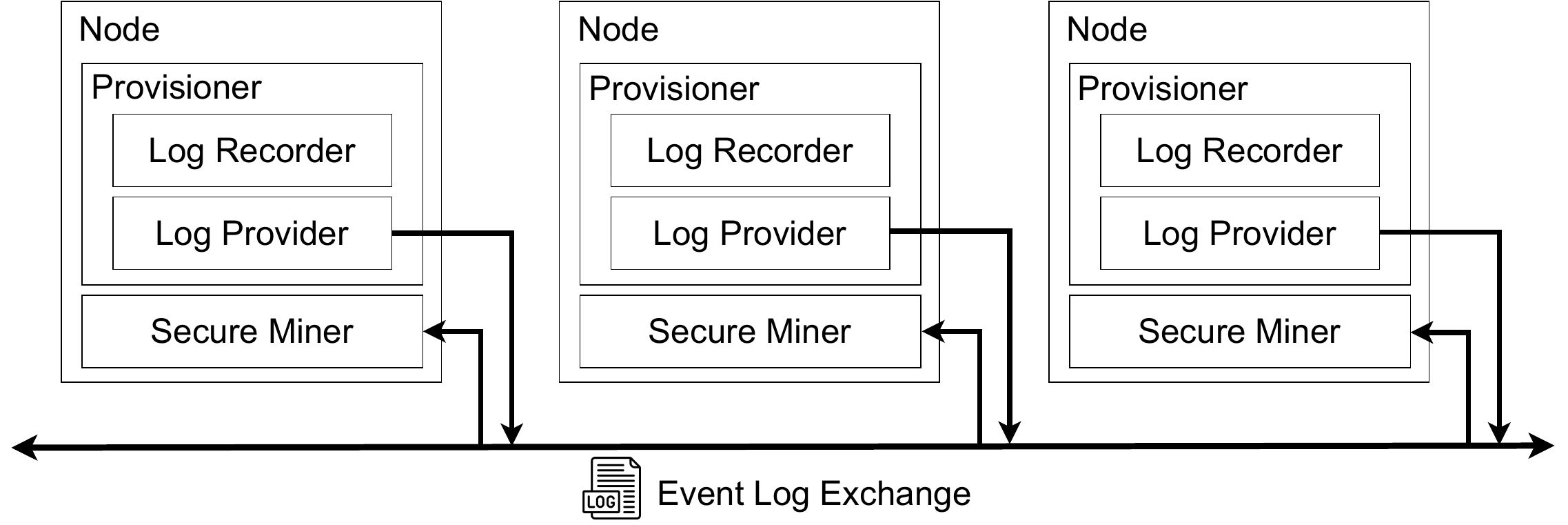}}
		{\caption{The CONFINE high-level architecture}\label{fig:architecture_diagram}
			\vspace{-0.5em}
		}%
		\ffigbox[0.29\textwidth]
		{\includegraphics[width=0.29\textwidth]{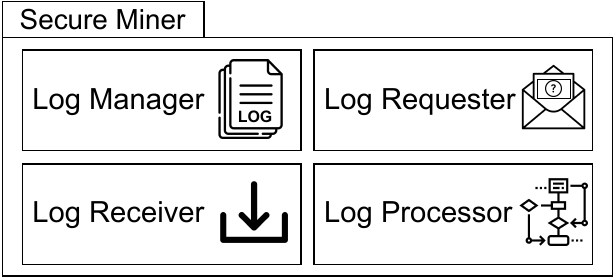}\vspace{.49cm}}
		{\caption[Secure Miner sub-components]{The Secure Miner}\label{fig:trusted_miner}
			\vspace{-0.3em}
		}
	\end{floatrow}
\end{figure}
In \cref{fig:architecture_diagram}, we propose the high-level schematization of the CONFINE framework.
Henceforth, we will refer to provisioning organizations' event logs as \emph{log partitions}, according to \cref{def:partition}.
%
%
In our solution, every organization hosts one or more nodes realizing the architectural components. Depending on the played role, nodes come endowed with a \Compo{Provisioner} or a \Compo{Secure Miner} component, or both. The \Compo{Provisioner} component consists of the following two main sub-components.
The \begin{inparadesc}
\item[\Compo{Log Recorder}] registers the \emph{events} (introduced in \cref{def:evt}) taking place in the organizations' systems, and builds the local log partitions. The \item[\Compo{Log Provider}] delivers on-demand log partitions to mining players.
\end{inparadesc}

As formalized in \cref{post:mandatoryattributes}, we assume events in log partitions are comprised of three mandatory attributes, namely, the \emph{iid}, the \emph{activity} label and the \emph{timestamp} (\Asm{}\label{asm:eventAttributes}; see~\cref{post:mandatoryattributes}). We further assume that iids are synchronized across organizations so that they consistently refer to the same process instance (\Asm{}\label{asm:synchronizedIIDs}).
The \Compo{Log Recorder} of the \Actor{Hospital} and all other organizations in our example refer to Alice and Bob's cases (see \cref{def:case}) with iids 312 and 711, respectively. The \Compo{Log Recorder} is queried by the \Compo{Log Provider} for log partitions to be made available for mining. The latter controls access to local 
log paritions by authenticating data requests by miners and rejecting those that come from unauthorized organizations.
In our motivating scenario, the \Actor{Specialized clinic}, the \Actor{Pharmaceutical company}, and the \Actor{Hospital} leverage \Compo{Log Provider}s to authenticate the miner party before sending their log partition. The \Compo{Secure Miner} component
shelters external log partitions inside a protected environment to preserve data confidentiality and integrity.
Notice that \Compo{Log Provider}s accept requests issued solely by \Compo{Secure Miner}s. 
Next, we provide an in-depth focus on the latter.

\subsubsection{The Secure Miner}
\label{sec:design:arch:secureminer}
The primary objective of the \Compo{Secure Miner} is to allow miners to securely execute process mining algorithms using decentralized log partitions retrieved from provisioners such as the \Actor{Specialized clinic}, \Actor{Pharmaceutical company}, and the \Actor{Hospital} in our example. \Compo{Secure Miner}s are isolated components that guarantee data inalterability and confidentiality. In \cref{fig:trusted_miner}, we show a schematization of the \Compo{Secure Miner}, which consists of four sub-components:
\begin{inparaenum}[\itshape(i)\upshape]
    \item \Compo{Log Requester};
    \item \Compo{Log Receiver};
    \item \Compo{Log Manager}; 
    \item \Compo{Log Processor}.
\end{inparaenum}
The \Compo{Log Requester} and the \Compo{Log Receiver} are the sub-components that we employ during the log partition retrieval. \Compo{Log Requester}s send authenticable data requests to the \Compo{Log Provider}s. The \Compo{Log Receiver} collects log partitions sent by \texttt{Log Providers} and entrusts them to the \Compo{Log Manager}, securing them from accesses that are external to the \Compo{Secure Miner}.
Miners of our motivating scenario, such as the \Actor{University} and the \Actor{National Institute of Statistics}, employ these three components to retrieve and store Alice and Bob's data. 
The \Compo{Log Manager} combines the event data locked in the \Compo{Secure Miner} to have a global view of the inter-organizational process, encompassing activities executed by each involved provisioning organization. 
We refer to this operation as \emph{merge}, according to \cref{def:merge}.
%
In the context of the merge, we assume that the events contained in the individual log partitions expose timestamp attributes that are mutually synchronized across organizations (\Asm{}\label{ams:timestamps}).
The \Compo{Log Processor} executes process mining algorithms using merged event logs in a protected environment, inaccessible from the outside computation environment.
A \Compo{Log Processor} implements a process mining function $\PMFunc$ following \cref{def:process:mining}.%
In our motivating scenario, the \Compo{Log Manager} combines the traces associated with the cases of Alice (i.e., $T^H_{312}$, $T^S_{312}$, and $T^C_{312}$) and Bob (i.e, $T^H_{711}$, $T^S_{711}$, and $T^C_{711}$) and generates the chronologically sorted traces $T_{312}$ and $T_{711}$ preserving the events' temporal sequence (as described in \cref{post:order:timestamp}). Subsequently, the \Compo{Log Processor} feeds these cases into the mining algorithms. 


%% file: content/realization.tex
In this section, we outline the technical aspects concerning the realization of our solution. Therefore, we first present the enabler technologies through which we instantiate the architectural principles presented in \cref{sec:design}. Afterwards, we discuss the CONFINE interaction protocol.
%
\begin{figure}[t]
	\centering
	\includegraphics[width=1\linewidth]{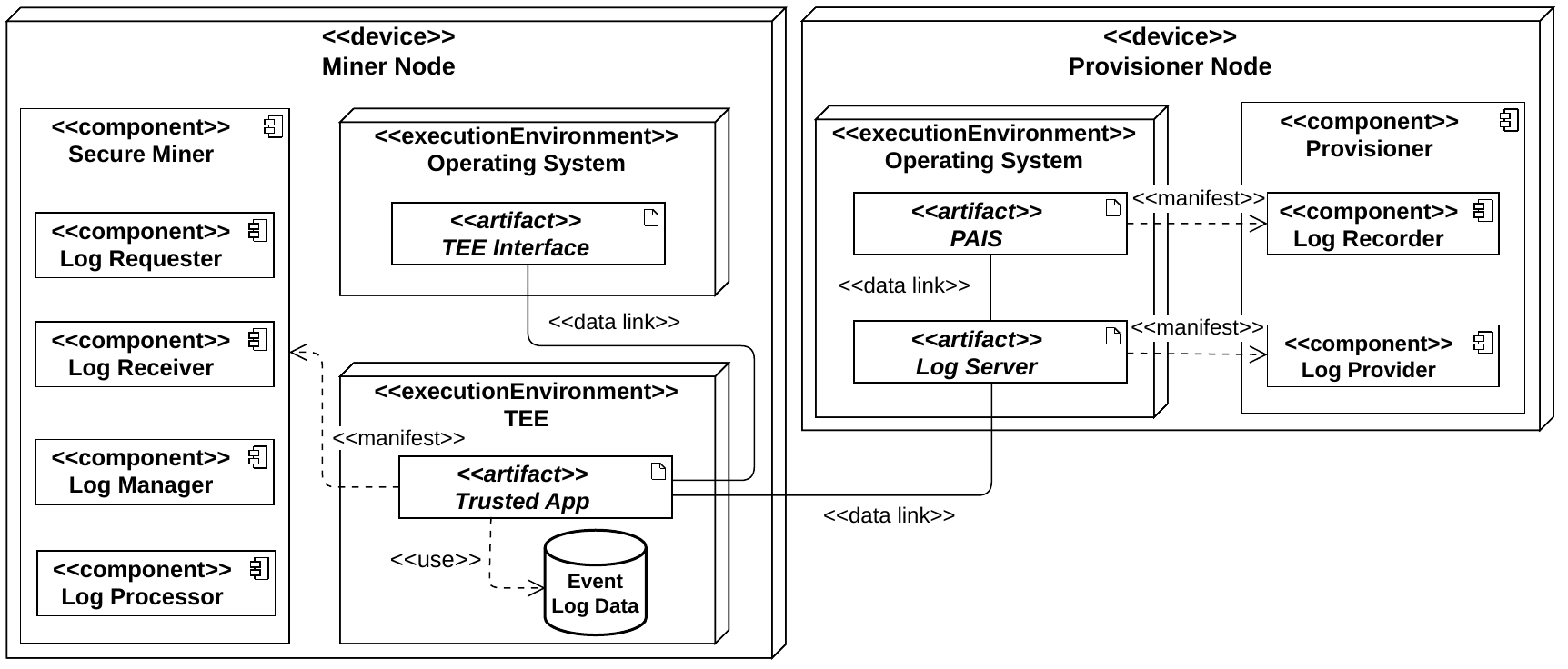}
	\caption{UML deployment diagram of the CONFINE architecture}
	\label{fig:deployment_diagram}
\end{figure}
\subsection{Deployment}
\label{sec:deployment}
%
\Cref{fig:deployment_diagram} depicts a Unified Modeling Language (UML) deployment diagram~\citep{koch2002expressive} to illustrate the employed technologies and computation environments. 
We recall that the \Compo{Miner} and \Compo{Provisioner} nodes are drawn as separated, although organizations can host both. In our motivating scenario, e.g., the \Actor{Hospital} can be a provisioner (as per \cref{def:provisioner}) and a miner at the same time.
Our deployment setting assumes the establishment of Identity and Access Management (IAM) services capable of issuing, for each node, an \emph{identity proof} that authenticates its owning organization (\Asm{}\label{asm:iam}). 
%

\Compo{Provisioner Node}s host the \Compo{Provisioner}'s components, encompassing the \Compo{Log Recorder} and the \Compo{Log Provider}. 
The Process-Aware Information System (\Compo{PAIS}) manifests the \Compo{Log Recorder}~\citep{Dumas.etal/2018:FundamentalsofBPM}. 
The \Compo{PAIS} grants access to the \Compo{Log Server}, enabling it to retrieve event log data. The \Compo{Log Server}, on the other hand, embodies the functionalities of the \Compo{Log Provider}, implementing web services aimed at handling remote data requests and providing log partitions (introduced in \cref{def:partition}) to miners. 

The \Compo{Miner Node} is characterized by two distinct \textit{execution environments}: the \Compo{Operating System} (\Compo{OS}) and the Trusted Execution Environment (\Compo{TEE})~\citep{DBLP:conf/trustcom/SabtAB15}. \Compo{TEE}s establish isolated contexts separate from the \Compo{OS}, safeguarding code and data through hardware-based encryption mechanisms.

As introduced in \cref{sec:background:tee}, we assume our \Compo{TEE} profile is characterized by the following distinctive features:
\begin{inparaenum}[\itshape(i)\upshape]
	\item application-level trusted surface,
	\item hardware-based memory encryption, and
	\item remote attestability (\Asm{}\label{asm:teefeatures}).
\end{inparaenum}
Through the application-level trusted surface, we execute a \Compo{Trusted App} responsible for implementing the core functions of the \Compo{Secure Miner} and its associated sub-components within the \Compo{TEE}'s isolated domain. The \Compo{TEE} preserves the integrity of the \Compo{Trusted App} code by shielding it from malicious tampering and preventing unauthorized access by processes running in the \Compo{Operating System}. 

Through hardware-based memory encryption, the \Compo{TEE} ensures that the log partitions processed by the \Compo{Secure Miner} remain encrypted throughout their computation, confining them to dedicated and protected regions of the main memory.
By enforcing memory access restrictions, the \Compo{TEE} prevent one trusted application from reading or altering the memory space of another. 
These dedicated areas in memory are limited, though.
Once the limits are exceeded, TEEs have to scout around in outer memory areas, thus conceding the opportunity to malicious readers to understand the saved data based on the memory reads and writes.
To avoid this risk, TEE implementations often raise errors that halt the program execution when memory demand goes beyond the available space. Within the CONFINE protocol, we mitigate this issue by directing the \Compo{Secure Miner} to elaborate log partitions in a smaller batch of cases (defined in \cref{def:case}) we refer to as \emph{segment} (see \cref{def:segment}).

Finally, we rely on the \Compo{TEE}'s feature of remote attestability, which enables \Compo{Provisioner}s to verify the trustworthiness of the \Compo{Secure Miner}. In particular, we leverage the customizable portion of the attestation evidence to embed ephemeral cryptographic means generated at runtime by our \Compo{Secure Miner}. Those cryptographic tools are then used to establish a secure communication channel for the encrypted transmission of log partitions.

Users interact with the \Compo{Trusted App} through the \Compo{Trusted App Interface}, which serves as the exclusive communication channel. The \Compo{Trusted App} offers secure methods, invoked by the \Compo{TEE Interface}, for safely receiving information from the \Compo{Operating System} and outsourcing the results of computations. 

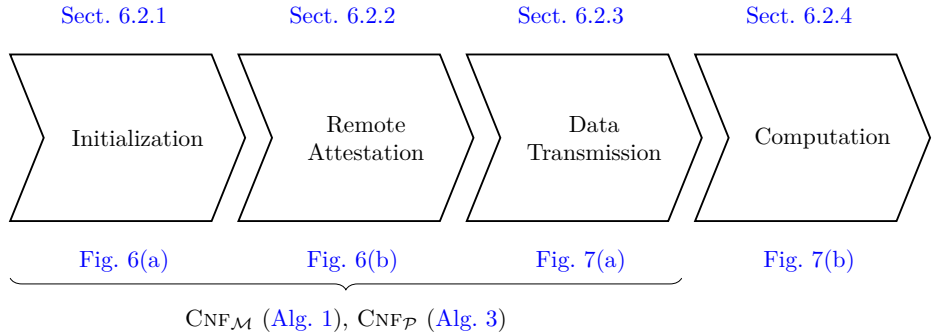
\begin{figure}[t]
\centering
\resizebox{0.75\textwidth}{!}{%
\input{content/figures/protocolphases}
}
\caption{Phases of the CONFINE protocol with references to the related sections, figures, and algorithms in this article}
\label{fig:protocolphases}
\end{figure}

\subsection{The CONFINE protocol}
\label{sec:deployment:protocol}
We orchestrate the interaction of the components in CONFINE via a protocol. 
We separate it in four subsequent stages, namely
\begin{inparaenum}[\itshape(i)\upshape]
	\item \textit{initialization}, \item \textit{remote attestation}, \item \textit{data transmission}, and \item \textit{computation}.
\end{inparaenum}
The initialization phase aims to exchange preliminary metadata, enabling the \Compo{Secure Miner} to acknowledge the distribution of cases among the participating \Compo{Provisioner}s. During the remote attestation phase, the \Compo{Provisioner}s verify the security properties of the \Compo{Secure Miner} prior to the externalization of their log partitions. In the subsequent data transmission phase, the \Compo{Provisioner}s transfer their event logs in smaller batches of cases according to a segmentation strategy. Finally, the \Compo{Secure Miner} executes the selected process mining algorithm over the merged cases received from the \Compo{Provisioner}s.

In \cref{fig:protocolphases}, we associate these phases with the sections, figures, and algorithms explicitly addressing them throughout the text.    
The protocol's phases are individually depicted in \cref{fig:init,fig:attestation,fig:transmission,fig:computation}, respectively.
Our protocol involves two primary entities: a \Compo{Secure Miner} (hereafter referred to as \SecM) and one or more \Compo{Provisioner}s ($\LPrv_1, \ldots, \LPrv_n \in \LPrvS$). 

The behavioral descriptions for {\SecM} (which we refer to as $\Confine_{\SecM}$) and any $\LPrv_i \in \LPrvS$ (i.e., $\Confine_{\LPrv}$) are outlined in \cref{alg:secm,alg:lprv}, respectively. These specifications adhere to the syntax for distributed algorithms detailed in~\cite{Cachin.etal/2011:ReliableSecureDistributedProgramming}.%
\footnote{In order to enhance clarity, we adapt the original syntax of the \Mesg{Deliver} and \Mesg{Send} expressions to emphasize message senders (preceded by the symbol `$\ll$') and receivers (preceded by `$\gg$'), respectively.}
We assume that communication between \Compo{Secure Miner}s and \Compo{Log Provisioner}s occurs through an \textit{Authenticated Point-to-Point Perfect Link} (\Asm{}\label{asm:authperfectlink})~\cite{Cachin.etal/2011:ReliableSecureDistributedProgramming}. 
This communication abstraction guarantees:
\begin{inparaenum}[\itshape(i)\upshape]
\item \textit{reliable delivery} (i.e., if a correct process sends a message \textit{m} to a correct
process \textit{q}, then \textit{q} eventually delivers \textit{m}),
\item \textit{no duplication} (i.e., no message is delivered by a correct process more than once), and
\item \textit{authenticity} (i.e., if some correct process \textit{q} delivers a message \textit{m} with sender \textit{p}
and process \textit{p} is correct, then \textit{m} was previously sent to \textit{q} by \textit{p}). 
\end{inparaenum} 
In addition, we assume that \Compo{Provisioner}s implement an access control policy that authorizes or denies requests issued by \Compo{Secure Miner}s based on the identity of the requesting organization (\Asm{}\label{asm:acpolicy})~\cite{DBLP:conf/post/CramptonM12}.


$\Confine_{\SecM}$ accepts as input the list of \Compo{Provisioner}s' references ($\LPrv_1, \ldots, \LPrv_n$), namely, descriptors of all the necessary information to locate and identify provisioners.
In the same specification \cref{alg:secm}, the \Compo{Secure Miner} is provided with the following hyperparameters:
\begin{inparadesc}
    \item the {\DoYieldCases} flag that specifies whether the computation is incrementally yielded on single complete cases or on the entire event log, and
    \item a segment size {\SegSize} employed for the generation of the log \emph{segments} (formally introduced in \cref{def:segment}) during the \textit{data transmission} phase. 
\end{inparadesc}
We posit that the segment size range spans from the largest case (see \cref{def:case}) size among all the merged cases, as the minimum value, to the maximum memory capacity of the \Compo{TEE}, as the maximum value (\Asm{}\label{asm:segsize}). 
Referring to our example, in which the merged Alice's case (comprehensive of the all the three organizations' events) constitutes the largest case with a cumulative size of \num{100}~KB, and considering the capacity of the \Compo{TEE} to be \num{128000}~KB, it follows that the segment size employed by the miner consistently falls between these two values.

Similarly, the \Compo{Provisioner}'s specification $\Confine_{\LPrv}$, documented in \cref{alg:lprv}, considers as input the list of references to miners ($\SecM_1, \ldots, \SecM_s$) for which event log access is enabled. According to the underlying syntax, {\SecM} and {\LPrv} execute code prompted by events mutually exclusively, implying that they do not manage two or more events/messages concurrently. For the sake of clarity, we omit any explicit representation of this feature in the pseudo-codes under discussion. Furthermore, to preserve the comprehensibility of the components' behavior, we leave out of the proposed pseudo-codes the validity checks and the encryption operations that are addressed in the text and in \cref{fig:init,fig:attestation,fig:transmission,fig:computation}. Also, we leave out of \cref{alg:secm} the behaviour of the \Compo{Secure Miner}  {\SecM} during the computation phase, as our approach is designed to be integrable with generic process mining algorithms (see \cref{def:process:mining}).

For the sake of readability, we provide in \cref{table:confine:keys} an overview of the 
ephemeral keys generated by the CONFINE components and referred to in the remainder of this section.
In the following, we describe each protocol phase in detail. 

\begin{figure}[t]
	\subfloat[][Initialization]{\includegraphics[height=0.38\textheight]{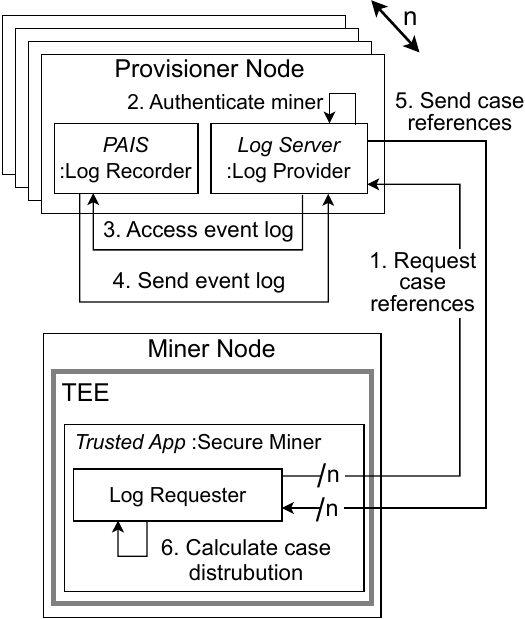}\label{fig:init}}\hfill
	\subfloat[][Remote attestation]{\includegraphics[height=0.38\textheight]{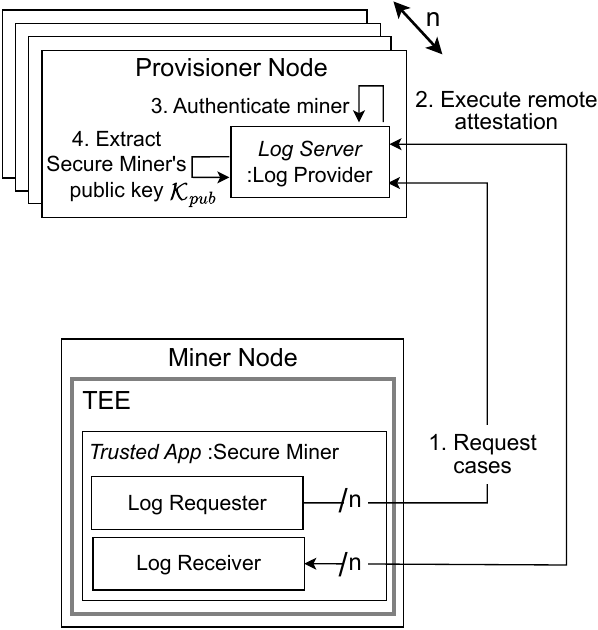}\label{fig:attestation}}\hfill
	\caption{Unfolding example for the initialization and remote attestation phases of the CONFINE protocol}
	\label{fig:workflow:example}
\end{figure}

\subsubsection{Initialization}\label{deployment:protocol:init} The protocol starts with the initialization stage whose objective is to inform the \Compo{Secure Miner} about the distribution of cases in the  
log partitions (see \cref{def:partition}) of the \Compo{Provisioner}s. At the onset of this stage, the \Compo{Log Requester} within the \Compo{Secure Miner} issues $n$ requests, one per \Compo{Log Server} component, 
to retrieve the list of case references they record (step 1 in \cref{fig:init} and (\crefalgln[alg:secm]{alg:secm:CaseRefsReq:call}). This data request contains an identity proof associated with the miner organization, which the \Compo{Provisioner}s evaluate in the context of their authorization policies over their log partitions. Following miner authentication (step 2), each \Compo{Log Server} retrieves the local log partition from the \Compo{PAIS} (steps 3 and 4) and subsequently responds to the \Compo{Log Requester} by providing the list of iids {\CIdS} within the owned log partition (step 5 and \crefalgln[alg:lprv]{alg:lprv:sendRefs}). After collecting these $n$ responses (\crefalgln[alg:secm]{alg:secm:respcollection}), the \Compo{Log Requester} delineates the distribution of {\CIdS}. In our motivating scenario, by the conclusion of the initialization, the \Compo{Secure Miner} gains knowledge that Bob's case, with {\CId} 711 and synthesized in the traces $T^H_{711}$ and $T^C_{711}$, is exclusively retained by the \Actor{Hospital} and the \Actor{Specialized Clinic}. In contrast, the trace of Alice's case having {\CId} 312, denoted as $T^H_{312}$, $T^C_{312}$, and $T^S_{312}$, is scattered across all three organizations.

%
In this phase, we solely rely on standard authentication of the \Compo{Secure Miner}'s organization since obtaining the list of case identifiers (\CIdS) typically does not reveal log partitions' sensitive data (it might be a random number or a hashcode as well, for instance).  Our design aims to perform the TEE attestation when the \Compo{Secure Miner} begins processing sensitive data, hence in the subsequent phase. 

%
%
\begin{algorithm2e}[tb]
	\input{algorithms/secureMinerJ}
	\caption{$\Confine_{\SecM}$, Secure Miner's behavior in CONFINE.}
	\label{alg:secm}
\end{algorithm2e} 
\begin{algorithm2e}[tb]
	\input{algorithms/mergeAndStoreAlg}
	\caption{The \Call{merge} function employed by the Secure Miner in $\Confine_{\SecM}$}
	\label{alg:merge}
\end{algorithm2e} 

\begin{algorithm2e}[tb]
	\input{algorithms/provisioningJ}
	\caption{$\Confine_{\LPrv}$, Provisioner's behavior in CONFINE.}
	\label{alg:lprv}
\end{algorithm2e} 
%
\begin{algorithm2e}[tb]
	\input{algorithms/segmentEventLog}

	\caption{The \Call{segmentEventLog} function of the Provisioner in $\Confine_{\LPrv}$.}
	\label{alg:segment}
\end{algorithm2e}

\begin{table}[b]
	\resizebox{\linewidth}{!}{%
		\centering
		\begin{tabular}{ccl}\toprule
			\textbf{Key name} & \textbf{Generated by} & \textbf{Objective} \\ \midrule
			  {\Kpub} & \multirow{2}{*}{\makecell[l]{\Compo{Secure Miner}}} & Used by the \Compo{Provisioner}s to encrypt the key {\Ksym} during the segment transmission\\ 
			{\Kpriv} &  & Used by the \Compo{Secure Miner} to decrypt the key {\Ksym} within the \Compo{TEE}  \\ \midrule
			{\Ksym} & \multirow{1}{*}{\makecell[l]{\Compo{Provisioner}}} & Used by the \Compo{Provisioner} to encrypt segments transmitted to the \Compo{Secure Miner}\\ 
			\bottomrule
		\end{tabular}
		\caption{Ephemeral encryption keys generated during the a session of the CONFINE protocol}
        \label{table:confine:keys}
        }%
\end{table}

\subsubsection{Remote attestation}\label{deployment:protocol:attestation} 
The remote attestation stage has a threefold objective: 
\begin{inparaenum}[\itshape(i)\upshape]
    \item to furnish provisioners with compelling evidence that the data request for a log partition originates from a \Compo{Trusted App} running within a TEE;
    \item to confirm the specific nature of the \Compo{Trusted App} as an authentic \Compo{Secure Miner} software entity;
    \item to authorize the organization of the \Compo{Secure Miner} application.
\end{inparaenum}

This phase is triggered when the \Compo{Log Requester} sends a new case request to the \Compo{Log Server} (step 1 in \cref{fig:attestation} and \crefalgln[alg:secm]{alg:secm:CasesReq:call}), specifying:
\begin{inparaenum}[\itshape(i)\upshape]
    \item the segment size (henceforth, \SegSize), and
    \item the set of the requested case \CIdS.
\end{inparaenum}
Both parameters will be used in the subsequent \textit{data transmission} phase.
Each of the $n$ \Compo{Log Server}s initiates a \Compo{TEE} remote attestation with the \Compo{Log Receiver} within the \Compo{Secure Miner} (step 2), as outlined in \cref{sec:background:tee}. During this procedure, the \Compo{Log Receiver} leverages the customizable data field of the attestation evidence (the inclusion of custom data is discussed in \cref{sec:background:tee}) to embed:
\begin{inparaenum}[\itshape(i)\upshape]
\item an identity proof of the \Compo{Secure Miner}'s organization, and
\item an ephemeral public key {\Kpub}, generated at runtime by the \Compo{Log Receiver} for the current protocol session and paired with a private key {\Kpriv}, which remains sealed within the \Compo{TEE}.
\end{inparaenum}

Upon successful completion of the remote attestation and verification of the \Compo{Secure Miner}'s workload, the $n$ \Compo{Log Server}s inspect the content of the custom data field of the validated attestation evidence. Using the embedded identity evidence, they authenticate the miner organization and, based on this assessment, render a final trust decision (step 3). If authentication succeeds, the \Compo{Log Server}s extract {\Kpub} from the validated evidence (step 4), after which the protocol proceeds to the data transmission phase.

\subsubsection{Data transmission}\label{deployment:protocol:transmission} Once the trusted nature of the \Compo{Secure Miner} is verified, the \Compo{Log Server}s proceed with the transmission of their cases. To accomplish this, each \Compo{Log Server} retrieves the log partition from the \Compo{PAIS} (steps 1 and 2 of \cref{fig:transmission}), and filters it according to the {\CIdS} specified by the miner. Given the constrained workload capacity of the \Compo{TEE}, it is imperative for \Compo{Log Server}s to split the filtered log partition into distinct segments (as per \cref{def:segment}).  
\begin{figure}[t]
	\subfloat[][Data transmission]{\includegraphics[width=0.4\linewidth]{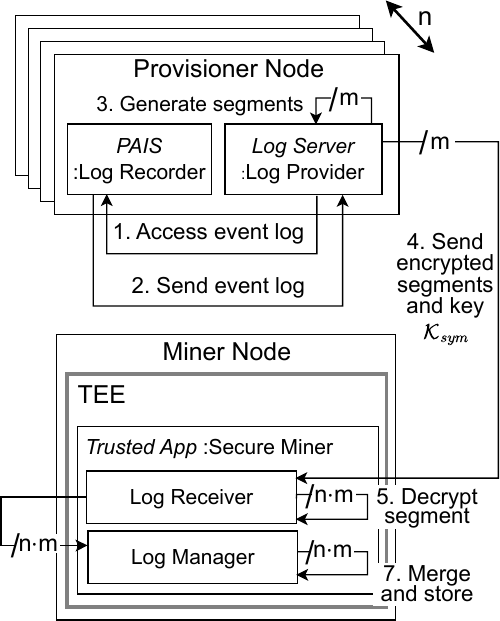}\label{fig:transmission}}\hfill
	\subfloat[][Computation]{\includegraphics[width=0.42\linewidth]{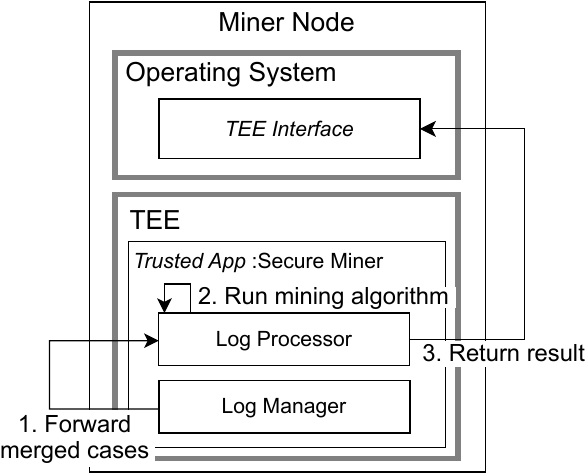}\label{fig:computation}}\hfill	
	\caption{Unfolding example for the data transmission and computation phases of the CONFINE protocol}
	\label{fig:workflow}
\end{figure}
Consequently, each \Compo{Log Server} generates $m$ log segments comprising a variable count of entire cases (step 3 and \crefalgln[alg:lprv]{alg:lprv:logaccess}). The cumulative size of these segments is governed by the {\SegSize} parameter specified by the miner in the initial request (as in step 1 of \cref{fig:attestation}). In our motivating scenario, the \Compo{Log Server} of the \Actor{Hospital} may structure the segmentation such that $T^H_{312}$ and $T^H_{711}$ reside within the same segment, whereas the \Actor{Specialized clinic} might have $T^S_{312}$ and $T^S_{711}$ in separate segments. In \cref{alg:segment}, we propose the specification of one of the possible segmentation strategies. Our pseudocode takes as input a log partition $\EvtL_{\LPrv}$, a set of instance identifiers $\CIdS$, and a segment size $\SegSize$ and it produces a set of $m$ segments $\Segm$ of $\EvtL_\LPrv$.
%
After the segmentation, the $n$ \Compo{Log Server}s generate a random symmetric encryption key, which we call {\Ksym} (see \cref{table:confine:keys}). Therefore, each \Compo{Log Server} encrypts its $m$ segments with {\Ksym} and transmits the encrypted data to the \Compo{Log Receiver} (step 4 and \crefalgln[alg:lprv]{alg:lprv:sendSegs}). Each transmission includes a cryptographic proof of the \Compo{Provisioner}'s identity, computed over the forwarded segments, along with the ciphertext obtained by encrypting {\Ksym} with the public key {\Kpub} found within the attestation evidence of the previous phase. Subsequently, the \Compo{Log Receiver} collects the $n \times m$ transmissions in a queue and processes them sequentially. It first uses {\Kpriv} to decrypt {\Ksym}, then uses the latter to decrypt the received segment (step 5).
Subsequently, the \Compo{Log Receiver} forwards the cases contained in the segment to the \Compo{Log Manager} (step 6 and \crefalgln[alg:secm]{alg:secm:mergeAndStore}). To reconstruct the process instance, cases 
with the same {\CId} must be merged by the \Compo{Log Manager} resulting in a single trace (e.g., $T_{312}$ for Alice case) comprehensive of all the events in the partial traces (e.g., $T^H_{312}$, $T^S_{312}$ and $T^C_{312}$ for Alice case) following the definition of \emph{safe merge} in \cref{def:merge}. In \cref{alg:merge}, we outline the pseudocode of the merge operation. This specification takes as input two event logs $\EvtL'$ and $\EvtL''$  and returns the merged event log $\EvtL = \EvtL' \Merge \EvtL''$. During this operation, the \Compo{Log Manager} applies a specific \textit{merging schema} (i.e., a rule specifying the attributes that link two cases during the merge \cite{claes2014merging}). In our illustrative scenario, the merging schema to combine the cases of Alice is contingent upon the linkage established through their {\CId} attributes, namely, \textit{`Case}' (in the \Actor{Hospital}'s log partition), \textit{`HospitalCaseID'} (in the \Actor{pharmaceutical company}'s log partition), and \textit{`TreatmentID'} (in the \Actor{specialized clinic}'s log partition). 
The outcomes arising from merging the cases within the processed segment are securely stored by the \Compo{Log Manager} in the \Compo{TEE}, until the subsequent computation phase.

\subsubsection{Computation}\label{deployment:protocol:computation} The \Compo{Secure Miner} requires all the \Compo{Provisioner}s to have delivered cases referring to the same {\CId} value (as per \cref{post:mandatoryattributes}). For example, when the \Actor{Hospital} and the other organizations have all delivered their information concerning case 312 to the \Compo{Trusted App}, the process instance associated with Alice becomes eligible for computation. Upon meeting this condition (\crefalgln[alg:secm]{alg:secm:computation}), the \Compo{Log Manager} forwards the case earmarked for computation to the \Compo{Log Processor} (step 1 in \cref{fig:computation} and \crefalgln[alg:secm]{alg:secm:yield1}, \crefalgln[alg:secm]{alg:secm:yield2}). This batch of cases may correspond either to the entire merged event log or to a subset thereof. In \cref{alg:secm}, we describe both the complete event log and the single case computation settings (\crefalgln[alg:secm]{alg:secm:yieldlog} and \crefalgln[alg:secm]{alg:secm:yieldcases}, respectively). The former entails a single computation routine, thus saving computation time but requiring a larger memory buffer in the \Compo{TEE}, whereas the latter necessitates multiple consecutive elaborations with a lower demand for space. Subsequently, the \Compo{Log Processor} proceeds to input the merged case/log, sorted according to their timestamp attributes, as per \cref{post:order:timestamp}, into the process mining algorithm {\PMFunc} (step 2), according to \cref{def:process:mining}. Notice that the above choice on the buffering of cases impacts the selection of {\PMFunc}.
If we elaborate subsequent batches, each containing a part of all merged cases, {\PMFunc} must support incremental processing, enriching the output as new batches come along. An example of this class of algorithms is the HeuristicsMiner~\cite{weijters2006process}.
Otherwise, incrementality is not required. Ultimately, the outcome of the computation is relayed by the \Compo{Log Processor} from the \Compo{TEE} to the \Compo{TEE Interface} running atop the \Compo{Operating System} of the \Compo{Miner Node} (step 3). The CONFINE protocol does not impose restrictions on the post-computational handling of results, once they leave the \Compo{TEE}. In our motivating scenario, the University and the National Institute of Statistics, serving as miners, disseminate the outcomes of computations, generating analyses that benefit the provisioners (though the original event log data are never revealed in clear). Furthermore, our protocol enables the potential for provisioners to have their proprietary \Compo{Secure Miner}, allowing them autonomous control over the computed results.

%% file: content/figures/protocolphases.tex
\begin{tikzpicture}[
	]
	\newcommand{\arrowbox}[4]{
		\begin{scope}[shift={(#1,0)}]
			\path[draw, thick, fill=white]
			(0,0) -- (3,0) -- (3.5,1.25) -- (3,2.5) -- (0,2.5) -- (0.5,1.25) -- cycle;
			\node[align=center] at (1.90,1.25) {#2};
			\node[align=center, text=blue] at (1.55,3.1) {\Cref{#3}};
			\node[text=blue] at (1.70,-0.6) {\Cref{#4}};
		\end{scope}
	}
	
	\arrowbox{0}    {Initialization}       {deployment:protocol:init}        {fig:init}
	\arrowbox{3.4}  {Remote\\Attestation}  {deployment:protocol:attestation} {fig:attestation}
	\arrowbox{6.8}  {Data\\Transmission}   {deployment:protocol:transmission}{fig:transmission}
	\arrowbox{10.2} {Computation}          {deployment:protocol:computation} {fig:computation}
	
	\draw [decorate,decoration={brace,amplitude=5pt,mirror,raise=4ex}]
	(0,-0.25) -- (10,-0.25) node[midway,yshift=-3.5em]{$\Confine_{\SecM}$ (\cref{alg:secm}), $\Confine_{\LPrv}$ (\cref{alg:lprv})};
\end{tikzpicture}

%% file: algorithms/secureMinerJ.tex
%
\tiny
\DontPrintSemicolon \SetAlgoVlined \SetInd{0.3em}{1em}
\setstretch{1.1}

\KwIn{%
	$\LPrvS = \{ \LPrv_1, \ldots, \LPrv_n \}$, the (references to) $n$ log provisioners.%
}
\Hyperparameters{
	$\SegSize$, the maximum size of the log segment to be transmitted by the log provisioners;
	\newline
	$\DoYieldCases$, indicating whether single cases (if assigned with \textit{true}) or only the full log at the end (if \textit{false})) should be yielded.
}
\KwData{%
	$\CIDMap : \CIdU \to 2^{\LPrvS}$, a map from case references $\CId \in \CIdU$ to the set of log provisioners in $\LPrvS$ ; \newline%
	$\LPrvMap : \LPrvS \to 2^{\CIdU}$, a map from log provisioners $\LPrv \in \LPrvS$ to the set of references to their cases in $\CIdU$ ; \newline%
	$\CStor : \CIdU \to \CasU$, a map from case references $\CId \in \CIdU$ to cases in $\CasU$.
}
\Implements{\Compo{SecureMiner}, {\Inst} {\SecM.}}
\Uses{\textit{AuthenticatedPerfectPointToPointLink}, {\Inst} {\AuthPtP.}}
\BlankLine
\UponEv(\tcp*[f]{The \Compo{Log Requester} of {\SecM} starts the CONFINE protocol -- \textit{initialization} phase in \cref{fig:init})}){$\langle \SecM, \Mesg{Init} \: |\: \LPrvS, \SegSize, \DoYieldCases\rangle$ \label{alg:secm:init}}{
	\ForEach(\tcp*[f]{For every \Compo{Provisioner} $\LPrv$}){$\LPrv \in \LPrvS$}{%
		\Trigger{$\langle \AuthPtP, \Mesg{Send} \gg \LPrv\ \:|\: \Mesg{CasesRefReq} \rangle$\label{alg:secm:CaseRefsReq:call}} \tcp*{Request \LPrv's case references (see \crefalgln[alg:lprv]{alg:lprv:CasesRefReq})}
	}
	\Upon(\tcp*[f]{Once all \Compo{Provisioner}s have answered with their case references}){$|\LPrvS| = |\textrm{dom}(\LPrvMap)|$\label{alg:secm:respcollection}}{%
		\lForEach
		{$\LPrv \in \LPrvS$}{
			\Trigger{$\langle \AuthPtP, \Mesg{Send} \gg {\LPrv}\:|\:\Mesg{CasesReq},\SegSize, \LPrvMap[\LPrv]\rangle$\label{alg:secm:CasesReq:call}}
			\tcp*[f]{Request their cases via {\AuthPtP} (see \crefalgln[alg:lprv]{alg:lprv:CasesReq})}
		}
	}
}

\UponEv(\tcp*[f]{{\SecM}'s \Compo{Log Requester} gets {\LPrv}'s case references via {\AuthPtP} (\crefalgln[alg:lprv]{alg:lprv:sendRefs})}){$\langle \AuthPtP, \Mesg{Deliver} \ll \LPrv \:|\: \Mesg{CasesRefRes}, \CIdS\rangle$ \SuchThat{$\LPrv \in \LPrvS$} \label{alg:secm:CaseRefsRes}}{
	\ForEach(\tcp*[f]{For every received case reference $\CId$ in {\CIdS}}){$\CId \in \CIdS$}{%
		$\CIDMap[\CId] \gets \CIDMap[\CId] \cup \{\LPrv\} $ \tcp*{Add $\LPrv$ to the set of provisioners for case $\CId$ in $\CIDMap$}
	}
	$\LPrvMap[\LPrv] \gets \LPrvMap[\LPrv] \cup \CIdS $ \tcp*{Register the references of the cases provided by $\LPrv$ in $\LPrvMap$}
}

\UponEv(\tcp*[f]{$\SecM$'s \Compo{Log Receiver} gets segment $\Segm = \left( \EqCls{\Segm}, \prec_\Segm \!\! \right)$ from {\LPrv} (\crefalgln[alg:lprv]{alg:lprv:sendSegs})}){$\langle \AuthPtP, \Mesg{Deliver} \ll \LPrv \:|\: \Mesg{CasesRes}, \Segm \rangle$ \SuchThat{$\LPrv \in \LPrvS$} \label{alg:secm:CasesRes} }{
	\ForEach(\tcp*[f]{For every $\Evt_{\Segm}$ in the delivered segment $\Segm$, each associated with an {\CId} -- see the \textit{data transmission} phase in \cref{fig:transmission} \label{alg:secm:visitevents}}){$\Evt_\Segm \in \EqCls{\Segm}$}{%
		\If(\tcp*[f]{If {\LPrv} has declared the ownership of {\CId} (see \crefalgln{alg:secm:CaseRefsRes} \label{alg:secm:checkiid})}){$\CIdF(\Evt_\Segm) \in \LPrvMap[\LPrv]$}  {%
			$\LPrvMap[\LPrv] \gets \LPrvMap[\LPrv] \,\setminus\, \{\CId\}$  \tcp*[r]{\raggedright Remove {\CId} from the set of case references to be provided by {\LPrv} \label{alg:secm:removecid}}
			$\CIDMap[\CId] \gets \CIDMap[\CId] \,\setminus\, \{\LPrv\}$ \tcp*[r]{Remove $\LPrv$ from the set of $\CId$ provisioners \label{alg:secm:removeprov}}
			$\CasP_\CId = \left( \EqCls{\Case_\LPrv}, \prec_{\Case_\LPrv} \right) \gets \CasF\left( \left( \EqCls{\Segm}, \prec_\Segm \right), \CId \right)$ 
			\tcp*[r]{\raggedright Collect in $\CasP_\CId$ the case extracted from the partition with case reference $\CId$ (see \cref{def:case}) \label{alg:secm:buildcase}}
			$\CStor[\CId] \gets \SecM.\Compo{LogManager}.\Call{merge}\left( \CStor[\CId], \CasP_\CId \right)$ \tcp*[r]{Merge the case and store the result in {\CStor}\label{alg:secm:mergeAndStore} -- see \cref{alg:merge}}
			$\EqCls{\Segm} \gets \EqCls{\Segm} \,\setminus\, \EqCls{\Case_\LPrv}; \quad \prec_\Segm \gets \prec_\Segm \,\setminus\, \prec_{\Case_\LPrv} $ \tcp*[r]{Remove the events and order relation's tuples of case $\CasP_\CId$ from segment $\Segm$}
		}
	}
}

\Upon(\tcp*[f]{When all the provisioners have delivered some $\CId$} to \SecM's \Compo{Log Manager}){$\CIDMap[\CId]= \emptyset$ for some $\CId \in \textrm{dom}(\CIDMap)$\label{alg:secm:computation}} {%
	$\textrm{dom}(\CIDMap) \gets \textrm{dom}(\CIDMap) \,\setminus\, \{\CId\}$\tcp*[r]{Remove $\CId$ from the domain of cases which still needs to be processed}
	\If(\tcp*[r]{If $\DoYieldCases$ flag is set to true}){\DoYieldCases\label{alg:secm:yieldcases}}{
		\Yield{\CStor[\CId] \To{\SecM.\Compo{LogProcessor}}} \tcp*[f]{Forward the case $\CId$ to the \Compo{Log Processor} of {\SecM} for mining -- \textit{computation} phase in \cref{fig:computation}\label{alg:secm:yield1}}}
	\Else(\tcp*[f]{If $\DoYieldCases$ flag is set to false}){\label{alg:secm:yieldlog}
		\If(\tcp*[f]{If all the cases have been collected}){$\textrm{dom}(\CIDMap) = \{\}$}{%
			$\EvtL \gets \left( \{\}, \{\}\right)$\tcp*[r]{Initialize an empty event log $\EvtL$}
			\ForEach(\tcp*[f]{For each complete case $\Case$ stored in $\CStor$\dots}){$\Case \in \CStor$}{
				$\EvtL \gets \SecM.\Compo{LogManager}.\Call{merge}\left(\EvtL, \Case \right)$\tcp*[r]{Add the case $\Case$ to the event log $\EvtL$ via merge}
			}
			\Yield{\EvtL \To{\SecM.\Compo{LogProcessor}}} \tcp*[f]{Forward the event log $\EvtL$ to the \Compo{Log Processor} of {\SecM} for mining -- \textit{computation} phase in \cref{fig:computation}\label{alg:secm:yield2}}
		}
	}
}

%% file: algorithms/mergeAndStoreAlg.tex
\tiny
\DontPrintSemicolon \SetAlgoVlined \SetInd{0.3em}{1em}
\setstretch{1.1}

\KwIn{%
	$\EvtL' = \left( \EvtU', \preceq' \right)$, an event log; \quad
	$\EvtL'' = \left( \EvtU'', \preceq'' \right)$, another event log.
}
\KwOut{%
	The merge of $\EvtL'$ and $\EvtL''$
}
\Fn{\Call{merge}$\left( \EvtL', \EvtL'' \right)$}{
	$\EvtU \gets \left(\EvtU' \cup \EvtU''\right);  \quad \preceq \gets \left(\preceq' \cup \preceq''\right); \quad \EvtL \gets \left(\EvtU, \preceq \right)$ \label{alg:merge:init}\\
	\ForEach(\tcp*[f]{For each event in $\EvtL''$}){$\Evt'' \in \EvtU''$\label{alg:merge:cycle}}{
		$\preceq \gets \preceq \cup \left\{ \left( \Evt',\Evt'' \right): \Evt' \in \EvtU' \textrm{ and } \TimeF\left(\Evt'\right) \leqslant \TimeF\left(\Evt''\right) \right\}$\tcp*[f]{For every $\Evt'$ in $\EvtU'$ s.t.\ $\TimeF\left(\Evt'\right) \leqslant \TimeF\left(\Evt''\right)$ add to $\preceq$ that $\Evt'$ precedes $\Evt''$\label{alg:merge:e1prece2}}\\
		$\preceq \gets \preceq \cup \left\{ \left( \Evt'',\Evt' \right): \Evt'' \in \EvtU \textrm{ and } \TimeF\left(\Evt''\right) \leqslant \TimeF\left(\Evt'\right) \right\}$\tcp*[f]{For every $\Evt'$ in $\EvtU'$ s.t.\ $\TimeF\left(\Evt''\right) \leqslant \TimeF\left(\Evt'\right)$ add to $\preceq$ that $\Evt''$ precedes $\Evt'$\label{alg:merge:e2prece1}}\\
	}
	\Return{\EvtL}
}

%% file: algorithms/provisioningJ.tex
\tiny
\DontPrintSemicolon \SetAlgoVlined \SetInd{0.3em}{1em}
\setstretch{1.1}

\KwIn{%
	$\SecMS = \{ \SecM_1, \ldots, \SecM_s \}$, the (references to) $s$ miners.%
}

\Hyperparameters{
    /
}

\Implements{\Compo{Provisioner}, {\Inst} {\LPrv}.}
\Uses{\textit{AuthenticatedPerfectPointToPointLink}, {\Inst} {\AuthPtP}.}
\UponEv(\tcp*[f]{{\LPrv} receives the request for case references from {\SecM} (see \crefalgln[alg:secm]{alg:secm:CaseRefsReq:call})}){$\langle \AuthPtP, \Mesg{Deliver} \ll \SecM \:|\: \Mesg{CasesRefsReq} \rangle$ \label{alg:lprv:CasesRefReq}\SuchThat{$\SecM \in \SecMS$}}{%
            $\left( \EqCls{\LPrv}, \prec_\LPrv \!\! \right) \gets \LPrv.{\Compo{LogRecorder}}.\Call{accessEventLog()}$\tcp*[f]{Access the log partition via \Compo{Log Recorder}}\\
            $\CIdS \gets \left\{\CIdF(\Evt'_\LPrv):\Evt'_\LPrv \in \EqCls{\LPrv} \right\}$\tcp*[f]{Get the set of case references in the log partition}\\
            \Trigger{$\langle \AuthPtP, \Mesg{Send} \gg \SecM \:|\: \Mesg{CasesRefRes},\CIdS \rangle$} \tcp*[f]{{Send the case references to {\SecM} (see \crefalgln[alg:secm]{alg:secm:CaseRefsRes})\label{alg:lprv:sendRefs}}}
}
\UponEv(\tcp*[f]{{\LPrv} gets the case request from {\SecM} (see \crefalgln[alg:secm]{alg:secm:CasesReq:call})}){$\langle \AuthPtP, \Mesg{Deliver} \ll \SecM \:|\: \Mesg{CasesReq},\SegSize, \CIdS\rangle$ \label{alg:lprv:CasesReq}\SuchThat{$\SecM \in \SecMS$}}{%
        $\Rep \gets \SecM.\Compo{LogReceiver}.\Call{getAttestationReport}(\LPrv)$ \label{alg:lprv:getReport}\tcp*[f]{Get the attestation report of {\SecM} -- \textit{remote attestation} phase in \cref{fig:attestation}}\\
	\If(\tcp*[f]{Execute the remote attestation of the report via \Compo{Log Provider}}){\LPrv.\Compo{LogProvider}.\Call{executeRemoteAttestation}(\Rep) is successful\label{alg:lprv:attestation}}{
            $\left( \EqCls{\LPrv}, \prec_\LPrv \!\! \right) \gets \LPrv.{\Compo{LogRecorder}}.\Call{accessEventLog()}$\tcp*[f]{Access the log partition via \Compo{Log Recorder}}\\
		$\{\Segm_1, \ldots, \Segm_m\} \gets \LPrv.\Compo{LogProvider}.\Call{segmentEventLog}\left(\left( \EqCls{\LPrv}, \prec_\LPrv \!\! \right),\CIdS,\SegSize \right)\label{alg:lprv:logaccess}$\tcp*[f]{Segment the event log -- see \cref{alg:segment}} \\
		\ForEach(\tcp*[f]{For every segment $\Segm_i$}){$i \in \{1, \ldots, m\}$}{%
			\Trigger{$\langle \AuthPtP, \Mesg{Send} \gg \SecM \:|\: \Mesg{CasesRes},\Segm_i \rangle$ } \tcp*[f]{Send the segment $\Segm_i$ to {\SecM} (see \crefalgln[alg:secm]{alg:secm:CasesRes}) -- \textit{data transmission} phase in \cref{fig:transmission}} \label{alg:lprv:sendSegs}
		}
	}
}

%% file: algorithms/segmentEventLog.tex
\tiny
\tiny
\DontPrintSemicolon \SetAlgoVlined \SetInd{0.3em}{1em}
\setstretch{0.9}

\KwIn{%
$ \EvtL_{\LPrv} = \left( \EqCls{\LPrv}, \prec_\LPrv \!\! \right)$, the log partition of provisioner $\LPrv$;
$\CIdS$, a set of case references;
$\SegSize$, a segment size.

}

\KwOut{
$\{\Segm_1 \dots \Segm_m \}$ a set of $m$ segments
}

\Fn{\Call{segmentEventLog}$\left( \EvtL_{\LPrv}, \CIdS, \SegSize \right)$}{
$\SegSet \gets \{\}$\tcp*[f]{Initialize the segment set to return}\\
$\EqCls{\Segm} \gets \{\};\quad \prec_\Segm \gets \{\};\quad \Segm \gets \left(\EqCls{\Segm}, \prec_\Segm \right) $\tcp*[f]{Initialize an empty segment}\\
\ForEach(\tcp*[f]{For each case in the log partition}){$\CId \in \CIdS$}{
$\CasP_\CId = \left( \EqCls{\Case}, \prec_{\Case} \right) \gets \CasF\left( \left( \EqCls{\LPrv}, \prec_\LPrv \right), \CId \right)$ 
\tcp*[r]{\raggedright Collect in $\CasP_\CId$ the case extracted from the partition with case reference $\CId$}
   \If(\tcp*[f]{If $\CasP_\CId$ do not fit the remaining space of the segment $\EqCls{\Segm}$}){\Call{SizeOf}$\left(\Segm\right)$ + \Call{SizeOf}$\left(\CasP_\CId\right) > \SegSize$}{
            $\SegSet \gets \SegSet \cup \left\{ \Segm \right\}$\tcp*[f]{Put the segment $\left( \EqCls{\Segm}, \prec_\Segm \!\! \right)$ in \SegSet}\\
            $\EqCls{\Segm} \gets \{\};\quad \prec_\Segm \gets \{\};\quad \Segm \gets \left(\EqCls{\Segm}, \prec_\Segm \right) $\tcp*[f]{Empty the segment}
		}
	$\prec_\Segm \gets \prec_\Segm \cup \prec_\Case$ \tcp*[f]{Add the case's order rel.\ to the current segment}\\
	$\prec_\Segm \gets \prec_\Segm \cup \left(\prec_\LPrv \cap \left( \EqCls{\Case} \times \EqCls{\Segm} \; \cup \; \EqCls{\Segm} \times \EqCls{\Case} \right) \right)$ \tcp*[f]{Add the order rel.\ from the case's events to the current segment's ones and vice-versa} \\
	$\EqCls{\Segm} \gets \EqCls{\Segm} \cup \EqCls{\Case}$ \tcp*[f]{Add the case's events to the current segment}\\
}

} 

%% file: content/formalvalidation.tex
Our solution provides security guarantees that are built upon an architecture leveraging TEEs as its core, compounded with a protocol we specifically devise for data provisioning and treatment.
In \cref{sec:discussion:security}, we analyze how data secrecy threats are mitigated by the integration of TEE's features in our architecture.
For the latter, we conduct a formal analysis of our protocol for segmentation and merge of event data stemming from heterogeneous sources in \cref{sec:correctness}, Specifically, we focus on the safety and liveness of CONFINE in its data transmission phase by formal verification of the underpinning algorithms' soundness and completeness.

\subsection{TEE-based Security Risk Mitigation}
\label{sec:discussion:security}
\begin{figure}
    \centering
    \includegraphics[width=1\linewidth]{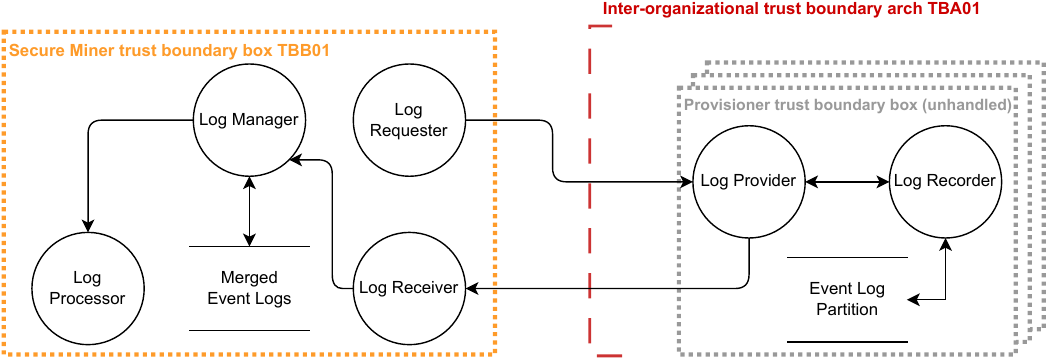}
    \caption{Data flow diagram and trust boundaries of our architecture}
    \label{fig:threatmodel:dataflow}
\end{figure}
\begin{table}[b]
\resizebox{\textwidth}{!}{%
    \centering
    \begin{tabular}{llll}\toprule
        \textbf{ID} & \textbf{Trust boundary}& \textbf{Type}  & \textbf{Threat description}\\ \midrule
        T01& \multirow{3}{*}{\makecell[l]{TBB01}}& Information disclosure  & The attacker accesses the \Compo{Secure Miner}'s memory location to leak the event logs\\ 
        T02& & Tampering  & The attacker meddles with the source code of the \Compo{Secure Miner} or its event log data\\ 
        T03& & Elevation of privileges  & The attacker gains the rights to run in the same environment of the \Compo{Secure Miner}\\\midrule 
        T04& \multirow{2}{*}{\makecell[l]{TBA01}}& Spoofing  & The attacker impersonates a \Compo{Secure Miner} to gain access to the \Compo{Provisioner}'s log\\ 
        T05& & Information disclosure & The attacker sniffs the data exchanged between the \Compo{Provisioner} and the \Compo{Secure Miner}\\
        \bottomrule
    \end{tabular}
    \caption[TEE-related threats in CONFINE]{TEE-related threats in CONFINE}
    \label{table:threatmodel:threats}
}
\end{table}
Here we pursue a twofold objective:
\begin{inparaenum}[\itshape(i)\upshape]
\item to provide a clear rationale for the integration of the \Compo{TEE} within our architecture, and
\item to examine the security guarantees that the \Compo{TEE}-related mechanisms offer in the context of our solution.
\end{inparaenum}
To this end, we build upon the STRIDE framework for the definition of a set of data-secrecy threats. Thereafter, we clarify how the \Compo{TEE}-based measures counteract these adversarial scenarios.

In \cref{fig:threatmodel:dataflow}, we present the Data Flow Diagram (DFD) which captures the components of the CONFINE architecture and the information exchange between them. Through this representation, we demarcate the \textit{trust boundaries} among the CONFINE components. In our DFD, a trust boundary box (the orange dotted square) groups entities with mutual trust. Differently, trust boundary arches (red dashed lines) mark the transitions of trust level on the occasion of data exchanges (the black arrows) between untrusted parties. We identify three trust boundaries: the \Compo{Secure Miner}'s trust boundary box (\textbf{TBB01}), the inter-organizational trust boundary arch (\textbf{TBA01}), and the \Compo{Provisioner}'s trust boundary box.

In \cref{table:threatmodel:threats}, we summarize the identified threats, indicating for each the trust boundary it targets and its corresponding STRIDE category: \textit{spoofing}, \textit{tampering}, \textit{repudiation}, \textit{information disclosure}, \textit{denial of service}, and \textit{elevation of privileges}. In our examination, we exclude the security considerations pertaining to \Compo{Provisioner}'s trust boundary boxes (the gray dotted square in \cref{fig:threatmodel:dataflow}), as menaces targeting event logs in their native organizational domain fall beyond the scope of this study.
In what follows, we examine how the \Compo{TEE} mechanisms embedded in our architecture mitigate each of these threats.

\noindent{\textbf{(T01) The attacker accesses the Secure Miner’s memory location to leak the event logs.}} 
As introduced in \cref{sec:deployment}, the \Compo{TEE} profile we consider in CONFINE provides hardware-based memory encryption. Therefore,  data and code belonging to the \Compo{Secure Miner} trusted app are managed by the CPU in a hardware-encrypted region of the main memory~\cite{costan2016intel}. The cryptographic key for decrypting this data randomly changes each time the machine undergoes a power cycle, and it is ingrained within the hardware of the CPU. Through a hardware-based access control mechanism, the processor interacts with reserved event log data solely through the execution of designated instructions delineated in the source code of the \Compo{Secure Miner}. In the event of a physical leakage attack, where attackers attempt to compromise sensitive information by gaining direct access to this private storage, they would only acquire ciphered data.

\noindent{\textbf{(T02) The attacker meddles with the source code of the Secure Miner or its event log data.}} Log partitions (defined in \cref{def:partition}) residing within the \Compo{TEE} can solely be manipulated by the source code of our \Compo{Secure Miner} trusted app. Any unauthorized modifications to the \Compo{Secure Miner}'s code after its deployment in memory would result in an altered measurement (i.e., the hash identifier for a trusted app source code). The processor possesses the capability to detect such alterations by comparing the measurement calculated on the tampered \Compo{Trusted App} with the one computed during the deployment phase. Through this mechanism, the \Compo{TEE} can forestall any malevolent modifications to event logs.

\noindent{\textbf{(T03) The attacker gains the rights to run in the same environment of the Secure Miner.}} In our approach, we consider application-level \Compo{TEE}s. Therefore, each \Compo{Secure Miner} application operates within its own dedicated \Compo{TEE} instance. This implies that reserved data can uniquely be accessed by a specific source code (in our case, the one of the \Compo{Secure Miner}). Any other trusted application software running on the same device is allocated a separate \Compo{TEE} instance, effectively rendering privilege escalation attacks infeasible.


\noindent{\textbf{(T04) The attacker impersonates a Secure Miner to gain access to the Provisioner’s log.}}
In CONFINE, we consider attestable \Compo{TEE}s. Therefore, during the attestation phase of our protocol, \Compo{Provisioner}s engage in the assessment of a hardware-signed attestation evidence to determine the authenticity of the sender, which purportedly represents a \Compo{Secure Miner} software entity operating within a \Compo{TEE} provided by an authorized organization. Depending on the specific TEE technology, \Compo{Provisioner}s may enlist an endorser~\cite{rfc9334}, serving as the fundamental authority (i.e., the CPU manufacturer in the majority of TEE instances). This authority is tasked with validating the report by confirming its signature with a valid attestation key. Attestation keys are securely embedded within the hardware and cannot be accessed by any entity other than the processor itself.
Once the \Compo{Provisioner} validates the report, it scrutinizes the measurement, which must correspond to the reference value associated with the \Compo{Secure Miner}'s source code. In case an unknown trusted app triggers the report generation, the CPU computes a software measurement that differs from the \Compo{Secure Miner}'s reference value, thus yielding a failure in the attestation process.

\noindent{\textbf{(T05) The attacker sniffs the Provisioner’s log sent to the Secure Miner.}} 
By leveraging attestation, the \Compo{TEE} enables the creation of a secure communication channel in which transmitted log partitions can only be decrypted inside the \Compo{TEE} itself. Following the successful attestation of the \Compo{Secure Miner}'s workload, the \Compo{Provisioner}s extract from the customizable section of the attestation evidence a public key that is randomly generated by the \Compo{Secure Miner} and whose corresponding private key is securely sealed within the \Compo{TEE}. This private key is inaccessible to any external actors, including users with administrative privileges over the mining machine. Using this key material, the \Compo{Provisioner} and the \Compo{Secure Miner} can establish the encrypted communication channel. Within this channel, segments are encrypted using an ephemeral symmetric key selected by the \Compo{Provisioner}. The symmetric key is transmitted in encrypted form using the aforementioned public key, together with the encrypted segments. Decryption of the symmetric key is only possible using the sealed private key within the \Compo{TEE}. Consequently, any adversary intercepting the transmission can access only ciphertext, thereby preserving the confidentiality of the transmitted information.

\subsection{Soundness and Completeness}\label{sec:correctness}

Hereafter, we formally prove the safety and liveness of our system~\cite{DBLP:journals/tse/Lamport77} in the execution of its data transmission phase by verifying soundness and completeness of the underlying algorithms. 
With the former, we ensure that if the protocol terminates, it returns a correct result. With the latter, we guarantee that the a result is always eventually returned.
Specifically, our objective is to prove that the ensemble of $\Confine_{\SecM}$ and $\Confine_{\LPrv}$ (\cref{alg:secm,alg:lprv}, respectively) always produces a safe merge as per \cref{def:merge} under the assumption of authenticated perfect point-to-point link used as a communication channel and honest \Compo{Provisioner}s. First, we prove soundness and completeness for our realization of the pivotal \Call{merge} function (\cref{alg:merge}). Thereupon, we prove the general statements pertaining to the protocol. Finally, we build upon the previous statements to prove the result's convergence with respect to a generic process mining algorithm (as per \cref{def:process:mining}). 

To begin with, we prove that our realization of the \Call{merge} function (\cref{alg:merge}) is \emph{safe}.
\begin{lemma}[Safety of \Call{merge}]\label{lma:merge:safety}

	Let $\EvtL' = \left( \EvtU', \preceq' \right)$ and
	$\EvtL'' = \left( \EvtU'', \preceq'' \right)$ be two event logs as per \cref{def:evt:log} whose total orders $\preceq'$ and $\preceq''$ are induced by the timestamp attribute as per \cref{post:order:timestamp}. The output of $\Call{merge}\left(\EvtL',\EvtL''\right)$ is always $\EvtL = (\EvtU, \preceq) = \EvtL' \Merge \EvtL''$ wherein $\Merge$ is safe as per \cref{def:merge}.

\end{lemma}

\begin{proof}
	\emph{(Base case)} Let $\EvtL'$ and $\EvtL''$ consist of one event each: $\EvtU'= \{\Evt'\}$ (hence $\Evt' \preceq' \Evt'$), $\EvtU''= \{\Evt''\}$ (hence $\Evt'' \preceq'' \Evt''$).
	\crefalgln[alg:merge]{alg:merge:init} initializes $\EvtU$ as $\{\Evt',\Evt''\}$ and $\preceq$ as $\{(\Evt',\Evt'),(\Evt'',\Evt'')\}$.
	It remains to establish whether $\Evt' \preceq \Evt''$ or vice-versa.
	The execution of \crefalgln[alg:merge]{alg:merge:e1prece2} and \crefalgln[alg:merge]{alg:merge:e2prece1} verify the statement by direct application of the isomorphism postulated in \cref{post:order:timestamp} since it employs the events' timestamps as the criterion for establishing the order.
	\emph{(Inductive case)} Let $\EvtU'$ and $\EvtU''$ be of cardinality $k'>1$ and $k''>1$, respectively. Let $e^{(k)} \in \EvtU''$ be the event fetched by the cycle on \crefalgln[alg:merge]{alg:merge:cycle} at the $k$-th iteration, with $1 < k \leq k''$. As per \cref{post:order:timestamp}, the isomorphism is guaranteed by \crefalgln[alg:merge]{alg:merge:e1prece2} and \crefalgln[alg:merge]{alg:merge:e2prece1}. Totality can only be violated if it was not preserved at the $(k-1)$-th iteration, which contradicts the induction over the base case. 
\end{proof}
Equipped with this result underpinning the basic operation backing the merging of event log partitions stemming from diverse \Compo{Provisioner}s, we can proceed to prove the general statement confirming the soundness and completeness of our protocol.

\begin{theorem}[Soundness and completeness]\label{theo:soundnessAndCompleteness}
	Let {\SecM} be a \Compo{Secure Miner} following $\Confine_{\SecM}$ (\cref{alg:secm}) with input references to \Compo{Provisioner}s $\LPrv_1, \ldots, \LPrv_n$, whereby the latter are honest, follow $\Confine_{\LPrv}$ (described in \cref{alg:lprv}), and provide log partitions $\EvtL_{\LPrv,1}, \ldots, \EvtL_{\LPrv,n}$ with $n \in \mathbb{N}$. Upon the completion of the data transmission phase, {\SecM} always (completeness) yields an event log $\EvtL = \EvtL_{\LPrv,1} \Merge \cdots \Merge \EvtL_{\LPrv,n}$ wherein $\Merge$ is safe (soundness).
\end{theorem}
\begin{proof}[Proof sketch.]
	The communication channel (henceforth {\AuthPtP}) is an Authenticated Point-to-point Perfect Link~\cite{Cachin.etal/2011:ReliableSecureDistributedProgramming}. 
	For the sake of brevity, we begin this proof sketch with the \Compo{Log Manager}'s assemblage of cases (\crefalgln[alg:secm]{alg:secm:mergeAndStore}) following \Compo{Log Receiver}'s receipt of segments (see \crefalgln[alg:secm]{alg:secm:CasesRes} and \crefalgln[alg:lprv]{alg:lprv:sendSegs}).
	We leave out the analysis of the preceding interactions (i.e., the initialization and remote attestation phases, the subsequent gathering of case references and polishing of internal data structures), since their soundness and completeness can be verified by inspection of the algorithms (see the code blocks starting at \crefalgln[alg:secm]{alg:secm:init}, \crefalgln[alg:lprv]{alg:lprv:CasesReq}, \crefalgln[alg:lprv]{alg:lprv:CasesRefReq}, and \crefalgln[alg:secm]{alg:secm:CaseRefsRes}) given the channel's guarantees, the \Compo{Provisioner}s' honesty, and the secure miner's authenticity.

	\noindent \emph{(Soundness).} Assume that every log partition $\LPrv_i$ consists of exactly one segment, $\Segm_i$, with $1 \leq i \leq n$. Let $\Segm_i$ contain (part of) a case identified by $\CId$, $\CasP_\CId$, and $\CStor[\CId]$ be the current copy of the case identified by $\CId$ in $\SecM$. Notice that the building of $\CStor[\CId]$ is based on \crefalgln[alg:secm]{alg:secm:buildcase}, wherein all events associated to $\CId$ in $\SecM$ are extracted and put in a total order relation that restricts the segment's total order to those events.
	The assemblage of cases is carried out by invoking \Call{merge} (\cref{alg:merge}) on $\CStor[\CId]$ and  $\CasP_\CId$ (\crefalgln[alg:secm]{alg:secm:mergeAndStore}). 
	By definition, a case is an event log~(see \cref{def:case}). Thus, $\CStor[\CId]$ and $\CasP_\CId$ can be considered as such, too. Therefore, \cref{lma:merge:safety} applies, and the merge of $\CStor[\CId]$ and  $\CasP_\CId$ is safe.
	Since a safe merge is associative and commutative (\cref{prop:safemerge:assoc:commut}), the number of events in the segment and the processing order thereof do not affect the correctness of the above statement. Also, this operation is repeated for every case in the segment as all events in it are visited (\crefalgln[alg:secm]{alg:secm:visitevents}) and every case is checked once (as guaranteed by \crefalgln[alg:secm]{alg:secm:checkiid}, \crefalgln[alg:secm]{alg:secm:removecid}, and \crefalgln[alg:secm]{alg:secm:removeprov}). Therefore, the statement holds regardless of the number of cases in the segment. 
	If a log partition consists of more than one segment, the \Compo{Secure Miner} reiterates the same above procedure for every such segment (\crefalgln[alg:secm]{alg:secm:CasesRes}), and does so in a mutual exclusive fashion, thus impeding faulty over-writings on $\CStor$ (\crefalgln[alg:secm]{alg:secm:mergeAndStore}). The same conclusion can be drawn extending the discourse to a setting with multiple \Compo{Provisioner}s delivering their own segments. Again, the processing order for segments does not influence the outcome due to the associativity and commutativity of the safe merge (see \cref{prop:safemerge:assoc:commut}) that the \Call{merge} function realizes (as per \cref{lma:merge:safety}).
	\\
	\emph{(Completeness).} Now, we show that, operating under our assumptions, no segment gets lost during the transmission phase (thus, blocking the protocol execution and undermining the completeness). Assuming every {\LPrv} is honest, they transmit all and only the segments of which their partitions consist following $\Confine_{\LPrv}$ (\cref{alg:lprv}). Since the communication channel resorts to an Authenticated Perfect Point-to-Point Perfect Link, all broadcast segments are delivered once, and no package is lost. Therefore, every event in a received segment (\crefalgln[alg:secm]{alg:secm:CasesRes}) is visited by the routine starting at \crefalgln[alg:secm]{alg:secm:visitevents}, hence all the cases therein. All merged cases are thus finally yielded with the procedure starting at \crefalgln[alg:secm]{alg:secm:computation}, thus completing the data transmission phase.
\end{proof}

In light of the above proof, we can finally prove the convergence of CONFINE. 
%

\begin{theorem}[Convergence]    \label{evaluation:correctness:convergence}
	Let $\EvtL$ be an event log such that $\EvtL = \EvtL_{\LPrv,1} \Merge \cdots \Merge \EvtL_{\LPrv,n}$ wherein $\Merge$ is safe as per \cref{def:merge}.
	Let $\EvtL' = \Confine_{\SecM}(\LPrv_1, \ldots, \LPrv_n)$ be the event log yielded as a result of the data transmission phase ($\Confine_{\SecM}$) by a \Compo{Secure Miner} {\SecM} following the protocol in \cref{alg:secm} with references to \Compo{Provisioner}s $\LPrv_1, \ldots, \LPrv_n$, whereby the latter follow the protocol described in \cref{alg:lprv} and provide logs $\EvtL_{\LPrv,1}, \ldots, \EvtL_{\LPrv,n}$ with $n \in \mathbb{N}$.
	Let $\PMFunc$ be a process mining function as per \cref{def:process:mining}. 
	The application of $\PMFunc$ on the event log $\EvtL$ and on $\EvtL'$ leads to the same result, i.e.,
	$\PMFunc(\EvtL) = \PMFunc(\EvtL')$.
\end{theorem}
\begin{proof}[Proof sketch.]
	By contradiction, suppose that  $\PMFunc(\EvtL) \neq \PMFunc(\EvtL')$. Given that the process mining algorithm \PMFunc is deterministic (as per \cref{def:process:mining}) and all its parameters are given, this would imply that $\EvtL$ is not correctly reconstructed during the execution of the protocol, i.e., $\Confine_{\SecM}(\LPrv_1, \ldots, \LPrv_n) \neq \EvtL$. In other words, {\SecM} computes $\EvtL_{\LPrv,1} \Merge \cdots \Merge \EvtL_{\LPrv,n} = \EvtL'$, with $\EvtL' \neq \EvtL$. We assume that each \Compo{Provisioner} $\LPrv_{i}$ behaves honestly and that the communication channel is an Authenticated Perfect Point-to-Point channel. Therefore, the incorrect reconstruction of $\EvtL'$ would contradict the merge safety property demonstrated in \cref{lma:merge:safety}, and more generally, it would violate the soundness of the protocol as established in \cref{theo:soundnessAndCompleteness}. 
\end{proof}

%% file: content/evaluation.tex
\begin{table}[bt]
	\caption{Event logs used for our experiments}
	\label{tab:testedlogs}
	\centering
	\input{content/tables/eventLogs}
\end{table}

In this section, we report on our empirical evaluation of CONFINE. First, we introduce our implementation in \cref{sec:implementation:details}. Then, we propose an experimental verification of \cref{evaluation:correctness:convergence}, empirically demonstrating the output convergence in \cref{sec:discussion:subsec:convergence}. Finally, we evaluate our implementation with runtime tests on memory usage to show that our approach determines low memory demand, compatible with the constrained capacity of TEEs (\cref{sec:discussion:subsec:convergence}).

\subsection{Implementation}
\label{sec:implementation:details}
We implemented the \texttt{Secure Miner} component as an Intel SGX~\cite{DBLP:journals/istr/BagherL23} trusted application, encoded in Go through the EGo framework.\footnote{\href{https://docs.edgeless.systems/ego}{\nolinkurl{https://docs.edgeless.systems/ego}.} Accessed: December 2, 2025.}
We opted for an Intel SGX TEE for our prototype because its features align with the TEE profile outlined in \cref{sec:background:tee} and offers a well-documented developer ecosystem. These characteristics made Intel SGX a practical choice for implementing our research prototype aimed at demonstrating the feasibility of CONFINE. However, we highlight that recent studies have reported vulnerabilities~\cite{DBLP:conf/sp/MurdockOGBGP20,DBLP:conf/uss/BulckMWGKPSWYS18}, and real-world deployments may therefore prefer more robust alternatives or incorporate tailored mitigation strategies to address these concerns~\cite{DBLP:journals/csur/FeiYDX21}.

We encoded the provisioner's \Compo{Log Server} in GO. 
The implementation of our \Compo{Secure Miner} component facilitates input reception from users via a console application (i.e., the \Compo{TEE Interface} in \cref{sec:deployment}). We store the descriptive information of the \Compo{Trusted App} (i.e., heap size, embedded files, and environment variables) in a JavaScript Object Notation (JSON) file. In our implementation, the \Compo{Secure Miner} app maintains references to provisioners' \Compo{Log Server}s in a dedicated JSON file, alongside the label of the {\CId} attribute in their respective log partitions. We embed in the \Compo{TEE} during deployment. We resort to a Transport Layer Security (TLS)~\citep{Thomas/2000:SSL-TLS} communication channel between miners and provisioners over the HyperText Transfer Protocol (HTTP) web protocol to secure the information exchange and authenticate the miner's organization. The control data structures {\CIDMap} and {\LPrvMap} (refer to \cref{alg:secm}) generated by the \Compo{Secure Miner} during the CONFINE protocol are formatted as JSON files.

We integrate two established process mining algorithms within our \Compo{Secure Miner} implementation: \textit{HeuristicsMiner}~\cite{weijters2006process} and \textit{DeclareConformance}~\cite{DBLP:conf/bpm/DonadelloRMS22} from the \textit{pm4py} library.%
\footnote{\href{github.com/process-intelligence-solutions/pm4py}{\nolinkurl{github.com/process-intelligence-solutions/pm4py}}. Accessed: December 2, 2025.}
In \cref{sec:background:processmining}, we provided an overview of their core mechanisms.
The selected algorithms are representative of process mining techniques along two complementary dimensions. First, \textit{HeuristicsMiner} and \textit{DeclareConformance} perform the core process mining tasks of process discovery and conformance checking, respectively. Second, they reflect the two predominant modeling paradigms in process mining, namely \textit{imperative} (as in the case of \textit{HeuristicsMiner}) and \textit{declarative} approaches (as in the case of \textit{DeclareConformance})~\cite{DBLP:conf/bpm/PichlerWZPMR11,DiCiccio.Montali/PMH2022:DeclarativeProcessMining}. Through the integration of these algorithms, we demonstrate that our implementation is applicable to distinct process mining tasks and accommodates different methodological paradigms.

Upon receiving each complete case (see \cref{def:case}), our \textit{HeuristicsMiner} implementation produces a causal net, which is subsequently transformed into the corresponding workflow net of the analyzed process. The \Compo{TEE} outputs the resulting models in the Petri Net Markup Language (PNML) format, a standard for representing workflow nets.%
\footnote{\href{https://www.pnml.org}{\nolinkurl{www.pnml.org}}. Accessed: December 2, 2025.}
Similarly, our \textit{DeclareConformance} implementation computes fitness scores, which are returned from the \Compo{TEE} in the form of JSON files.
%
For both algorithms, we provide the following two versions.
\begin{inparadesc}
	\item[Incremental:] In this version, the algorithm takes batches of the complete cases as an input and updates a partial result while data transmission is still ongoing;
	\item[Non-incremental:] The algorithm is triggered at the end of the data transmission, once all the cases are collected.
\end{inparadesc}
Our implementation of CONFINE, including the \textit{HeuristicsMiner} and the \textit{DeclareConformance} in Go, is openly accessible at the following URL: \href{https://github.com/Process-in-Chains/CONFINE/}{\nolinkurl{github.com/Process-in-Chains/CONFINE/}}.

\begin{figure}[t]
	\includegraphics[width=1\linewidth]{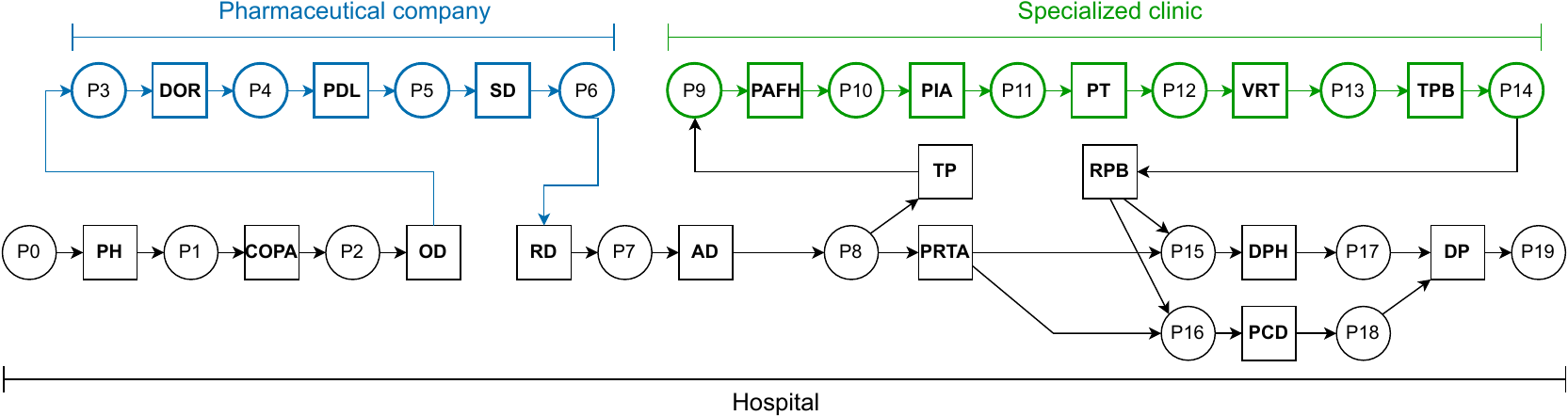}\label{fig:wfnet:d}
	\caption[HeuristicsMiner output]{%
    \emph{HeuristicsMiner} output in CONFINE%
    }
	\label{fig:wfnet}
\end{figure}
\subsection{Convergence}
\label{sec:discussion:subsec:convergence}
To experimentally validate \cref{evaluation:correctness:convergence}, we run a \emph{convergence} test. We created a synthetic event log consisting of \num{1000} cases (as per \cref{def:case}) of \num{14} events on average (see \cref{tab:testedlogs}) by simulating the inter-organizational process of  \cref{fig:BPMN_Healthcare}.%
\footnote{We generated the event log through BIMP (\href{https://bimp.cs.ut.ee/}{\nolinkurl{bimp.cs.ut.ee}}; accessed: December 2, 2025). We filtered the log by retaining the sole events that report on the completion of activities, and removed the start and end events of the \Actor{Pharmaceutical company} and \Actor{Specialized clinic}'s sub-processes.}
and we partitioned it into three sub-logs (one per involved organization), an excerpt of which is listed in~\cref{tab:trace}.
We run the stand-alone \textit{HeuristicsMiner} on the former and process the latter through our CONFINE toolchain.
As expected, the results converge and are depicted in \cref{fig:wfnet} in the form of a workflow net~\citep{Aalst/ICATPN97:VerificationofWfNs}. 
For the sake of clarity, we color activities recorded by the organizations following the scheme of~\cref{tab:testedlogs} (black for the \Actor{Hospital}, blue for the \Actor{Pharmaceutical company}, and green for the \Actor{Specialized clinic}). Additionally, we add labels to specify the organization involved in separate fragments of the workflow net.

\subsection{Performance Assessment}
\label{sec:evaluation:performanceAssessment}
In what follows, we 
report on our experiments conducted to put our tool implementation to the test. 
As discussed in~\cref{sec:deployment}, the availability of space in the dedicated TEE areas is subject to hardware limitations. Therefore, we focus on memory consumption since exceeding those limits could diminish the level of security guaranteed by TEEs.
We 
gauge the memory usage with synthetic and real-life event logs, to observe the trend during the enactment of our protocol.
In particular, we focus on runtime consumption (\cref{sec:evaluation:subsec:MemoryUsage}), the effect that the segment size hyperparameter discussed in \cref{sec:deployment:protocol} has on efficiency (\cref{sec:eval:segmentsize}), and scalability with respect to the event log size and the number of provisioners (\cref{sec:evaluation:scalability}).
The raw execution data alongside the scripts to plot the obtained results are openly accessible at \href{https://github.com/Process-in-Chains/CONFINE/tree/main/evaluation/}{\nolinkurl{github.com/Process-in-Chains/CONFINE/tree/main/evaluation/}}.

\begin{figure}[tbp]
\centering
\begin{subfigure}{0.49\linewidth}
  \centering
  \includegraphics[width=\linewidth]{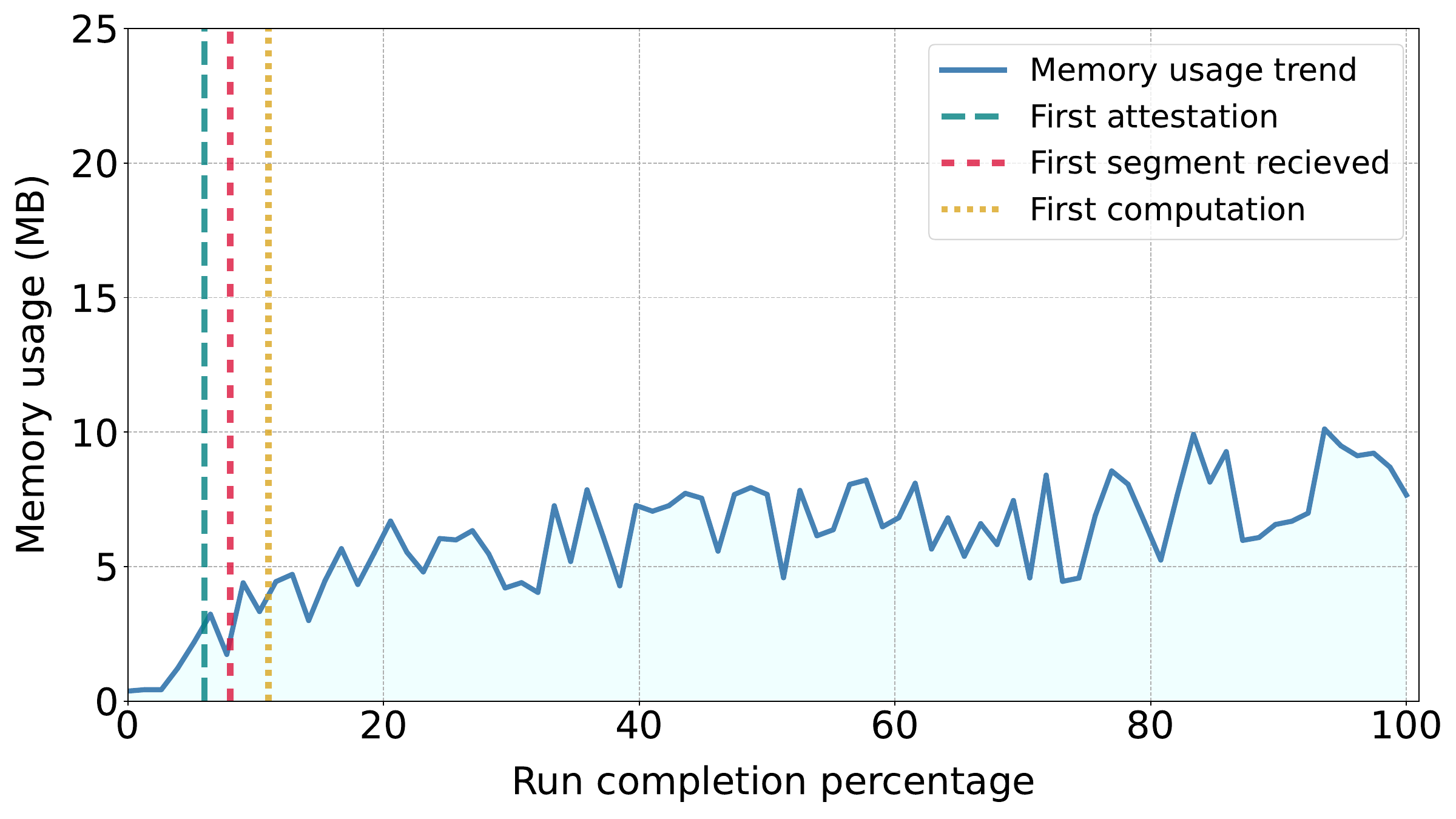}
  \caption{Memory usage with incremental HeuristicsMiner}
  \label{snr_a}
\end{subfigure}\hfill
\begin{subfigure}{0.49\linewidth}
  \centering
  \includegraphics[width=\linewidth]{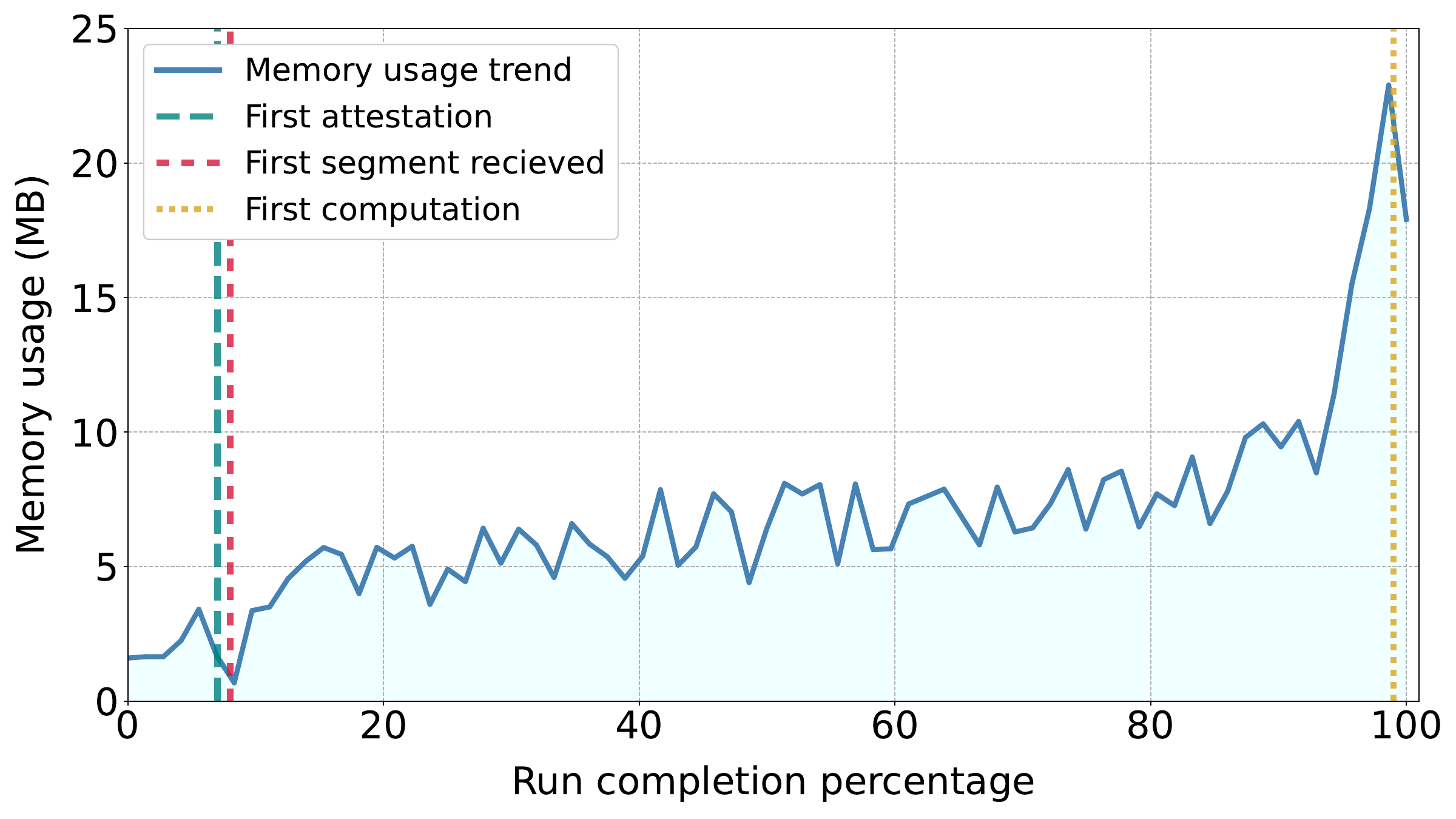}
  \caption{Memory usage with the non-incremental HeuristicsMiner}
  \label{snr_b}   
\end{subfigure}

\begin{subfigure}{0.49\linewidth}
  \centering
  \includegraphics[width=\linewidth]{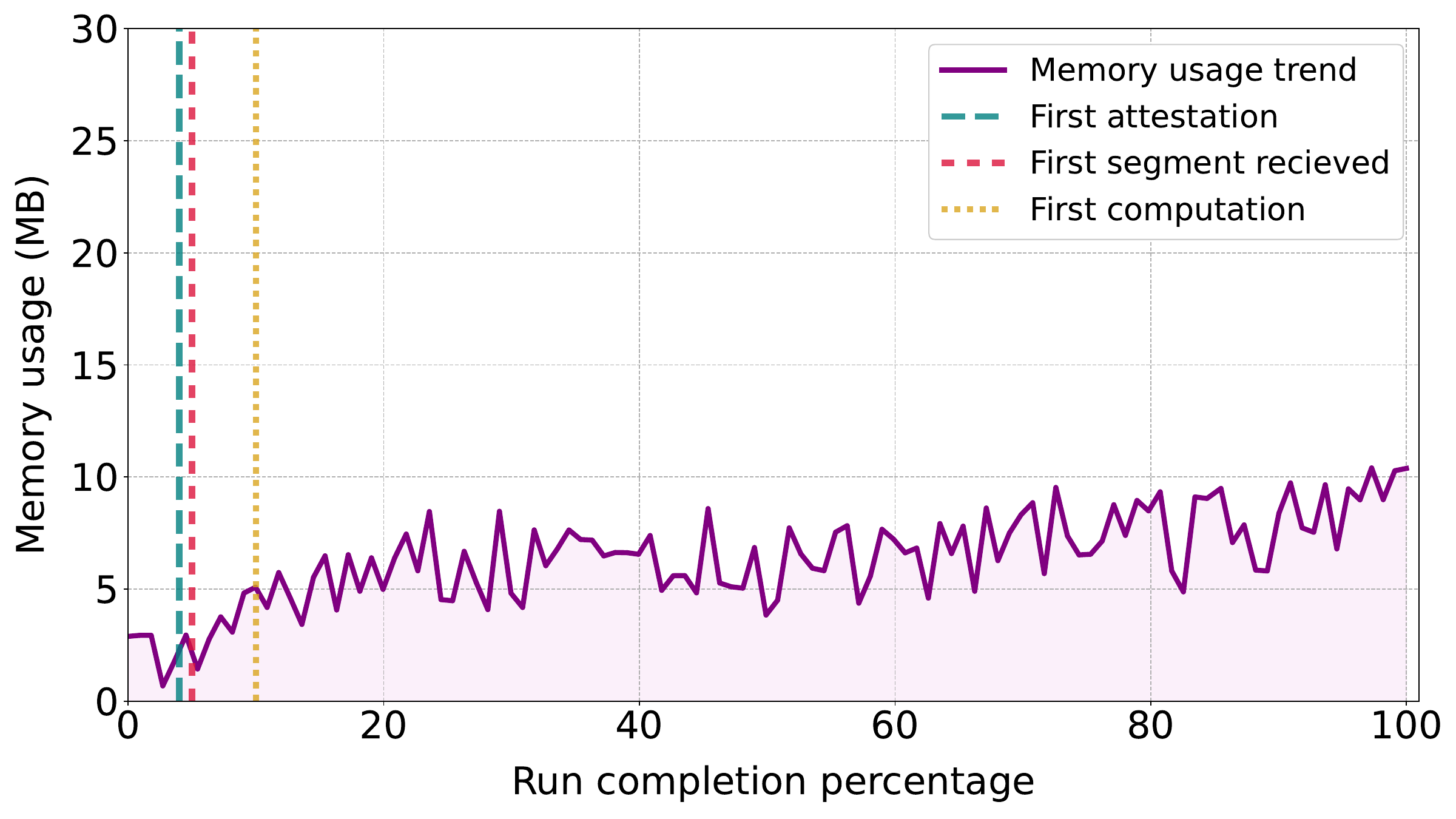}
  \caption{Memory usage with incremental DeclareConformance}
  \label{snr_c}
\end{subfigure}\hfill
\begin{subfigure}{0.49\linewidth}
  \centering
  \includegraphics[width=\linewidth]{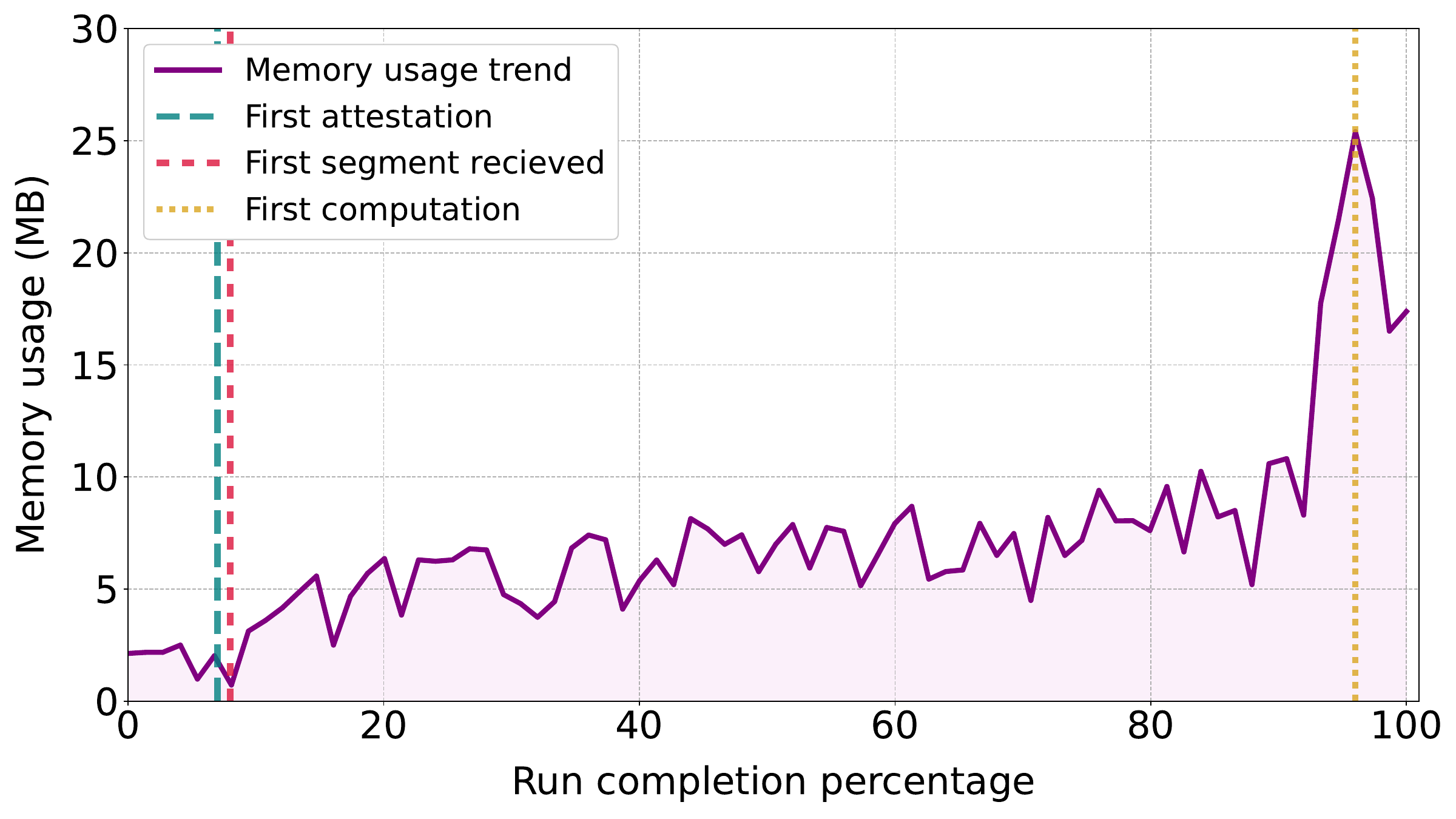}
  \caption{Memory usage with non-incremental DeclareConformance}
  \label{snr_d}   
\end{subfigure}

\begin{subfigure}{0.49\linewidth}   
  \centering      
  \includegraphics[width=\linewidth]{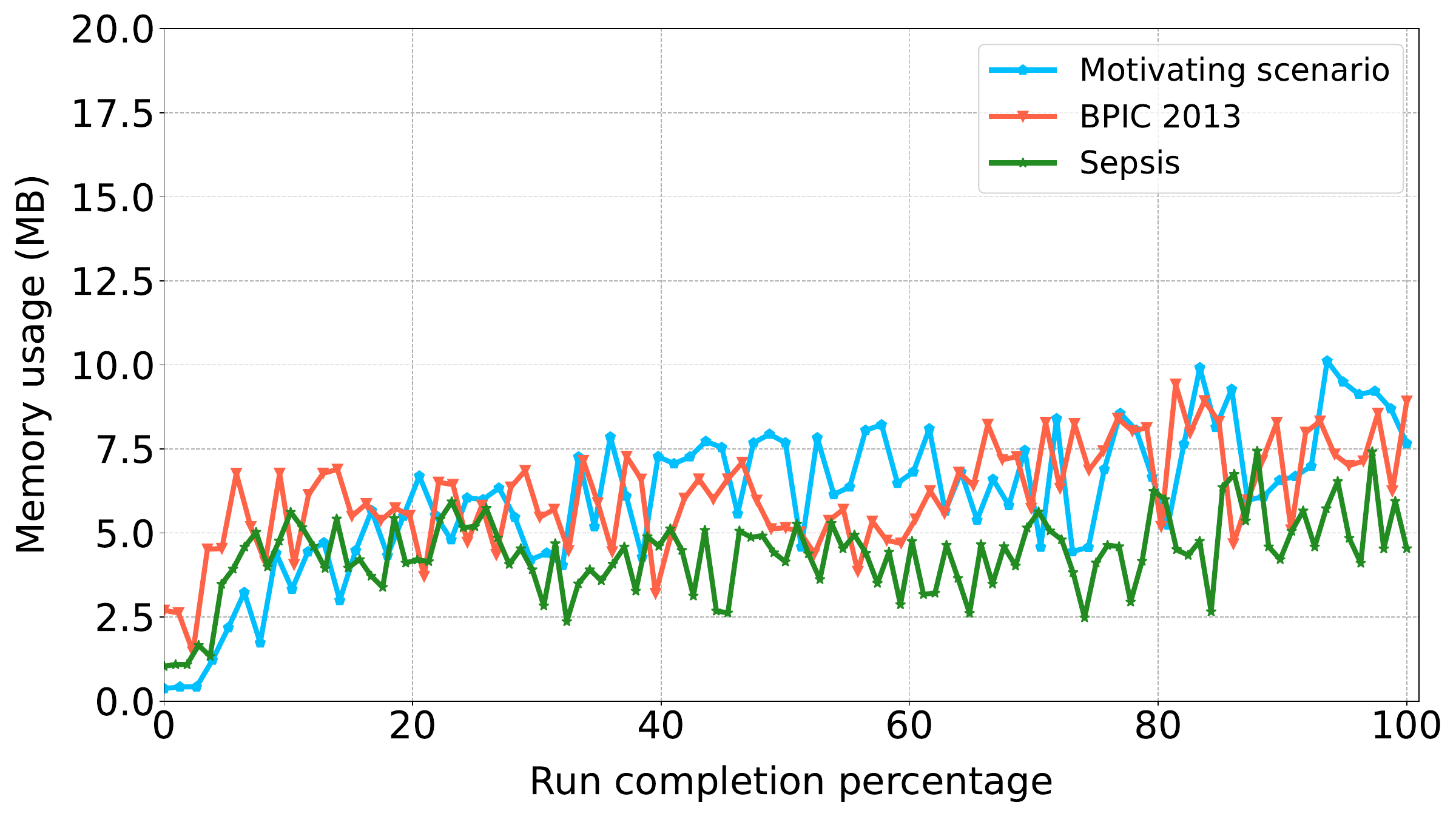}
  \caption{Memory usage with three event logs and the incremental HeuristicsMiner}
  \label{snr_e}
\end{subfigure}\hfill
\begin{subfigure}{0.49\linewidth}   
  \centering      
  \includegraphics[width=\linewidth]{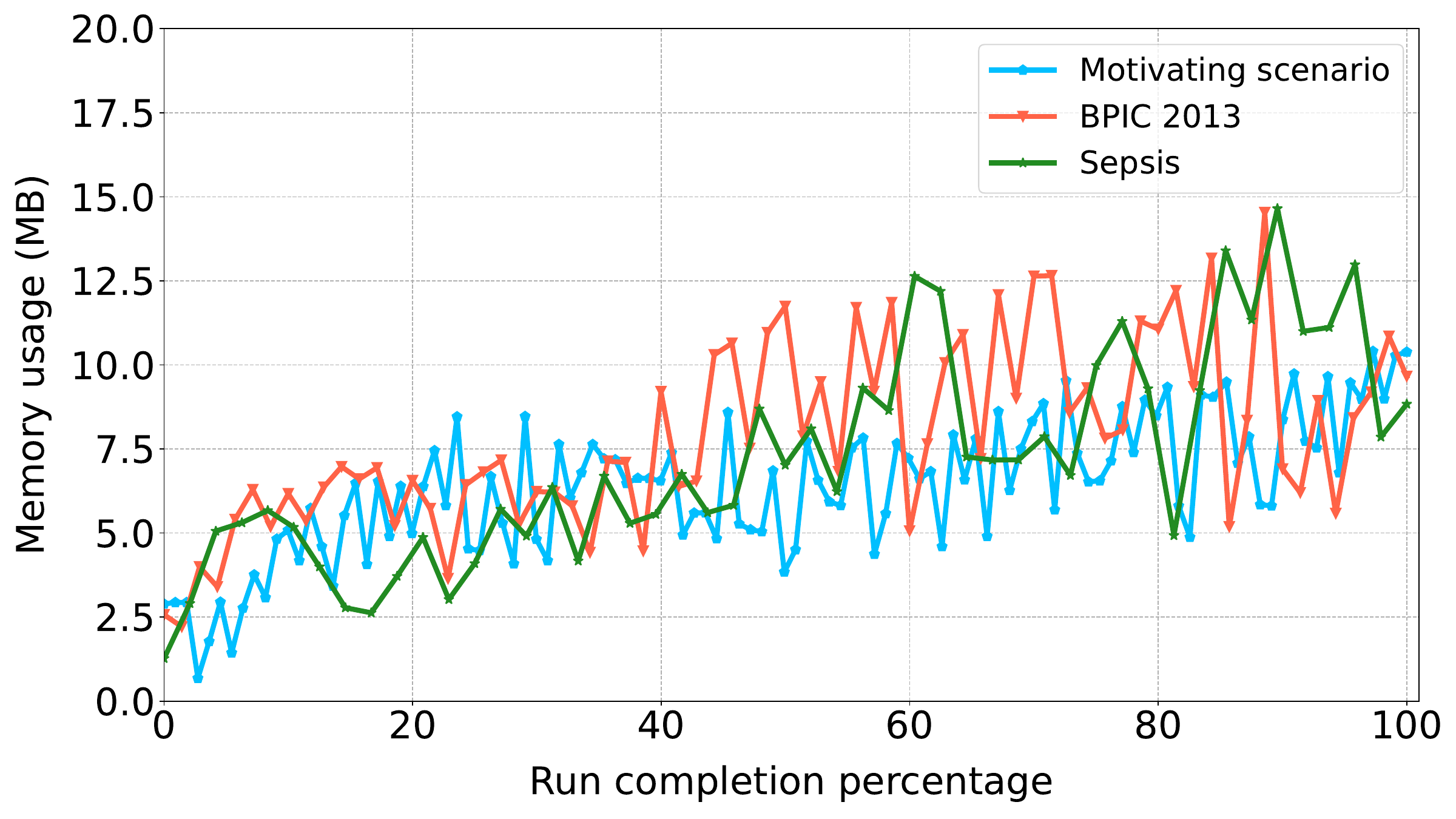}
  \caption{Memory usage with three event logs and the incremental DeclareConformance}
  \label{snr_f}
\end{subfigure}
\caption{Memory usage test results}
\label{fig:memtest}
\end{figure}
\subsubsection{Runtime analysis}\label{sec:evaluation:subsec:MemoryUsage}
\Cref{snr_a,snr_b,snr_c,snr_d} display the runtime memory utilization of our \Compo{Secure Miner} implementation (in MegaBytes) using the synthetic log derived from our motivating scenario as input.  \Cref{snr_a} illustrates the memory utilization trend with the integration of the incremental variant of the HeuristicsMiner, as described in \cref{deployment:protocol:computation}. Differently, \cref{snr_b} pictures the test outcome resulting from the execution of the protocol with the non-incremental variant of the algorithm. The dashed lines mark the starting points for the remote attestation, the data transmission, and the computation stages. We held the {\SegSize} constant at \num{100} KiloBytes. In both figures, we observe that after the initiation of the data transmission stage (the dashed red line), the memory utilization falls approximately between 5 and 10 MB. The major difference between the two trends concerns the computation phase. In \cref{snr_a}, the HeuristicsMiner starts computing cases when the data transmission is still ongoing. This ensures the memory utilization trend remains within a constant range until the end of the run. In \cref{snr_b}, the HeuristicsMiner is triggered after all the \Compo{Provisioners} have delivered their respective log partitions (as per \cref{def:partition}). Consequently, the final peak at the yellow dashed line signifies the computation of the entire event log, making the memory utilization diverge to a value of between 20 and 25 MB. This behavior is confirmed by \cref{snr_c,snr_d}, in which we consider the memory usage trend of the DeclareConformance algorithm variants. The tests indicate that the DeclareConformance demands a higher memory utilization than the HeuristicsMiner due to the involvement of the declarative process model in the conformance checking task.

To verify whether these behaviors are due to the synthetic nature of our simulation-based event log, we also gauge the runtime memory usage of two public real-world event logs (Sepsis~\citep{seps} and BPIC 2013~\citep{bpic2013}). 
The characteristics of the event logs are summarized in \cref{tab:testedlogs}.
Since those are \textit{intra-organizational} event logs, we 
split the contents to mimic an \textit{inter-organizational} context.
In particular, we separated the Sepsis event log based on the distinction between normal-care and intensive-care paths, as if they were conducted by two distinct organizations. Similarly, we processed the BPIC~2013 event log to sort it out into the three departments of the Volvo IT incident management system. 
\Cref{snr_e,snr_f} depict the results with the integration of the incremental variants of HeuristicsMiner and DeclareConformance. We observe that the BPIC 2013 event log processing demands more memory probably owing to its larger size. Conversely, the Sepsis event log turns out to entail the least expensive run.

\begin{figure}[t]
        \begin{subfigure}{0.49\linewidth}   
          \centering
          \includegraphics[width=\linewidth]{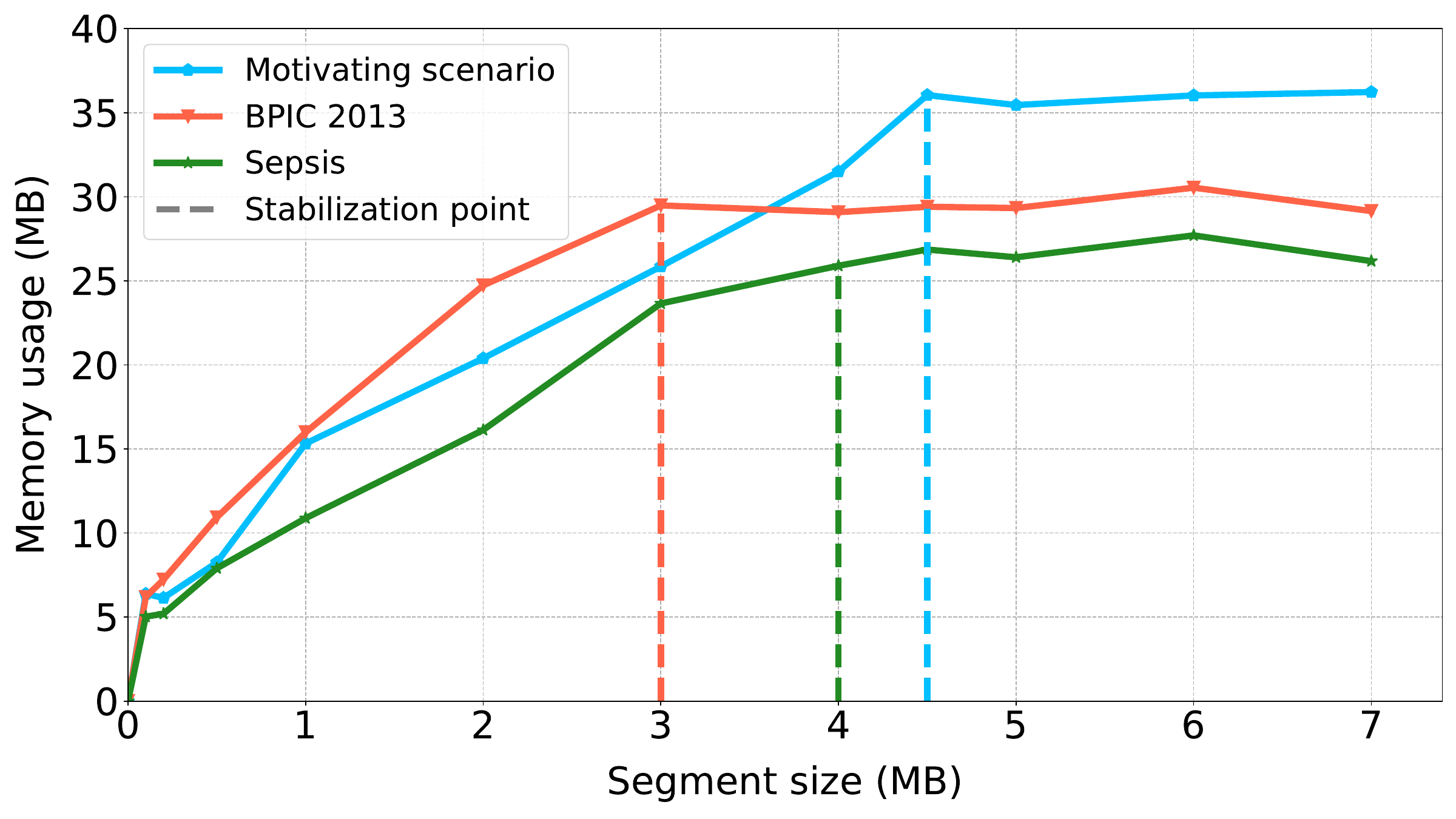}
          \caption{Segment size impact on memory usage}
          \label{fig:segmentsize}
        \end{subfigure}
	\begin{subfigure}{0.49\linewidth}   
		\centering
		\includegraphics[width=\linewidth]{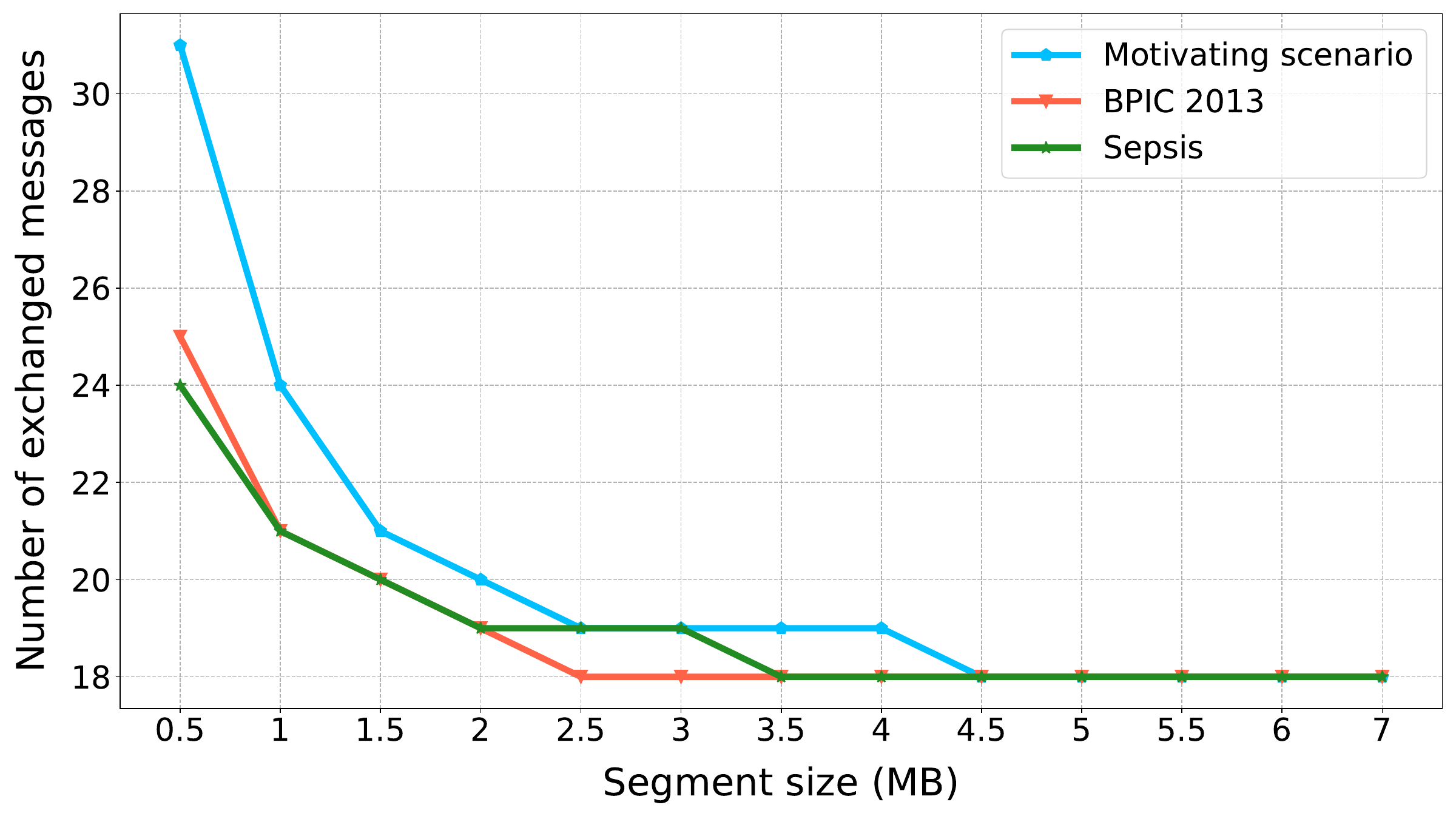}
		\caption{Segment size impact on exchanged messages overhead}
		\label{fig:messages_result}
	\end{subfigure}
	\caption{Segment size test result}
	\label{fig:segsizeresult}
\end{figure}

\subsubsection{Segment Size}\label{sec:eval:segmentsize}
To verify whether these trends are affected by the dimension of the exchanged data segments, we conducted additional tests to examine the trend of memory usage and the overhead of the exchanged message as the {\SegSize} varies with all the aforementioned event logs. Notably, the polylines displayed in \cref{fig:segmentsize} indicate a linear increment of memory occupation until a breakpoint is reached. After that, the memory in use is steady. These points, marked by vertical dashed lines, correspond 
to the {\SegSize} value that allows the provider's segments to be contained in a single data segment.

\Cref{fig:messages_result} depicts the relationship between segment size and the number of message exchanges through CONFINE. As the segment size increases, there is a decrease in the number of messages exchanged, indicating a reduction in communication overhead. This trend is consistent across different logs and segment sizes, with a stabilizing effect observed after the segment size of 4.5 MB suggesting an optimal range to minimize message exchanges.

\begin{figure}[t]
	\begin{subfigure}{0.49\linewidth}   
		\centering
		\includegraphics[width=\linewidth]{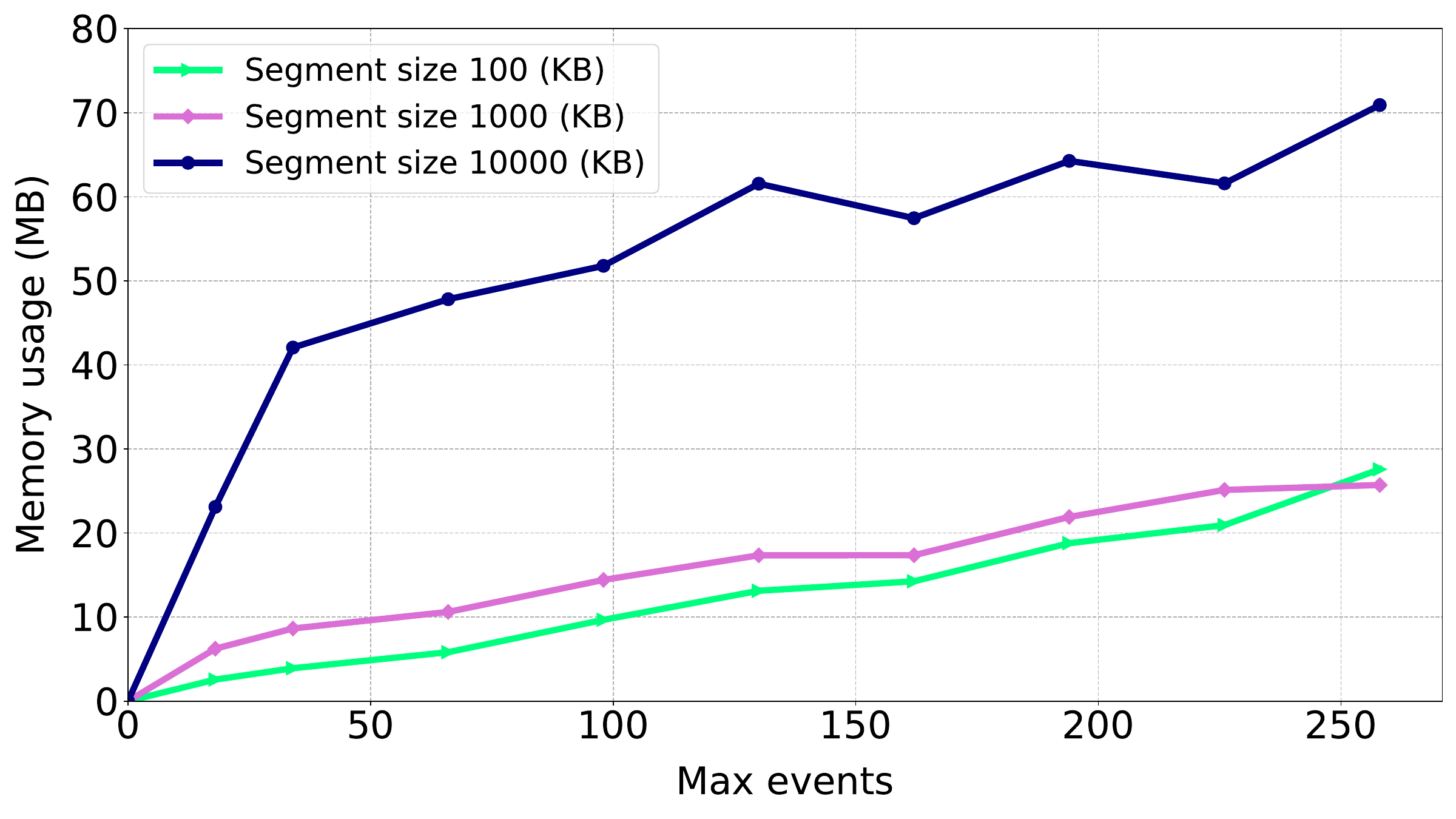}
		\caption{Results of scalability test \ref{test:events}}
		\label{fig:event_results}
	\end{subfigure}
	\begin{subfigure}{0.49\linewidth}   
		\centering
		\includegraphics[width=\linewidth]{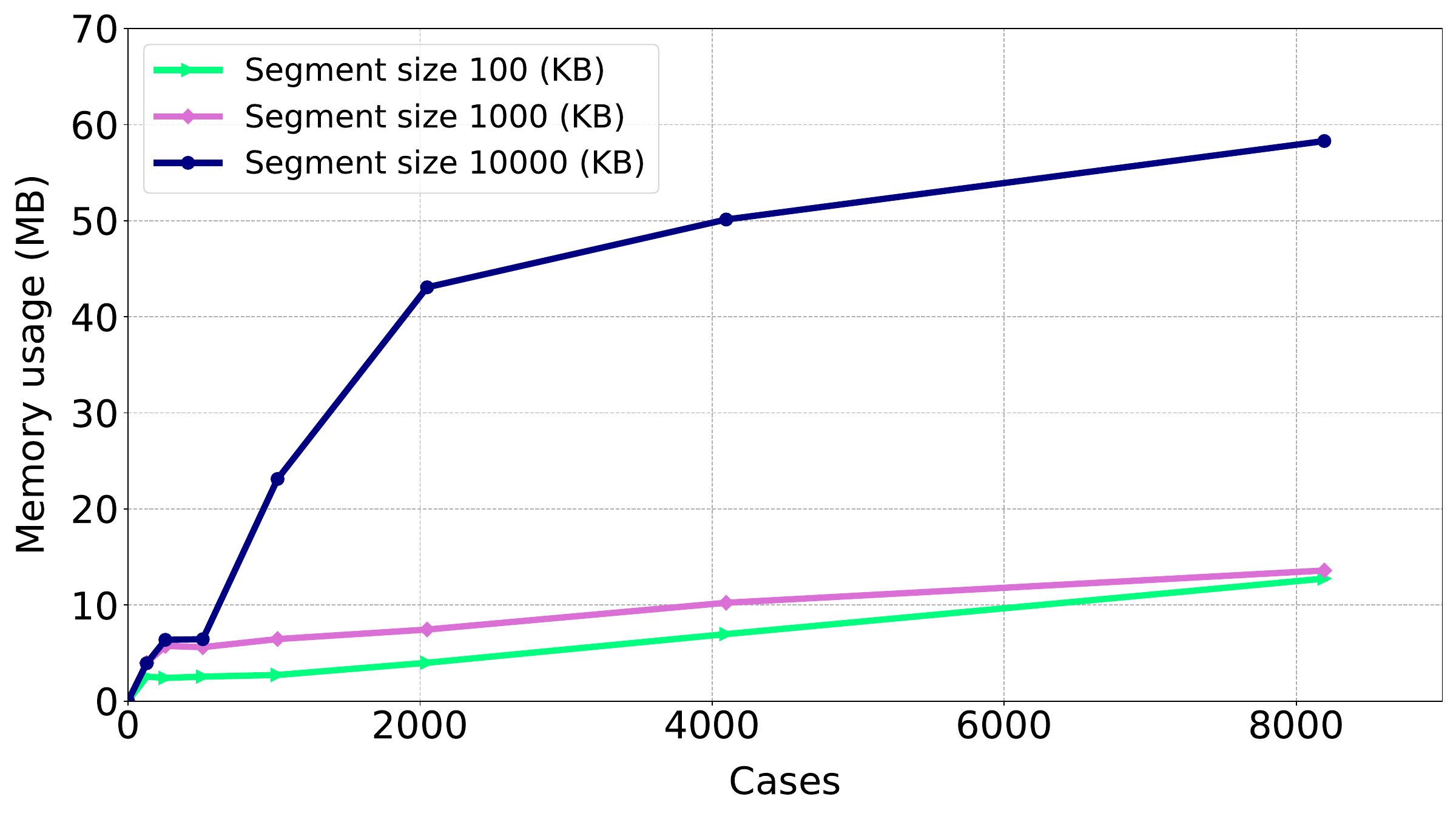}
		\caption{Results of scalability test \ref{test:cases}}
		\label{fig:cases_results}
	\end{subfigure}
	
	\begin{subfigure}{0.49\linewidth}   
		\centering
		\includegraphics[width=\linewidth]{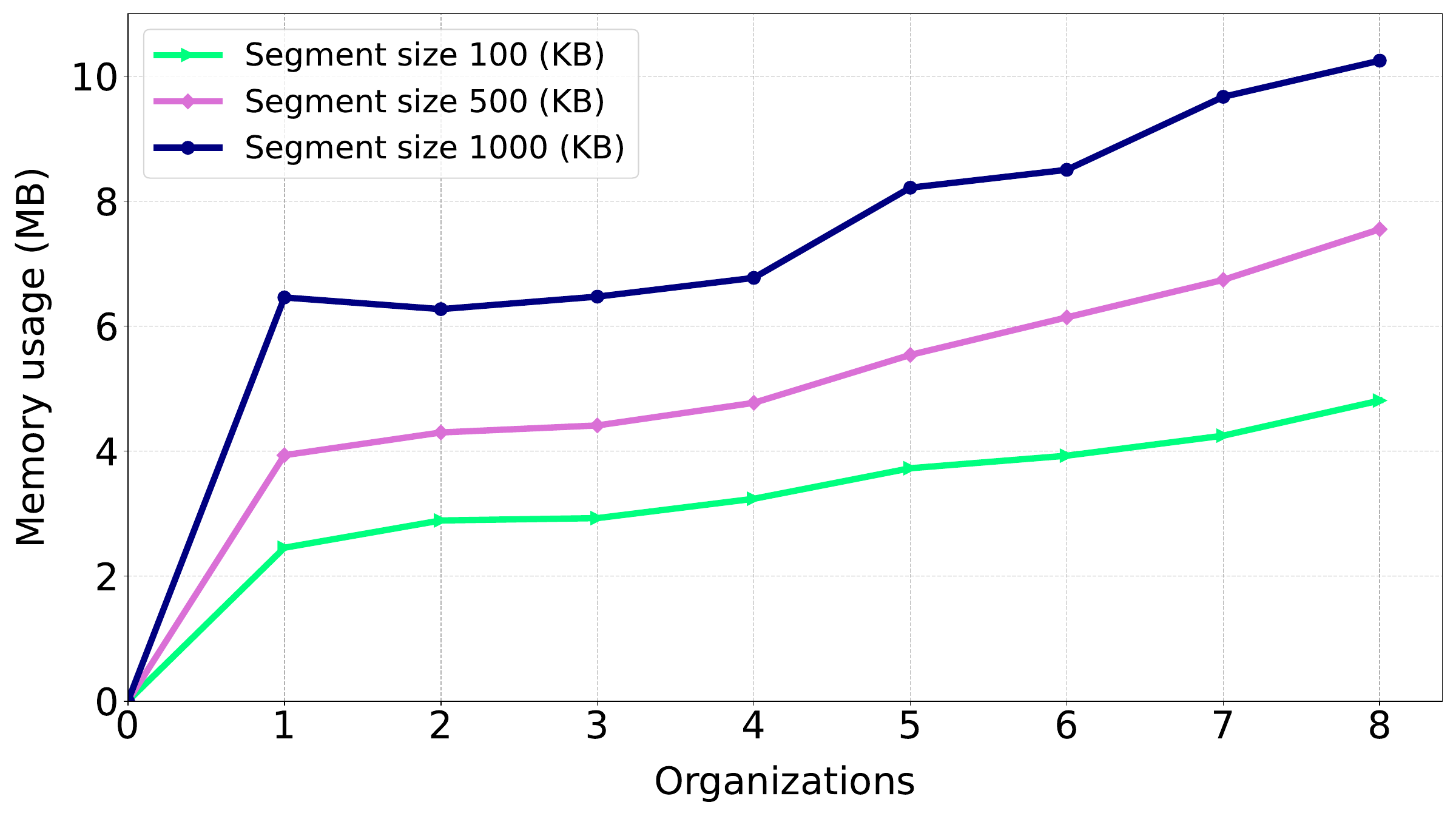}
		\caption{Results of scalability test \ref{test:organizations}}
		\label{fig:org_results}
	\end{subfigure}
	\begin{subtable}{0.49\linewidth}
		\centering
		\resizebox{0.97\linewidth}{!}{%
			\input{tables/scalability-metrics}
		}
		\label{table:TestCoefficentTable}
		\caption{Scalability measurements}
	\end{subtable}
	\caption{Scalability test results}
	\label{fig:scalabtest}
\end{figure}

\subsubsection{Scalability}\label{sec:evaluation:scalability}
In the remainder of this section, we examine the scalability of the \Compo{Secure Miner}, focusing on its capacity to efficiently manage an increasing workload in the presence of limited memory resources. We implemented three distinct test configurations gauging the average runtime memory usage as variations of our motivating scenario's log. In particular, we considered
\begin{inparaenum}[(I)]
	\item\label{test:events} the maximum number of events per case,
	\item\label{test:cases} the number of cases $|{\CIdU}|$, and 
	\item\label{test:organizations} the number of provisioning organizations $|{\OrgU}|$
\end{inparaenum}
as independent integer variables. To conduct the test on the maximum number of events, we added a loop back from the final to the initial activity of the process model, progressively increasing the number of iterations $2 \leqslant \,x_\circlearrowleft\, \leqslant 16$ at a step of \num{2}, resulting in ${18+16\cdot(x_\circlearrowleft-1)}$ events. Concerning the test on the number of cases, we simulated additional process instances so that $|{\CIdU}| = 2^{x_{\CId}}$ having $x_{\CId} \in \{7,8,\ldots,13\}$. Finally, for the assessment of the number of organizations, the test necessitated the distribution of the process model activities' into a variable number of pools, each representing a different organization ($|{\OrgU}| \in \{1,2,\ldots,8\}$).
We parameterized the above configurations with three segment sizes (in KiloBytes): $\SegSize \in \{100,1000,10000\}$ for tests \ref{test:events}~and~\ref{test:cases}, and $\SegSize \in \{100,500,1000\}$ for test~\ref{test:organizations} (the range is reduced without loss of generality to compensate the partitioning of activities into multiple organizations). To facilitate a more rigorous interpretation of the output trends across varying {\SegSize} configurations, we employ 
two well-known statistical measures. As a primary measure of goodness-of-fit, we employ the coefficient of determination {\RCoefficent}~\citep{barrett1974coefficient}, which assesses the degree to which the observed data adheres to the linear ({\Rlin}) and logarithmic ({\Rlog}) regressions derived from curve fitting approximations. To further delve into the analysis of trends with a high {\Rlin}, we consider the slope {\Slope} of the approximated linear regression~\citep{altman2015simplelinearregression}. 

\hyperref[table:TestCoefficentTable]{Tab.\ 9(d)} lists the measurements we obtained. We describe them to elucidate the observed patterns. \Cref{fig:event_results} depicts the results of test \ref{test:events}, focusing on the increase of memory utilization when the number of events in the event logs grows. We observe that the memory usage for {\SegSize} \num{100} and \num{1000} (depicted by green and lilac lines, respectively) are quite similar, whereas the setting with {\SegSize} \num{10000} (blue line) exhibits significantly higher memory usage. For the settings with {\SegSize} \num{100} and \num{1000}, {\Rlin} approaches \num{1}, 
signifying an almost perfect approximation of the linear relation, against lower {\Rlog} values. 
In these test settings, {\Slope} is very low 
yet higher than \num{0}, thus indicating that memory usage is likely to continue increasing as the number of max events grows. The configuration with {\SegSize} \num{10000} yields a higher {\Rlog} value, 
thus suggesting a logarithmic trend, hence 
a greater likelihood of stabilizing memory usage growth rate as the number of maximum events increases. 
In  \cref{fig:cases_results}, we present the results of test \ref{test:cases}, assessing the impact of the number of cases on the memory consumption. As expected, the configurations with {\SegSize} set to \num{100} and \num{1000} exhibit a trend of lower memory usage than settings with {\SegSize} \num{10000}. The {\Rlin} score of the trends with {\SegSize} \num{100} and \num{1000} 
indicate a strong linear relationship between the dependent and independent variables compared to the trend with {\SegSize} \num{10000}, which is better described by a logarithmic regression (${\Rlog} = \num{0.9303}$). For the latter, the {\Rlog} value is higher than the corresponding {\Rlin} 
thus suggesting that the logarithmic approximation is better suited to describe the trend. 
Differently from test \ref{test:events}, the {\Slope} score associated with the linear approximations of the trends with {\SegSize} \num{100} and \num{1000} 
approaches \num{0}, indicating that the growth rate of memory usage as the number of cases increases is negligible.
In \cref{fig:org_results}, we present the results of test~\ref{test:organizations}, on the relation between the number of organizations and the memory usage. The chart shows that memory usage trends increase as provisioning organizations increase for all three segment sizes. The {\Rlin} values for the three {\SegSize}s are very high, 
indicating a strong positive linear correlation. 
The test with {\SegSize} \num{100} exhibits the slowest growth rate, as corroborated by the lowest {\Slope} result (\num{0.3184}). 
For the configuration with {\SegSize} \num{500}, the memory usage increases slightly faster (${\Slope} = {0.5174}$). With {\SegSize} \num{1000}, the overall memory usage increases significantly faster than the previous configurations (${\Slope} = 0.6102$). We derive from these findings that the \Compo{Secure Miner} may encounter scalability issues when handling settings with a large number of provisioning organizations. Further investigation is warranted to determine the precise cause of this behavior and identify mitigation strategies. \

%% file: content/tables/eventLogs.tex
       \resizebox{\textwidth}{!}{%
            \begin{tabular}{l l S[table-format = 2.0] S[table-format = 4.0] S[table-format = 3.0] S[table-format = 1.0] S[table-format = 2.0] l} \toprule
            	\textbf{Name}               & \textbf{Type} & \textbf{Activities} & \textbf{Cases} & \textrm{\textbf{Max} events} & \textrm{\textbf{Min} events} & \textrm{\textbf{Avg}.\ events} & \textbf{Organization $\mapsto$ Activities}                        \\ \midrule
            	Motivating scenario  & Synthetic     & 19                  & 1000           &18 &9 &14 & ${\Org}^P \mapsto 3$, ${\Org}^C \mapsto 5$, ${\Org}^H \mapsto 14$ \\
            	Sepsis~\citep{seps}                       & Real          & 16                & 1050          & 185 & 3 & 15  & ${\Org}^1 \mapsto 1$, ${\Org}^2 \mapsto 1$, ${\Org}^3 \mapsto 14$ \\                                                                 
            	BPIC2013~\citep{bpic2013}                   & Real          & 7                  & 1487           &123 &1 &9 & ${\Org}^1 \mapsto 6$, ${\Org}^2 \mapsto 7$, ${\Org}^3 \mapsto 6$  \\ \bottomrule& & 
            \end{tabular}
        }

%% file: tables/scalability-metrics.tex
	\begin{tabular}{ll>{\hspace{1em}}cc>{\hspace{1em}}c}
	\toprule
	\textbf{Test} & \textbf{Seg.size} & 
	\textbf{\Rlin} & 
		\textbf{$\Slope$}& 
		$\Rlog$\\
		\midrule
		\multirow{3}{*}{\makecell[l]{Max\\events}} &
		100&0.9847& 0.0980&0.8291\\
		& 1000&0.9544&0.0821&0.9043\\
		& 10000&0.7357&0.1518&0.9386\\
		\midrule
		\multirow{3}{*}{\makecell[l]{Number\\of\\cases}} &
		100&0.9896&0.0013&0.6822\\
		& 1000&0.9629&0.0010&0.8682\\
		& 10000&0.7729&0.0068&0.9303\\
		\midrule
		\multirow{3}{*}{\makecell[l]{Provisioning\\organizations}} &
		100 &0.9770&0.3184&0.8577\\
		& 500 &0.9602&0.5174&0.7902\\
		& 1000 &0.9066&0.6102&0.6977\\
		\bottomrule
	\end{tabular}

%% file: content/discussion.tex
Following the above empirical examination of our approach implementation, we discuss CONFINE by addressing the implications associated with its practical instantiation (\cref{sec:discussion:practicals}) and outlining the limitations that pave the path for future work (\cref{sec:discussion:limitations}).

\subsection{Practical Implications}
\begin{table}[tb]
	\caption{Summary of the assumptions considered in this article}
	\label{tab:assumptions}
	\resizebox{\linewidth}{!}{%
		\input{tables/assumptions}%
	}%
\end{table}
\label{sec:discussion:practicals}
The applications of CONFINE span across an array of domains, including and beyond healthcare, in which our running example in~\cref{sec:motivating} is rooted. It particularly applies to scenarios in which one or more organizations are interested in process analytics outcomes based on data they hold, but cannot be disclosed to others or to the miners. In healthcare, CONFINE can aid in analyzing patient pathways on multiple healthcare providers while maintaining the secrecy of patient information, thereby enabling the identification of bottlenecks and the optimization of treatment procedures \cite{DBLP:journals/tase/LiuLZCZ23,DBLP:journals/csse/JeonKSK23}. Supply chain management companies can benefit from CONFINE by securely analyzing the flow of goods and information between multiple partners, thereby streamlining operations, reducing costs, and improving delivery timelines. In the manufacturing domain, CONFINE aids in examining branching workflows across independent organizations, optimizing resource allocation, and maintaining operational data secrecy. However, the acquisition of competitive advantage out of knowledge leakage must be prevented~\cite{DBLP:journals/isf/TanWC16}. 

Our solution enables organizations to differentiate the sources from which the \Compo{Secure Miner} retrieves the input data. This feature allows organizations to share event logs with others while locally enriching their own dataset with additional information for their internal CONFINE protocol executions ~\cite{DBLP:journals/dke/FahrenkrogPetersenAW23}. The advantage is twofold. On the one hand, this scheme lets an organization cross-link information stemming from the overarching multi-party process with internal procedure records and thus enriches the output of the \Compo{Secure Miner}.  On the other hand, information that is not relevant to the external organizations is not disclosed.

A crucial aspect for facilitating the adoption of our solution in real-world settings is the availability of accessible repositories hosting and distributing variants of the \Compo{Secure Miner}, each tailored to a specific TEE technology and equipped with different process mining algorithms (as in the the work of Soriente et al. \cite{DBLP:conf/eurosp/SorienteKLF19}). The development of trusted software for application-level TEEs requires specific Software Development Kits (SDKs), which may pose a significant integration burden for organizations lacking specialized expertise. By relying on such repositories, organizations could select the \Compo{Secure Miner} trusted application that best matches their intended process mining tasks and infrastructural constraints. At the same time, we acknowledge that this approach introduces potential risks, including malicious modifications to the \Compo{Secure Miner} source code before its deployment in a TEE. 
Relying on a trusted third-party hosting services for distributing verified versions of the \Compo{Secure Miner}, and counterchecking the associated measurements before startup through code-review practices can mitigate these risks.

    \Cref{tab:assumptions} collects all the assumptions underpinning this work, which pertains to the adversarial roles of the involved organizations (\cref{asm:honestprovisioners,asm:maliciousminers}), the event log format (\cref{asm:eventAttributes,asm:synchronizedIIDs,ams:timestamps}), the mechanisms for organizational trust and authorization (\cref{asm:iam,asm:acpolicy}), the security guarantees provided by the TEE (\cref{asm:teefeatures}), and 
    the protocol execution (\cref{asm:authperfectlink,asm:segsize}). A real-world deployment of CONFINE should ensure that those assumptions are properly addressed.

IT engineers responsible for implementing the theoretical principles of CONFINE can instantiate the framework by following the steps outlined below.  
\begin{inparaenum}[(1)]
\item They analyze the privacy requirements of the involved business entities and define their role within the framework (i.e., provisioner, miner, or both).  
\item Organizations intending to participate in the CONFINE framework as miners select a TEE technology and equip themselves with a TEE-enabled processor. Cloud-based solutions  may serve as an alternative for organizations lacking the necessary hardware.
Different organizations may adopt distinct TEE implementations, provided that each conform to the TEE profile defined in \cref{sec:background:tee}. With such a heterogeneous setting, the specific TEE technology chosen by each organization should be known to the others to ensure interoperability. This is especially relevant for attestations, since they require the correct handling of different attestation formats and verification services.
\item Process mining vendors select the algorithms to be integrated into their \Compo{Secure Miner}s and either obtain a verified distribution of the corresponding trusted application (as discussed above) or build their own implementation.
\item The organizations exchange the credentials that are necessary for executing the CONFINE protocol, such as the \Compo{Provisioner}s' web references, organizational credentials, and the \Compo{Secure Miner}s' measurements (which depend on the specific process mining algorithms integrated into each instance).
\item Finally, the involved organizations can start up their \Compo{Secure Miner} and \Compo{Provisioner} components to execute the CONFINE protocol.  
\end{inparaenum}

\subsection{Limitations and Future Work}
\label{sec:discussion:limitations}
In \cref{sec:discussion:security}, we delineated the security guarantees that the \Compo{TEE} core provides within CONFINE. However, recent studies have demonstrated that \Compo{TEE} technologies remain susceptible to three principal threat vectors:
\begin{inparaenum}[\itshape(i)\upshape]
\item side-channel attacks,
\item physical attacks,
\item root-of-trust attacks.
\end{inparaenum}
Side-channel attacks exploit secondary information leakage such as memory access patterns~\cite{DBLP:conf/eurosys/Morbitzer0HW18}, or timing variations~\cite{DBLP:conf/uss/BulckMWGKPSWYS18}, to undermine hardware-backed memory encryption and infer confidential data.
Physical attacks arise when adversaries obtain direct access to the TEE's machine and manipulate hardware-level parameters, such as inducing voltage faults, to cause the processor to skip critical security instructions~\cite{DBLP:conf/sp/MurdockOGBGP20,DBLP:conf/ccs/QiuWLQ19}.
Root-of-trust attacks stem from vendor-specific vulnerabilities that enable adversaries to compromise attestation keys at the production stage to forge attestation evidence or recover chip-endorsement keys~\cite{DBLP:conf/ccs/BuhrenWS19}.
%

%
In future work, we plan to revisit the assumptions listed in \cref{tab:assumptions}. First, we aim to incorporate mitigation strategies to address scenarios in which provisioners behave maliciously and inject manipulated event logs. We also intend to extend our protocol to operate over unreliable communication channels.

Our merge operations currently count on the availability of consistent instance identifier attributes to be used as merge keys. This presupposes prior knowledge of how each participant encodes its case identifiers. Moreover, we hypothesize the existence of a universal clock to ensure timestamp synchronization across log partitions. To overcome these limitations, future work will explore techniques for establishing or inferring a shared schema for case identifiers as well as synchronization-agnostic strategies for merging temporally misaligned log partitions.

In addition, we plan to provide a more explicit specification of the mechanisms governing organizational trust and authorization. We will investigate the integration of federated IAM solutions for cross-organization identity management and the adoption of attribute-based access control models to formalize authorization policies grounded in contextual process data~\cite{DBLP:conf/post/CramptonM12}.
Also, we plan to integrate usage control policies that regulate how event logs are processed by the \Compo{Secure Miner}. Specifically, building on prior blockchain-based approaches~\cite{DBLP:journals/fbloc/BasileCGK23,DBLP:conf/icdcs/BasileCGK23}, we plan to extend CONFINE with enforcement and monitoring mechanisms that represent usage policies as smart contracts, thus providing verifiable governance and auditable lifecycle management. In parallel, recent work on the dynamic generation of trusted applications~\cite{Goretti.etal/ICSOC2025:ProMiSe} offers a complementary avenue for embedding policy-specific enforcement logic within TEEs at runtime. Integrating these two directions within CONFINE would enable organizations to dynamically generate \Compo{Secure Miner} services tailored to their usage policies while continuously monitoring policy compliance through a blockchain-based infrastructure.
A threat to secrecy lies in the reconstruction of original data back from the mining outcome. Keeping this aspect in mind is crucial to determine the mining algorithm to be embedded in the \Compo{Secure Miner}. Studies in this regard have been conducted, among others, in~\cite{DBLP:conf/caise/VoigtFJKTMLW20,DBLP:journals/is/FahrenkrogPetersenKAW23}. Integrating the proposed recommendations with CONFINE paves the path for future investigations.

%% file: tables/assumptions.tex
\footnotesize
\begin{tabular}{lp{14cm}}
	\toprule
	\textbf{Reference} & \textbf{Description} \\
	\midrule
	\cref{asm:honestprovisioners} & Provisioning organizations play the role of the \emph{honest} party\\
        \cref{asm:maliciousminers} & Miner organizations play the role of the \emph{malicious} party\\
    \cref{asm:eventAttributes} & Events are comprised of an \emph{iid}, an \emph{activity} label and a \emph{timestamp}\\
       \cref{asm:synchronizedIIDs} & The iids are synchronized across the involved organizations\\
       \cref{ams:timestamps} & The timestamps across organizations are temporally consistent\\
       \cref{asm:iam} & Organizations can verify one another’s identities through IAM services\\
       \cref{asm:teefeatures} & The TEE provides application-level protection, hardware-based memory encryption, and remote attestability\\
       \cref{asm:authperfectlink} & The communication channel underlying CONFINE ensures reliable delivery, no-duplication, and authenticity\\
       \cref{asm:acpolicy} & \Compo{Provisioner}s retain authorization policies to decide whether requests from \Compo{Secure Miner}s are permitted\\
       \cref{asm:segsize} & The segment size range spans from the largest case size among the merged cases to the TEE's memory capacity\\
	\bottomrule
       
\end{tabular}

%% file: content/conclusion.tex
This article presented CONFINE, a framework to safeguard data secrecy in process mining. Unlike traditional approaches that either modify the input data or impose intermediate log representations to protect sensitive information, CONFINE leverages TEEs to securely process unaltered event logs beyond organizational boundaries. 
Compared to our previous work~\cite{DBLP:conf/caise/GorettiBBC24,DBLP:conf/icpm/Goretti0BC24}, the research work that led to this article contributes several key advancements. We conduct a formal analysis of the CONFINE framework, including the pseudocode for the CONFINE protocol.
We analyze the security guarantees that TEEs offer in the context of our solution. We extend the framework’s applicability by incorporating a conformance checking algorithm among the supported process mining techniques. Finally, we enhance our evaluation with both qualitative insights and quantitative assessments. The new results demonstrate the formal correctness of our approach and provide empirical evidence for its memory feasibility with different process mining tasks. 

Future research will focus on addressing many of the assumptions outlined in \cref{tab:assumptions}, and beyond. Among other endeavors, we plan to investigate strategies for managing temporal misalignment in distributed settings, handling heterogeneity in the identification of cases, mitigating the effects of maliciously altered event logs, and designing dedicated authentication policies for integration into collaborative environments. Furthermore, future work will focus on the governance of information exchange, usage policy enforcement for data treatment, and cross-organization identity management.

%% file: content/ack.tex
\bigskip
\noindent\textbf{Data availability.} The source code of our implementation and the experimental results are available at \href{https://github.com/Process-in-Chains/CONFINE}{\nolinkurl{github.com/Process-in-Chains/CONFINE}}.

\bigskip
\noindent\textbf{Acknowledgments.} This work was partially funded by the Italian Ministry of University and Research under PRIN grant B87G22000450001 (PINPOINT), by Sapienza University of Rome under grant RG123188B3F7414A (ASGARD), and by the EU-NGEU under the Cyber~4.0 NRRP MIMIT programme (Health-e-Data).

%% file: main.bbl
\begin{thebibliography}{74}
\expandafter\ifx\csname natexlab\endcsname\relax\def\natexlab#1{#1}\fi
\providecommand{\url}[1]{\texttt{#1}}
\providecommand{\href}[2]{#2}
\providecommand{\path}[1]{#1}
\providecommand{\DOIprefix}{doi:}
\providecommand{\ArXivprefix}{arXiv:}
\providecommand{\URLprefix}{URL: }
\providecommand{\Pubmedprefix}{pmid:}
\providecommand{\doi}[1]{\href{http://dx.doi.org/#1}{\path{#1}}}
\providecommand{\Pubmed}[1]{\href{pmid:#1}{\path{#1}}}
\providecommand{\bibinfo}[2]{#2}
\ifx\xfnm\relax \def\xfnm[#1]{\unskip,\space#1}\fi
\bibitem[{Dumas et~al.(2018)Dumas, La~Rosa, Mendling, and Reijers}]{Dumas.etal/2018:FundamentalsofBPM}
\bibinfo{author}{M.~Dumas}, \bibinfo{author}{M.~La~Rosa}, \bibinfo{author}{J.~Mendling}, \bibinfo{author}{H.~A. Reijers}, \bibinfo{title}{Fundamentals of Business Process Management, Second Edition}, \bibinfo{publisher}{Springer}, \bibinfo{year}{2018}. \DOIprefix\doi{10.1007/978-3-662-56509-4}.
\bibitem[{van~der Aalst(2012)}]{DBLP:journals/tmis/Aalst12}
\bibinfo{author}{W.~M.~P. van~der Aalst},
\newblock \bibinfo{title}{Process mining: Overview and opportunities},
\newblock \bibinfo{journal}{{ACM} Transactions on Management Information Systems} \bibinfo{volume}{3} (\bibinfo{year}{2012}) \bibinfo{pages}{7:1--7:17}. \DOIprefix\doi{10.1145/2229156.2229157}.
\bibitem[{{De Weerdt} and Wynn(2022)}]{DBLP:books/sp/22/WeerdtW22}
\bibinfo{author}{J.~{De Weerdt}}, \bibinfo{author}{M.~T. Wynn},
\newblock \bibinfo{title}{Foundations of process event data},
\newblock in:  \cite{PMH2022}, \bibinfo{year}{2022}, pp. \bibinfo{pages}{193--211}. \DOIprefix\doi{10.1007/978-3-031-08848-3_6}.
\bibitem[{van~der Aalst(2011)}]{van2011intra}
\bibinfo{author}{W.~M.~P. van~der Aalst},
\newblock \bibinfo{title}{Intra- and inter-organizational process mining: Discovering processes within and between organizations},
\newblock in: \bibinfo{booktitle}{The Practice of Enterprise Modeling ({IFIP}) 2011}, volume~\bibinfo{volume}{92} of \textit{\bibinfo{series}{Lecture Notes in Business Information Processing}}, \bibinfo{publisher}{Springer}, \bibinfo{year}{2011}, pp. \bibinfo{pages}{1--11}. \DOIprefix\doi{10.1007/978-3-642-24849-8_1}.
\bibitem[{Oberdorf et~al.(2023)Oberdorf, Schaschek, Weinzierl, Stein, Matzner, and Flath}]{DBLP:journals/bise/OberdorfSWSMF23}
\bibinfo{author}{F.~Oberdorf}, \bibinfo{author}{M.~Schaschek}, \bibinfo{author}{S.~Weinzierl}, \bibinfo{author}{N.~Stein}, \bibinfo{author}{M.~Matzner}, \bibinfo{author}{C.~M. Flath},
\newblock \bibinfo{title}{Predictive end-to-end enterprise process network monitoring},
\newblock \bibinfo{journal}{Business Information Systems Engineering} \bibinfo{volume}{65} (\bibinfo{year}{2023}) \bibinfo{pages}{49--64}. \DOIprefix\doi{10.1007/S12599-022-00778-4}.
\bibitem[{Morales{-}Sandoval et~al.(2021)Morales{-}Sandoval, Molina, Mar{\'{\i}}n{-}Castro, and Compe{\'{a}}n}]{DBLP:journals/peerj-cs/Morales-Sandoval21}
\bibinfo{author}{M.~Morales{-}Sandoval}, \bibinfo{author}{J.~A. Molina}, \bibinfo{author}{H.~M. Mar{\'{\i}}n{-}Castro}, \bibinfo{author}{J.~L.~G. Compe{\'{a}}n},
\newblock \bibinfo{title}{Blockchain support for execution, monitoring and discovery of inter-organizational business processes},
\newblock \bibinfo{journal}{PeerJ Computer Science} \bibinfo{volume}{7} (\bibinfo{year}{2021}) \bibinfo{pages}{e731}. \DOIprefix\doi{10.7717/PEERJ-CS.731}.
\bibitem[{Ho et~al.(2020)Ho, Kumar, and Shiwakoti}]{ho2020supply}
\bibinfo{author}{T.~Ho}, \bibinfo{author}{A.~Kumar}, \bibinfo{author}{N.~Shiwakoti},
\newblock \bibinfo{title}{Supply chain collaboration and performance: an empirical study of maturity model},
\newblock \bibinfo{journal}{SN Applied Sciences} \bibinfo{volume}{2} (\bibinfo{year}{2020}) \bibinfo{pages}{726}. \DOIprefix\doi{10.1007/s42452-020-2468-y}.
\bibitem[{Tan et~al.(2016)Tan, Wong, and Chung}]{DBLP:journals/isf/TanWC16}
\bibinfo{author}{K.~H. Tan}, \bibinfo{author}{W.~P. Wong}, \bibinfo{author}{L.~Chung},
\newblock \bibinfo{title}{Information and knowledge leakage in supply chain},
\newblock \bibinfo{journal}{Information Systems Frontiers} \bibinfo{volume}{18} (\bibinfo{year}{2016}) \bibinfo{pages}{621--638}. \DOIprefix\doi{10.1007/s10796-015-9553-6}.
\bibitem[{Fdhila et~al.(2022)Fdhila, Knuplesch, Rinderle{-}Ma, and Reichert}]{DBLP:journals/is/FdhilaKRR22}
\bibinfo{author}{W.~Fdhila}, \bibinfo{author}{D.~Knuplesch}, \bibinfo{author}{S.~Rinderle{-}Ma}, \bibinfo{author}{M.~Reichert},
\newblock \bibinfo{title}{Verifying compliance in process choreographies: Foundations, algorithms, and implementation},
\newblock \bibinfo{journal}{Information Systems} \bibinfo{volume}{108} (\bibinfo{year}{2022}) \bibinfo{pages}{101983}. \DOIprefix\doi{10.1016/J.IS.2022.101983}.
\bibitem[{Liu et~al.(2023)Liu, Li, Zhang, Cheng, and Zeng}]{DBLP:journals/tase/LiuLZCZ23}
\bibinfo{author}{C.~Liu}, \bibinfo{author}{H.~Li}, \bibinfo{author}{S.~Zhang}, \bibinfo{author}{L.~Cheng}, \bibinfo{author}{Q.~Zeng},
\newblock \bibinfo{title}{Cross-department collaborative healthcare process model discovery from event logs},
\newblock \bibinfo{journal}{{IEEE} Transactions on Automation Science and Engineering} \bibinfo{volume}{20} (\bibinfo{year}{2023}) \bibinfo{pages}{2115--2125}. \DOIprefix\doi{10.1109/TASE.2022.3194312}.
\bibitem[{Jeon et~al.(2023)Jeon, Kim, Shin, and Kim}]{DBLP:journals/csse/JeonKSK23}
\bibinfo{author}{J.~Jeon}, \bibinfo{author}{J.~Kim}, \bibinfo{author}{M.~Shin}, \bibinfo{author}{M.~Kim},
\newblock \bibinfo{title}{A blockchain-based trust model for supporting collaborative healthcare data management},
\newblock \bibinfo{journal}{Computer Systems Science and Engineering} \bibinfo{volume}{46} (\bibinfo{year}{2023}) \bibinfo{pages}{3403--3421}. \DOIprefix\doi{10.32604/CSSE.2023.036658}.
\bibitem[{Liu et~al.(2009)Liu, Li, and Zhao}]{Liu.etal/ISF2009:ChallengesOpportunitiesCollaborativeBPM}
\bibinfo{author}{C.~Liu}, \bibinfo{author}{Q.~Li}, \bibinfo{author}{X.~Zhao},
\newblock \bibinfo{title}{Challenges and opportunities in collaborative business process management: Overview of recent advances and introduction to the special issue},
\newblock \bibinfo{journal}{Information Systems Frontiers} \bibinfo{volume}{11} (\bibinfo{year}{2009}) \bibinfo{pages}{201--209}. \DOIprefix\doi{10.1007/S10796-008-9089-0}.
\bibitem[{M{\"{u}}ller et~al.(2021)M{\"{u}}ller, Ostern, Koljada, Grunert, Rosemann, and K{\"{u}}pper}]{muller2021trust}
\bibinfo{author}{M.~M{\"{u}}ller}, \bibinfo{author}{N.~Ostern}, \bibinfo{author}{D.~Koljada}, \bibinfo{author}{K.~Grunert}, \bibinfo{author}{M.~Rosemann}, \bibinfo{author}{A.~K{\"{u}}pper},
\newblock \bibinfo{title}{Trust mining: analyzing trust in collaborative business processes},
\newblock \bibinfo{journal}{IEEE Access} \bibinfo{volume}{9} (\bibinfo{year}{2021}) \bibinfo{pages}{65044--65065}. \DOIprefix\doi{10.1109/ACCESS.2021.3075568}.
\bibitem[{Fahrenkrog{-}Petersen et~al.(2023)Fahrenkrog{-}Petersen, van~der Aa, and Weidlich}]{DBLP:journals/dke/FahrenkrogPetersenAW23}
\bibinfo{author}{S.~A. Fahrenkrog{-}Petersen}, \bibinfo{author}{H.~van~der Aa}, \bibinfo{author}{M.~Weidlich},
\newblock \bibinfo{title}{Optimal event log sanitization for privacy-preserving process mining},
\newblock \bibinfo{journal}{Data Knowledge Engineering} \bibinfo{volume}{145} (\bibinfo{year}{2023}) \bibinfo{pages}{102175}. \DOIprefix\doi{10.1016/J.DATAK.2023.102175}.
\bibitem[{Fahrenkrog{-}Petersen et~al.(2019)Fahrenkrog{-}Petersen, van~der Aa, and Weidlich}]{DBLP:conf/icpm/Fahrenkrog-Petersen19}
\bibinfo{author}{S.~A. Fahrenkrog{-}Petersen}, \bibinfo{author}{H.~van~der Aa}, \bibinfo{author}{M.~Weidlich},
\newblock \bibinfo{title}{{PRETSA:} event log sanitization for privacy-aware process discovery},
\newblock in: \bibinfo{booktitle}{International Conference on Process Mining ({ICPM}) 2019}, \bibinfo{publisher}{{IEEE}}, \bibinfo{year}{2019}, pp. \bibinfo{pages}{1--8}. \DOIprefix\doi{10.1109/ICPM.2019.00012}.
\bibitem[{Elkoumy et~al.(2020{\natexlab{a}})Elkoumy, Fahrenkrog{-}Petersen, Dumas, Laud, Pankova, and Weidlich}]{elkoumy2020shareprom}
\bibinfo{author}{G.~Elkoumy}, \bibinfo{author}{S.~A. Fahrenkrog{-}Petersen}, \bibinfo{author}{M.~Dumas}, \bibinfo{author}{P.~Laud}, \bibinfo{author}{A.~Pankova}, \bibinfo{author}{M.~Weidlich},
\newblock \bibinfo{title}{Shareprom: {A} tool for privacy-preserving inter-organizational process mining},
\newblock in: \bibinfo{booktitle}{Proceedings of the Best Dissertation Award, Doctoral Consortium, and Demonstration {\&} Resources Track at {BPM} 2020}, \bibinfo{publisher}{CEUR-WS.org}, \bibinfo{year}{2020}{\natexlab{a}}, pp. \bibinfo{pages}{72--76}. \URLprefix \url{https://ceur-ws.org/Vol-2673/paperDR02.pdf}.
\bibitem[{Elkoumy et~al.(2020{\natexlab{b}})Elkoumy, Fahrenkrog{-}Petersen, Dumas, Laud, Pankova, and Weidlich}]{elkoumy2020secure}
\bibinfo{author}{G.~Elkoumy}, \bibinfo{author}{S.~A. Fahrenkrog{-}Petersen}, \bibinfo{author}{M.~Dumas}, \bibinfo{author}{P.~Laud}, \bibinfo{author}{A.~Pankova}, \bibinfo{author}{M.~Weidlich},
\newblock \bibinfo{title}{Secure multi-party computation for inter-organizational process mining},
\newblock in: \bibinfo{booktitle}{Enterprise, Business-Process and Information Systems Modeling {EMMSAD} 2020}, volume \bibinfo{volume}{387} of \textit{\bibinfo{series}{Lecture Notes in Business Information Processing}}, \bibinfo{publisher}{Springer}, \bibinfo{year}{2020}{\natexlab{b}}, pp. \bibinfo{pages}{166--181}. \DOIprefix\doi{10.1007/978-3-030-49418-6_11}.
\bibitem[{Mansour et~al.(2016)Mansour, Sambra, Hawke, Zereba, Capadisli, Ghanem, Aboulnaga, and Berners{-}Lee}]{DBLP:conf/www/MansourSHZCGAB16}
\bibinfo{author}{E.~Mansour}, \bibinfo{author}{A.~V. Sambra}, \bibinfo{author}{S.~Hawke}, \bibinfo{author}{M.~Zereba}, \bibinfo{author}{S.~Capadisli}, \bibinfo{author}{A.~Ghanem}, \bibinfo{author}{A.~Aboulnaga}, \bibinfo{author}{T.~Berners{-}Lee},
\newblock \bibinfo{title}{A demonstration of the solid platform for social web applications},
\newblock in: \bibinfo{booktitle}{International Conference on World Wide Web ({WWW}) 2016}, \bibinfo{publisher}{{ACM}}, \bibinfo{year}{2016}, pp. \bibinfo{pages}{223--226}. \DOIprefix\doi{10.1145/2872518.2890529}.
\bibitem[{Basile et~al.(2023)Basile, {Di Ciccio}, Goretti, and Kirrane}]{DBLP:conf/icdcs/BasileCGK23}
\bibinfo{author}{D.~Basile}, \bibinfo{author}{C.~{Di Ciccio}}, \bibinfo{author}{V.~Goretti}, \bibinfo{author}{S.~Kirrane},
\newblock \bibinfo{title}{A blockchain-driven architecture for usage control in solid},
\newblock in: \bibinfo{booktitle}{International Conference on Distributed Systems ({ICDCS}) 2023 - Workshops}, \bibinfo{publisher}{{IEEE}}, \bibinfo{year}{2023}, pp. \bibinfo{pages}{19--24}. \DOIprefix\doi{10.1109/ICDCSW60045.2023.00009}.
\bibitem[{Sabt et~al.(2015)Sabt, Achemlal, and Bouabdallah}]{DBLP:conf/trustcom/SabtAB15}
\bibinfo{author}{M.~Sabt}, \bibinfo{author}{M.~Achemlal}, \bibinfo{author}{A.~Bouabdallah},
\newblock \bibinfo{title}{Trusted execution environment: What it is, and what it is not},
\newblock in: \bibinfo{booktitle}{International Conference on Trust, Security and Privacy in Computing and Communications {TrustCom} 2015}, \bibinfo{publisher}{{IEEE}}, \bibinfo{year}{2015}, pp. \bibinfo{pages}{57--64}. \DOIprefix\doi{10.1109/TRUSTCOM.2015.357}.
\bibitem[{Goretti et~al.(2024{\natexlab{a}})Goretti, Basile, Barbaro, and {Di Ciccio}}]{DBLP:conf/caise/GorettiBBC24}
\bibinfo{author}{V.~Goretti}, \bibinfo{author}{D.~Basile}, \bibinfo{author}{L.~Barbaro}, \bibinfo{author}{C.~{Di Ciccio}},
\newblock \bibinfo{title}{Trusted execution environment for decentralized process mining},
\newblock in: \bibinfo{booktitle}{CAiSE 2024}, \bibinfo{publisher}{Springer}, \bibinfo{year}{2024}{\natexlab{a}}, pp. \bibinfo{pages}{509--527}. \DOIprefix\doi{10.1007/978-3-031-61057-8_30}.
\bibitem[{Goretti et~al.(2024{\natexlab{b}})Goretti, Basile, Barbaro, and {Di Ciccio}}]{DBLP:conf/icpm/Goretti0BC24}
\bibinfo{author}{V.~Goretti}, \bibinfo{author}{D.~Basile}, \bibinfo{author}{L.~Barbaro}, \bibinfo{author}{C.~{Di Ciccio}},
\newblock \bibinfo{title}{{CONFINE:} preserving data secrecy in decentralized process mining},
\newblock in: \bibinfo{booktitle}{Doctoral Consortium and Demo Track {ICPM} 2024}, volume \bibinfo{volume}{3783}, \bibinfo{year}{2024}{\natexlab{b}}. \URLprefix \url{https://ceur-ws.org/Vol-3783/paper_324.pdf}.
\bibitem[{van~der Aalst et~al.(2011)van~der Aalst, Adriansyah, de~Medeiros, Arcieri, Baier, Blickle, Bose, van~den Brand, Brandtjen, Buijs et~al.}]{van2012process}
\bibinfo{author}{W.~M.~P. van~der Aalst}, \bibinfo{author}{A.~Adriansyah}, \bibinfo{author}{A.~K.~A. de~Medeiros}, \bibinfo{author}{F.~Arcieri}, \bibinfo{author}{T.~Baier}, \bibinfo{author}{T.~Blickle}, \bibinfo{author}{R.~P. J.~C. Bose}, \bibinfo{author}{P.~van~den Brand}, \bibinfo{author}{R.~Brandtjen}, \bibinfo{author}{J.~C. A.~M. Buijs}, et~al.,
\newblock \bibinfo{title}{Process mining manifesto},
\newblock in: \bibinfo{editor}{F.~Daniel}, \bibinfo{editor}{K.~Barkaoui}, \bibinfo{editor}{S.~Dustdar} (Eds.), \bibinfo{booktitle}{International Conference on Business Process Management Workshops ({BPM}) 2011}, volume~\bibinfo{volume}{99} of \textit{\bibinfo{series}{Lecture Notes in Business Information Processing}}, \bibinfo{publisher}{Springer}, \bibinfo{year}{2011}, pp. \bibinfo{pages}{169--194}. \DOIprefix\doi{10.1007/978-3-642-28108-2_19}.
\bibitem[{Weijters et~al.(2006)Weijters, van~der Aalst, and Alves De~Medeiros}]{weijters2006process}
\bibinfo{author}{A.~J. M.~M. Weijters}, \bibinfo{author}{W.~M.~P. van~der Aalst}, \bibinfo{author}{A.~K. Alves De~Medeiros}, \bibinfo{title}{Process mining with the {H}euristics{M}iner algorithm}, \bibinfo{year}{2006}. \URLprefix \url{https://pure.tue.nl/ws/portalfiles/portal/2388011/615595.pdf}.
\bibitem[{van~der Aalst et~al.(2003)van~der Aalst, ter Hofstede, Kiepuszewski, and Barros}]{Aalst.etal/2003:WorkflowPatterns}
\bibinfo{author}{W.~M.~P. van~der Aalst}, \bibinfo{author}{A.~H.~M. ter Hofstede}, \bibinfo{author}{B.~Kiepuszewski}, \bibinfo{author}{A.~P. Barros},
\newblock \bibinfo{title}{Workflow patterns},
\newblock \bibinfo{journal}{Distributed and Parallel Databases} \bibinfo{volume}{14} (\bibinfo{year}{2003}) \bibinfo{pages}{5--51}. \URLprefix \url{http://dx.doi.org/10.1023/A:1022883727209}.
\bibitem[{van~der Aalst(2016)}]{vanderAalst2016}
\bibinfo{author}{W.~M.~P. van~der Aalst}, \bibinfo{title}{Conformance Checking}, \bibinfo{publisher}{Springer Berlin Heidelberg}, \bibinfo{address}{Berlin, Heidelberg}, \bibinfo{year}{2016}, pp. \bibinfo{pages}{243--274}. \DOIprefix\doi{10.1007/978-3-662-49851-4_8}.
\bibitem[{Donadello et~al.(2022)Donadello, Riva, Maggi, and Shikhizada}]{DBLP:conf/bpm/DonadelloRMS22}
\bibinfo{author}{I.~Donadello}, \bibinfo{author}{F.~Riva}, \bibinfo{author}{F.~M. Maggi}, \bibinfo{author}{A.~Shikhizada},
\newblock \bibinfo{title}{Declare4py: {A} python library for declarative process mining},
\newblock in: \bibinfo{booktitle}{Best Dissertation Award, Doctoral Consortium, and Demonstration {\&} Resources Trackco-located with International Conference on Business Process Management ({BPM}) 2022}, volume \bibinfo{volume}{3216} of \textit{\bibinfo{series}{{CEUR} Workshop Proceedings}}, \bibinfo{publisher}{CEUR-WS.org}, \bibinfo{year}{2022}, pp. \bibinfo{pages}{117--121}.
\bibitem[{Yasmin et~al.(2018)Yasmin, Bukhsh, and de~Alencar~Silva}]{DBLP:conf/simpda/YasminBS18}
\bibinfo{author}{F.~A. Yasmin}, \bibinfo{author}{F.~A. Bukhsh}, \bibinfo{author}{P.~de~Alencar~Silva},
\newblock \bibinfo{title}{Process enhancement in process mining: {A} literature review},
\newblock in: \bibinfo{editor}{P.~Ceravolo}, \bibinfo{editor}{M.~T. G{\'{o}}mez{-}L{\'{o}}pez}, \bibinfo{editor}{M.~van Keulen} (Eds.), \bibinfo{booktitle}{International Symposium on Data-driven Process Discovery and Analysis ({SIMPDA}) 2018}, volume \bibinfo{volume}{2270} of \textit{\bibinfo{series}{{CEUR} Workshop Proceedings}}, \bibinfo{publisher}{CEUR-WS.org}, \bibinfo{year}{2018}, pp. \bibinfo{pages}{65--72}. \URLprefix \url{https://ceur-ws.org/Vol-2270/short4.pdf}.
\bibitem[{Jauernig et~al.(2020)Jauernig, Sadeghi, and Stapf}]{DBLP:journals/ieeesp/JauernigSS20}
\bibinfo{author}{P.~Jauernig}, \bibinfo{author}{A.-R. Sadeghi}, \bibinfo{author}{E.~Stapf},
\newblock \bibinfo{title}{Trusted execution environments: Properties, applications, and challenges},
\newblock \bibinfo{journal}{{IEEE} Security and Privacy} \bibinfo{volume}{18} (\bibinfo{year}{2020}) \bibinfo{pages}{56--60}. \DOIprefix\doi{10.1109/MSEC.2019.2947124}.
\bibitem[{Antonino et~al.(2021)Antonino, Woloszyn, and Roscoe}]{antonino2021guardian}
\bibinfo{author}{P.~Antonino}, \bibinfo{author}{W.~A. Woloszyn}, \bibinfo{author}{A.~W. Roscoe},
\newblock \bibinfo{title}{Guardian: Symbolic validation of orderliness in {SGX} enclaves},
\newblock in: \bibinfo{booktitle}{Cloud Computing Security Workshop (CCSW@CCS) 2021}, \bibinfo{publisher}{{ACM}}, \bibinfo{year}{2021}, pp. \bibinfo{pages}{111--123}. \DOIprefix\doi{10.1145/3474123.3486755}.
\bibitem[{Mei et~al.(2024)Mei, Xia, Wang, and Lin}]{mei2024cabin}
\bibinfo{author}{B.~Mei}, \bibinfo{author}{S.~Xia}, \bibinfo{author}{W.~Wang}, \bibinfo{author}{D.~Lin},
\newblock \bibinfo{title}{Cabin: Confining untrusted programs within confidential vms},
\newblock in: \bibinfo{booktitle}{Information and Communications Security ({ICICS}) 2024}, volume \bibinfo{volume}{15056} of \textit{\bibinfo{series}{Lecture Notes in Computer Science}}, \bibinfo{publisher}{Springer}, \bibinfo{year}{2024}, pp. \bibinfo{pages}{165--184}. \DOIprefix\doi{10.1007/978-981-97-8798-2_9}.
\bibitem[{Costan and Devadas(2016)}]{costan2016intel}
\bibinfo{author}{V.~Costan}, \bibinfo{author}{S.~Devadas},
\newblock \bibinfo{title}{Intel {SGX} explained},
\newblock \bibinfo{journal}{{IACR} Cryptology ePrint Archive}  (\bibinfo{year}{2016}) \bibinfo{pages}{86}. \URLprefix \url{https://eprint.iacr.org/2016/086.pdf}.
\bibitem[{Birkholz et~al.(2023)Birkholz, Thaler, Richardson, Smith, and Pan}]{birkholz2023rfc}
\bibinfo{author}{H.~Birkholz}, \bibinfo{author}{D.~Thaler}, \bibinfo{author}{M.~Richardson}, \bibinfo{author}{N.~Smith}, \bibinfo{author}{W.~Pan},
\newblock \bibinfo{title}{Remote attestation procedures {(RATS)} architecture},
\newblock \bibinfo{journal}{{RFC}} \bibinfo{volume}{9334} (\bibinfo{year}{2023}) \bibinfo{pages}{1--46}. \DOIprefix\doi{10.17487/RFC9334}.
\bibitem[{Helble et~al.(2021)Helble, Kretz, Loscocco, Ramsdell, Rowe, and Alexander}]{helble2021flexible}
\bibinfo{author}{S.~C. Helble}, \bibinfo{author}{I.~D. Kretz}, \bibinfo{author}{P.~A. Loscocco}, \bibinfo{author}{J.~D. Ramsdell}, \bibinfo{author}{P.~D. Rowe}, \bibinfo{author}{P.~Alexander},
\newblock \bibinfo{title}{Flexible mechanisms for remote attestation},
\newblock \bibinfo{journal}{{ACM} Transactions on privacy and security} \bibinfo{volume}{24} (\bibinfo{year}{2021}) \bibinfo{pages}{29:1--29:23}. \DOIprefix\doi{10.1145/3470535}.
\bibitem[{M{\'{e}}n{\'{e}}trey et~al.(2022)M{\'{e}}n{\'{e}}trey, G{\"{o}}ttel, Khurshid, Pasin, Felber, Schiavoni, and Raza}]{DBLP:conf/dais/MenetreyGKPFSR22}
\bibinfo{author}{J.~M{\'{e}}n{\'{e}}trey}, \bibinfo{author}{C.~G{\"{o}}ttel}, \bibinfo{author}{A.~Khurshid}, \bibinfo{author}{M.~Pasin}, \bibinfo{author}{P.~Felber}, \bibinfo{author}{V.~Schiavoni}, \bibinfo{author}{S.~Raza},
\newblock \bibinfo{title}{Attestation mechanisms for trusted execution environments demystified},
\newblock in: \bibinfo{booktitle}{Distributed Applications and Interoperable Systems ({DAIS}) 2022}, volume \bibinfo{volume}{13272}, \bibinfo{year}{2022}, pp. \bibinfo{pages}{95--113}. \DOIprefix\doi{10.1007/978-3-031-16092-9_7}.
\bibitem[{Batista et~al.(2022)Batista, Mart{\'{\i}}nez{-}Ballest{\'{e}}, and Solanas}]{DBLP:journals/istr/BatistaMS22}
\bibinfo{author}{E.~Batista}, \bibinfo{author}{A.~Mart{\'{\i}}nez{-}Ballest{\'{e}}}, \bibinfo{author}{A.~Solanas},
\newblock \bibinfo{title}{Privacy-preserving process mining: {A} microaggregation-based approach},
\newblock \bibinfo{journal}{Journal of Information Security and Application} \bibinfo{volume}{68} (\bibinfo{year}{2022}) \bibinfo{pages}{103235}. \DOIprefix\doi{10.1016/J.JISA.2022.103235}.
\bibitem[{van~der Aalst(2021)}]{FederatedPM2021}
\bibinfo{author}{W.~M.~P. van~der Aalst},
\newblock \bibinfo{title}{Federated process mining: Exploiting event data across organizational boundaries},
\newblock in: \bibinfo{booktitle}{{IEEE} International Conference on Smart Data Services ({SMDS}) 2021}, \bibinfo{publisher}{{IEEE}}, \bibinfo{year}{2021}, pp. \bibinfo{pages}{1--7}. \DOIprefix\doi{10.1109/SMDS53860.2021.00011}.
\bibitem[{Cramer et~al.(2015)Cramer, Damg{\aa}rd, and Nielsen}]{SMPC2015}
\bibinfo{author}{R.~Cramer}, \bibinfo{author}{I.~Damg{\aa}rd}, \bibinfo{author}{J.~B. Nielsen}, \bibinfo{title}{Secure Multiparty Computation and Secret Sharing}, \bibinfo{publisher}{Cambridge University Press}, \bibinfo{year}{2015}. \DOIprefix\doi{doi.org/10.1017/CBO9781107337756}.
\bibitem[{M{\"{u}}ller et~al.(2021)M{\"{u}}ller, Simonet{-}Boulogne, Sengupta, and Beige}]{muller2021process}
\bibinfo{author}{M.~M{\"{u}}ller}, \bibinfo{author}{A.~Simonet{-}Boulogne}, \bibinfo{author}{S.~Sengupta}, \bibinfo{author}{O.~Beige},
\newblock \bibinfo{title}{Process mining in trusted execution environments: Towards hardware guarantees for trust-aware inter-organizational process analysis},
\newblock in: \bibinfo{booktitle}{International Conference on Process Mining {ICPM} 2021 Workshops,}, volume \bibinfo{volume}{433} of \textit{\bibinfo{series}{Lecture Notes in Business Information Processing}}, \bibinfo{publisher}{Springer}, \bibinfo{year}{2021}, pp. \bibinfo{pages}{369--381}. \DOIprefix\doi{10.1007/978-3-030-98581-3_27}.
\bibitem[{Yuhala et~al.(2025)Yuhala, G{\"{o}}ttel, M{\'{e}}n{\'{e}}trey, Schiavoni, Kozhaya, and Felber}]{DBLP:conf/ecrts/YuhalaGMSKF25}
\bibinfo{author}{P.~Yuhala}, \bibinfo{author}{C.~G{\"{o}}ttel}, \bibinfo{author}{J.~M{\'{e}}n{\'{e}}trey}, \bibinfo{author}{V.~Schiavoni}, \bibinfo{author}{D.~Kozhaya}, \bibinfo{author}{P.~Felber},
\newblock \bibinfo{title}{On real-time guarantees in intel {SGX} and {TDX}},
\newblock in: \bibinfo{booktitle}{Euromicro Conference on Real-Time Systems {ECRTS} 2025}, volume \bibinfo{volume}{335} of \textit{\bibinfo{series}{LIPIcs}}, \bibinfo{year}{2025}, pp. \bibinfo{pages}{8:1--8:25}. \DOIprefix\doi{10.4230/LIPICS.ECRTS.2025.8}.
\bibitem[{Sun et~al.(2021)Sun, Wang, Li, and Li}]{DBLP:journals/pvldb/SunWL021}
\bibinfo{author}{Y.~Sun}, \bibinfo{author}{S.~Wang}, \bibinfo{author}{H.~Li}, \bibinfo{author}{F.~Li},
\newblock \bibinfo{title}{Building enclave-native storage engines for practical encrypted databases},
\newblock \bibinfo{journal}{Proceedings of the {VLDB} Endowment} \bibinfo{volume}{14} (\bibinfo{year}{2021}) \bibinfo{pages}{1019--1032}. \DOIprefix\doi{10.14778/3447689.3447705}.
\bibitem[{Brandenburger et~al.(2017)Brandenburger, Cachin, Lorenz, and Kapitza}]{DBLP:conf/dsn/BrandenburgerCL17}
\bibinfo{author}{M.~Brandenburger}, \bibinfo{author}{C.~Cachin}, \bibinfo{author}{M.~Lorenz}, \bibinfo{author}{R.~Kapitza},
\newblock \bibinfo{title}{Rollback and forking detection for trusted execution environments using lightweight collective memory},
\newblock in: \bibinfo{booktitle}{International Conference on Dependable Systems and Networks ({DSN}) 2017}, \bibinfo{publisher}{{IEEE} Computer Society}, \bibinfo{year}{2017}, pp. \bibinfo{pages}{157--168}. \DOIprefix\doi{10.1109/DSN.2017.45}.
\bibitem[{Aguilar et~al.(2008)Aguilar, Garc{\'{\i}}a, Ruiz, Piattini, Calahorra, Garc{\'{\i}}a, and Martin}]{DBLP:conf/biostec/AguilarCRPCGM08}
\bibinfo{author}{E.~R. Aguilar}, \bibinfo{author}{F.~Garc{\'{\i}}a}, \bibinfo{author}{F.~Ruiz}, \bibinfo{author}{M.~Piattini}, \bibinfo{author}{L.~Calahorra}, \bibinfo{author}{M.~Garc{\'{\i}}a}, \bibinfo{author}{R.~Martin},
\newblock \bibinfo{title}{Process modeling of the health sector using {BPMN:} {A} case study},
\newblock in: \bibinfo{booktitle}{International Conference on Health Informatics (HEALTHINF)}, \bibinfo{publisher}{{INSTICC} - Institute for Systems and Technologies of Information, Control and Communication}, \bibinfo{year}{2008}, pp. \bibinfo{pages}{173--178}. \DOIprefix\doi{10.5220/0001042201730178}.
\bibitem[{Sampson et~al.(2015)Sampson, Schmidt, Gardner, and Van~Orden}]{sampson2015process}
\bibinfo{author}{S.~E. Sampson}, \bibinfo{author}{G.~Schmidt}, \bibinfo{author}{J.~W. Gardner}, \bibinfo{author}{J.~Van~Orden},
\newblock \bibinfo{title}{Process coordination within a health care service supply network},
\newblock \bibinfo{journal}{Journal of Business Logistics} \bibinfo{volume}{36} (\bibinfo{year}{2015}) \bibinfo{pages}{355--373}. \DOIprefix\doi{10.1111/jbl.12106}.
\bibitem[{Jans and Hosseinpour(2019)}]{Jans.Hosseinpour/IJAIS2019:ActiveLearningProcessMiningForAuditing}
\bibinfo{author}{M.~Jans}, \bibinfo{author}{M.~Hosseinpour},
\newblock \bibinfo{title}{How active learning and process mining can act as continuous auditing catalyst},
\newblock \bibinfo{journal}{Int. J. Accounting Inf. Systems} \bibinfo{volume}{32} (\bibinfo{year}{2019}) \bibinfo{pages}{44--58}. \DOIprefix\doi{10.1016/j.accinf.2018.11.002}.
\bibitem[{Mendelson(2015)}]{Mendelson/2015:IntroductionMathematicalLogic}
\bibinfo{author}{E.~Mendelson}, \bibinfo{title}{Introduction to Mathematical Logic}, \bibinfo{publisher}{Chapman and Hall}, \bibinfo{year}{2015}.
\bibitem[{van~der Aalst(2016)}]{Aalst/2016:ProcessMiningBook:DataScienceinAction}
\bibinfo{author}{W.~M.~P. van~der Aalst}, \bibinfo{title}{Process Mining - Data Science in Action, Second Edition}, \bibinfo{publisher}{Springer}, \bibinfo{year}{2016}. \DOIprefix\doi{10.1007/978-3-662-49851-4}.
\bibitem[{Koch and Kraus(2002)}]{koch2002expressive}
\bibinfo{author}{N.~Koch}, \bibinfo{author}{A.~Kraus},
\newblock \bibinfo{title}{The expressive power of {UML}-based web engineering},
\newblock in: \bibinfo{booktitle}{Workshop on Web-Oriented Software Technology (IWWOST) 2002}, volume~\bibinfo{volume}{16}, \bibinfo{organization}{Citeseer}, \bibinfo{year}{2002}, pp. \bibinfo{pages}{40--41}. \DOIprefix\doi{doi.org/10.1007/3-540-36208-8_5}.
\bibitem[{Cachin et~al.(2011)Cachin, Guerraoui, and Rodrigues}]{Cachin.etal/2011:ReliableSecureDistributedProgramming}
\bibinfo{author}{C.~Cachin}, \bibinfo{author}{R.~Guerraoui}, \bibinfo{author}{L.~E.~T. Rodrigues}, \bibinfo{title}{Introduction to Reliable and Secure Distributed Programming {(2.} ed.)}, \bibinfo{publisher}{Springer}, \bibinfo{year}{2011}. \DOIprefix\doi{10.1007/978-3-642-15260-3}.
\bibitem[{Crampton and Morisset(2012)}]{DBLP:conf/post/CramptonM12}
\bibinfo{author}{J.~Crampton}, \bibinfo{author}{C.~Morisset},
\newblock \bibinfo{title}{Ptacl: {A} language for attribute-based access control in open systems},
\newblock in: \bibinfo{booktitle}{Principles of Security and Trust ({POST}) 2012}, volume \bibinfo{volume}{7215} of \textit{\bibinfo{series}{Lecture Notes in Computer Science}}, \bibinfo{publisher}{Springer}, \bibinfo{year}{2012}, pp. \bibinfo{pages}{390--409}. \DOIprefix\doi{10.1007/978-3-642-28641-4_21}.
\bibitem[{Claes and Poels(2014)}]{claes2014merging}
\bibinfo{author}{J.~Claes}, \bibinfo{author}{G.~Poels},
\newblock \bibinfo{title}{Merging event logs for process mining: {A} rule based merging method and rule suggestion algorithm},
\newblock \bibinfo{journal}{Expert Systems with Applications} \bibinfo{volume}{41} (\bibinfo{year}{2014}) \bibinfo{pages}{7291--7306}. \DOIprefix\doi{10.1016/J.ESWA.2014.06.012}.
\bibitem[{Birkholz et~al.(2023)Birkholz, Thaler, Richardson, Smith, and Pan}]{rfc9334}
\bibinfo{author}{H.~Birkholz}, \bibinfo{author}{D.~Thaler}, \bibinfo{author}{M.~Richardson}, \bibinfo{author}{N.~Smith}, \bibinfo{author}{W.~Pan},
\newblock \bibinfo{title}{Remote attestation procedures {(RATS)} architecture},
\newblock volume \bibinfo{volume}{9334}, \bibinfo{year}{2023}, pp. \bibinfo{pages}{1--46}. \DOIprefix\doi{10.17487/RFC9334}.
\bibitem[{Lamport(1977)}]{DBLP:journals/tse/Lamport77}
\bibinfo{author}{L.~Lamport},
\newblock \bibinfo{title}{Proving the correctness of multiprocess programs},
\newblock \bibinfo{journal}{{IEEE} Transactions on Software Engineering} \bibinfo{volume}{3} (\bibinfo{year}{1977}) \bibinfo{pages}{125--143}. \DOIprefix\doi{10.1109/TSE.1977.229904}.
\bibitem[{Mannhardt(2016)}]{seps}
\bibinfo{author}{F.~Mannhardt}, \bibinfo{title}{Sepsis cases - event log}, \bibinfo{year}{2016}. \DOIprefix\doi{10.4121/UUID:915D2BFB-7E84-49AD-A286-DC35F063A460}.
\bibitem[{Steeman(2013)}]{bpic2013}
\bibinfo{author}{W.~Steeman}, \bibinfo{title}{{BPI} challenge 2013, incidents}, \bibinfo{year}{2013}. \DOIprefix\doi{10.4121/UUID:500573E6-ACCC-4B0C-9576-AA5468B10CEE}.
\bibitem[{Bagher and Lai(2023)}]{DBLP:journals/istr/BagherL23}
\bibinfo{author}{K.~Bagher}, \bibinfo{author}{S.~Lai},
\newblock \bibinfo{title}{Sgx-stream: {A} secure stream analytics framework in sgx-enabled edge cloud},
\newblock \bibinfo{journal}{Journal of Information Security and Application} \bibinfo{volume}{72} (\bibinfo{year}{2023}) \bibinfo{pages}{103403}. \DOIprefix\doi{10.1016/J.JISA.2022.103403}.
\bibitem[{Murdock et~al.(2020)Murdock, Oswald, Garcia et~al.}]{DBLP:conf/sp/MurdockOGBGP20}
\bibinfo{author}{K.~Murdock}, \bibinfo{author}{D.~F. Oswald}, \bibinfo{author}{F.~D. Garcia}, et~al.,
\newblock \bibinfo{title}{Plundervolt: Software-based fault injection attacks against intel {SGX}},
\newblock in: \bibinfo{booktitle}{{IEEE} Symposium on Security and Privacy ({SP}) 2020}, \bibinfo{publisher}{{IEEE}}, \bibinfo{year}{2020}, pp. \bibinfo{pages}{1466--1482}. \DOIprefix\doi{10.1109/SP40000.2020.00057}.
\bibitem[{Bulck et~al.(2018)Bulck, Minkin, Weisse, Genkin, Kasikci et~al.}]{DBLP:conf/uss/BulckMWGKPSWYS18}
\bibinfo{author}{J.~V. Bulck}, \bibinfo{author}{M.~Minkin}, \bibinfo{author}{O.~Weisse}, \bibinfo{author}{D.~Genkin}, \bibinfo{author}{B.~Kasikci}, et~al.,
\newblock \bibinfo{title}{Foreshadow: Extracting the keys to the intel {SGX} kingdom with transient out-of-order execution},
\newblock in: \bibinfo{booktitle}{{USENIX} 2018}, \bibinfo{publisher}{{USENIX} Association}, \bibinfo{year}{2018}. \URLprefix \url{https://www.usenix.org/conference/usenixsecurity18/presentation/bulck}.
\bibitem[{Fei et~al.(2022)Fei, Yan, Ding, and Xie}]{DBLP:journals/csur/FeiYDX21}
\bibinfo{author}{S.~Fei}, \bibinfo{author}{Z.~Yan}, \bibinfo{author}{W.~Ding}, \bibinfo{author}{H.~Xie},
\newblock \bibinfo{title}{Security vulnerabilities of {SGX} and countermeasures: {A} survey},
\newblock \bibinfo{journal}{{ACM} Computing Surveys} \bibinfo{volume}{54} (\bibinfo{year}{2022}) \bibinfo{pages}{126:1--126:36}. \DOIprefix\doi{10.1145/3456631}.
\bibitem[{Thomas(2000)}]{Thomas/2000:SSL-TLS}
\bibinfo{author}{S.~A. Thomas}, \bibinfo{title}{{SSL} and {TLS} Essentials: Securing the Web}, \bibinfo{publisher}{Wiley}, \bibinfo{year}{2000}. \URLprefix \url{https://dl.acm.org/doi/10.5555/519018}.
\bibitem[{Pichler et~al.(2011)Pichler, Weber, Zugal, Pinggera, Mendling, and Reijers}]{DBLP:conf/bpm/PichlerWZPMR11}
\bibinfo{author}{P.~Pichler}, \bibinfo{author}{B.~Weber}, \bibinfo{author}{S.~Zugal}, \bibinfo{author}{J.~Pinggera}, \bibinfo{author}{J.~Mendling}, \bibinfo{author}{H.~A. Reijers},
\newblock \bibinfo{title}{Imperative versus declarative process modeling languages: An empirical investigation},
\newblock in: \bibinfo{booktitle}{Business Process Management Workshops ({BPM}) 2011}, volume~\bibinfo{volume}{99} of \textit{\bibinfo{series}{Lecture Notes in Business Information Processing}}, \bibinfo{publisher}{Springer}, \bibinfo{year}{2011}, pp. \bibinfo{pages}{383--394}. \DOIprefix\doi{10.1007/978-3-642-28108-2_37}.
\bibitem[{Di~Ciccio and Montali(2022)}]{DiCiccio.Montali/PMH2022:DeclarativeProcessMining}
\bibinfo{author}{C.~Di~Ciccio}, \bibinfo{author}{M.~Montali},
\newblock \bibinfo{title}{Declarative process specifications: Reasoning, discovery, monitoring},
\newblock in:  \cite{PMH2022}, \bibinfo{year}{2022}, pp. \bibinfo{pages}{108--152}. \DOIprefix\doi{10.1007/978-3-031-08848-3_4}.
\bibitem[{van~der Aalst(1997)}]{Aalst/ICATPN97:VerificationofWfNs}
\bibinfo{author}{W.~M.~P. van~der Aalst},
\newblock \bibinfo{title}{Verification of workflow nets},
\newblock in: \bibinfo{booktitle}{International Conference on Applications and Theory of Petri Nets and Concurrency ({ICATPN}) 1997}, \bibinfo{year}{1997}, pp. \bibinfo{pages}{407--426}. \DOIprefix\doi{10.1007/3-540-63139-9_48}.
\bibitem[{Barrett(1974)}]{barrett1974coefficient}
\bibinfo{author}{J.~P. Barrett},
\newblock \bibinfo{title}{The coefficient of determination—some limitations},
\newblock \bibinfo{journal}{The American Statistician}  (\bibinfo{year}{1974}) \bibinfo{pages}{19--20}. \DOIprefix\doi{10.1080/00031305.1974.10479056}.
\bibitem[{Su et~al.(2012)Su, Yan, and Tsai}]{altman2015simplelinearregression}
\bibinfo{author}{X.~Su}, \bibinfo{author}{X.~Yan}, \bibinfo{author}{C.-L. Tsai},
\newblock \bibinfo{title}{Linear regression},
\newblock \bibinfo{journal}{Wiley Interdisciplinary Reviews: Computational Statistics} \bibinfo{volume}{4} (\bibinfo{year}{2012}) \bibinfo{pages}{275--294}. \DOIprefix\doi{doi.org/10.1002/wics.1198}.
\bibitem[{Soriente et~al.(2019)Soriente, Karame, Li, and Fedorov}]{DBLP:conf/eurosp/SorienteKLF19}
\bibinfo{author}{C.~Soriente}, \bibinfo{author}{G.~Karame}, \bibinfo{author}{W.~Li}, \bibinfo{author}{S.~Fedorov},
\newblock \bibinfo{title}{Replicatee: Enabling seamless replication of {SGX} enclaves in the cloud},
\newblock in: \bibinfo{booktitle}{{IEEE} European Symposium on Security and Privacy (EuroS{\&}P) 2019}, \bibinfo{publisher}{{IEEE}}, \bibinfo{year}{2019}, pp. \bibinfo{pages}{158--171}. \DOIprefix\doi{10.1109/EUROSP.2019.00021}.
\bibitem[{Morbitzer et~al.(2018)Morbitzer, Huber, Horsch, and Wessel}]{DBLP:conf/eurosys/Morbitzer0HW18}
\bibinfo{author}{M.~Morbitzer}, \bibinfo{author}{M.~Huber}, \bibinfo{author}{J.~Horsch}, \bibinfo{author}{S.~Wessel},
\newblock \bibinfo{title}{Severed: Subverting amd's virtual machine encryption},
\newblock in: \bibinfo{booktitle}{European Workshop on Systems Security (EuroSec@EuroSys) 2018}, \bibinfo{publisher}{{ACM}}, \bibinfo{year}{2018}, pp. \bibinfo{pages}{1:1--1:6}. \DOIprefix\doi{10.1145/3193111.3193112}.
\bibitem[{Qiu et~al.(2019)Qiu, Wang, Lyu, and Qu}]{DBLP:conf/ccs/QiuWLQ19}
\bibinfo{author}{P.~Qiu}, \bibinfo{author}{D.~Wang}, \bibinfo{author}{Y.~Lyu}, \bibinfo{author}{G.~Qu},
\newblock \bibinfo{title}{Voltjockey: Breaching trustzone by software-controlled voltage manipulation over multi-core frequencies},
\newblock in: \bibinfo{booktitle}{Conference on Computer and Communications Security ({CCS}) 2019}, \bibinfo{publisher}{{ACM}}, \bibinfo{year}{2019}, pp. \bibinfo{pages}{195--209}. \DOIprefix\doi{10.1145/3319535.3354201}.
\bibitem[{Buhren et~al.(2019)Buhren, Werling, and Seifert}]{DBLP:conf/ccs/BuhrenWS19}
\bibinfo{author}{R.~Buhren}, \bibinfo{author}{C.~Werling}, \bibinfo{author}{J.~Seifert},
\newblock \bibinfo{title}{Insecure until proven updated: Analyzing {AMD} sev's remote attestation},
\newblock in: \bibinfo{booktitle}{Communications Security ({CCS}) 2019}, \bibinfo{publisher}{{ACM}}, \bibinfo{year}{2019}, pp. \bibinfo{pages}{1087--1099}. \DOIprefix\doi{10.1145/3319535.3354216}.
\bibitem[{Basile et~al.(2023)Basile, {Di Ciccio}, Goretti, and Kirrane}]{DBLP:journals/fbloc/BasileCGK23}
\bibinfo{author}{D.~Basile}, \bibinfo{author}{C.~{Di Ciccio}}, \bibinfo{author}{V.~Goretti}, \bibinfo{author}{S.~Kirrane},
\newblock \bibinfo{title}{Blockchain based resource governance for decentralized web environments},
\newblock \bibinfo{journal}{Frontiers Blockchain} \bibinfo{volume}{6} (\bibinfo{year}{2023}). \DOIprefix\doi{10.3389/FBLOC.2023.1141909}.
\bibitem[{Goretti et~al.(2025)Goretti, Kirrane, and Di~Ciccio}]{Goretti.etal/ICSOC2025:ProMiSe}
\bibinfo{author}{V.~Goretti}, \bibinfo{author}{S.~Kirrane}, \bibinfo{author}{C.~Di~Ciccio},
\newblock \bibinfo{title}{Usage control for process discovery through a trusted execution environment},
\newblock in: \bibinfo{booktitle}{ICSOC}, \bibinfo{year}{2025}, pp. \bibinfo{pages}{268--286}. \DOIprefix\doi{10.1007/978-981-95-5015-9_20}.
\bibitem[{von Voigt et~al.(2020)von Voigt, Fahrenkrog{-}Petersen, Janssen, Koschmider, Tschorsch, Mannhardt, Landsiedel, and Weidlich}]{DBLP:conf/caise/VoigtFJKTMLW20}
\bibinfo{author}{S.~N. von Voigt}, \bibinfo{author}{S.~A. Fahrenkrog{-}Petersen}, \bibinfo{author}{D.~Janssen}, \bibinfo{author}{A.~Koschmider}, \bibinfo{author}{F.~Tschorsch}, \bibinfo{author}{F.~Mannhardt}, \bibinfo{author}{O.~Landsiedel}, \bibinfo{author}{M.~Weidlich},
\newblock \bibinfo{title}{Quantifying the re-identification risk of event logs for process mining - empiricial evaluation paper},
\newblock in: \bibinfo{booktitle}{Advanced Information Systems Engineering ({CAiSE}) 2020}, volume \bibinfo{volume}{12127} of \textit{\bibinfo{series}{Lecture Notes in Computer Science}}, \bibinfo{publisher}{Springer}, \bibinfo{year}{2020}, pp. \bibinfo{pages}{252--267}. \DOIprefix\doi{10.1007/978-3-030-49435-3_16}.
\bibitem[{Fahrenkrog{-}Petersen et~al.(2023)Fahrenkrog{-}Petersen, Kabierski, van~der Aa, and Weidlich}]{DBLP:journals/is/FahrenkrogPetersenKAW23}
\bibinfo{author}{S.~A. Fahrenkrog{-}Petersen}, \bibinfo{author}{M.~Kabierski}, \bibinfo{author}{H.~van~der Aa}, \bibinfo{author}{M.~Weidlich},
\newblock \bibinfo{title}{Semantics-aware mechanisms for control-flow anonymization in process mining},
\newblock \bibinfo{journal}{Information Systems} \bibinfo{volume}{114} (\bibinfo{year}{2023}) \bibinfo{pages}{102169}. \DOIprefix\doi{10.1016/J.IS.2023.102169}.
\bibitem[{van~der Aalst and Carmona(2022)}]{PMH2022}
\bibinfo{editor}{W.~M.~P. van~der Aalst}, \bibinfo{editor}{J.~Carmona} (Eds.), \bibinfo{title}{Process Mining Handbook}, volume \bibinfo{volume}{448} of \textit{\bibinfo{series}{Lecture Notes in Business Information Processing}}, \bibinfo{publisher}{Springer}, \bibinfo{year}{2022}.

\end{thebibliography}
